\documentclass[a4paper,11pt]{article}
\usepackage{jheppub}
\usepackage{amsmath}
\usepackage{graphicx}
\usepackage[utf8]{inputenc}
\usepackage{xcolor}
\usepackage{amsthm}
\usepackage{float}
\usepackage{comment}
\usepackage{soul}
\usepackage{hyperref} 
\usepackage{subfigure}

\newtheorem{theorem}{Theorem}[section]
\newtheorem{lemma}[theorem]{Lemma}

\newtheorem{definition}[theorem]{Definition}

\newcommand{\jp}[1]{\textcolor{red}{#1}} % Puts comments added by John in red.
\newcommand{\IK}[1]{\textcolor{blue}{#1}} % Puts comments added by Isaac in blue.
	\title{The ghost in the radiation: Robust encodings of the black hole interior}

	\author[a,c]{Isaac Kim,}
    \author[b]{Eugene Tang,}
    \author[b]{and John Preskill}
	\affiliation[a]{Stanford Institute for Theoretical Physics, Stanford University, Stanford CA 94305, USA}
    \affiliation[b]{Institute for Quantum Information and Matter and Walter Burke Institute for Theoretical Physics, California Institute of Technology, Pasadena CA 91125, USA}
    \affiliation[c]{School of Physics, The University of Sydney, Sydney, Australia}
\emailAdd{isaac.kim@sydney.edu.au}
\emailAdd{eugene.tang@caltech.edu}
\emailAdd{preskill@caltech.edu}

\abstract{
We reconsider the black hole firewall puzzle, emphasizing that quantum error-correction, computational complexity, and pseudorandomness are crucial concepts for understanding the black hole interior. We assume that the Hawking radiation emitted by an old black hole is pseudorandom, meaning that it cannot be distinguished from a perfectly thermal state by any efficient quantum computation acting on the radiation alone. We then infer the existence of a subspace of the radiation system which we interpret as an encoding of the black hole interior. This encoded interior is entangled with the late outgoing Hawking quanta emitted by the old black hole, and is inaccessible to computationally bounded observers who are outside the black hole. Specifically, efficient operations acting on the radiation, those with quantum computational complexity polynomial in the entropy of the remaining black hole, commute with a complete set of logical operators acting on the encoded interior, up to corrections which are exponentially small in the entropy. Thus, under our pseudorandomness assumption, the black hole interior is well protected from exterior observers as long as the remaining black hole is macroscopic. On the other hand, if the radiation is not pseudorandom, an exterior observer may be able to create a firewall by applying a polynomial-time quantum computation to the radiation.}

\begin{document}
\maketitle
\flushbottom

\section{Introduction}
\label{sec:intro}
The discovery that black holes emit Hawking radiation raised deep puzzles about the quantum physics of black holes \cite{Hawking1975}. What happens to quantum information that falls into a black hole, if that black hole subsequently evaporates completely and disappears? Is the information lost forever, or does it escape in the radiation emitted by the black hole, albeit in a highly scrambled form that is difficult to decode? And if the information does escape, how? The struggle to definitively answer these questions has been a major theme of quantum gravity research during the 45 years since Hawking's pivotal discovery.

The AdS/CFT holographic correspondence provides powerful evidence indicating that quantum information really does escape from an evaporating black hole \cite{Maldacena1997}. This correspondence, for which there is now substantial evidence, asserts that the process in which a black hole forms and then completely evaporates in an asymptotically anti-de Sitter bulk spacetime admits a dual description in terms of a conformally-invariant quantum field theory living on the boundary of the spacetime. In this dual description, the system evolves unitarily and therefore the process is microscopically reversible --- on the boundary there is no gravity, no black hole, no place for information to hide. Since this observation applies to evaporating black holes that are small compared to the AdS curvature scale, it seems plausible that a similar conclusion should apply to more general spacetimes which are not asymptotically AdS, even though we currently lack a firm grasp of how quantum gravity works in that more general setting. 

However, so far the holographic correspondence has not provided a satisfying picture of the \textit{mechanism} that allows the information to escape from behind the black hole's event horizon. It is not even clear how the boundary theory encodes the experience of observers who cross the event horizon and visit the black hole interior. 

That describing the inside of a black hole raises subtle issues was emphasized in 2012 by the authors known as AMPS \cite{Almheiri2013}. Following AMPS, consider a black hole $H$ that is maximally entangled with another system $E$ which is outside the black hole, and suppose that $B$ is a thermally occupied Hawking radiation mode which is close to the horizon and moving radially outward. Since the black hole is maximally entangled with $E$, the highly mixed state of $B$ must be purified by a subsystem of $E$. But on the other hand, we expect that a freely falling observer who enters the black hole will not encounter any unexpected excitations at the moment of crossing the horizon; since field modes are highly entangled in the vacuum state, this means that $B$ should be purified by a mode $\tilde{B}$ located inside the black hole. Now we have a problem, because it is not possible for the mixed state of $B$ to be purified by both $E$ and $\tilde{B}$. Something has to give! Were we to break the entanglement between $B$ and $\tilde{B}$ for the sake of preserving the entanglement between $B$ and $E$, the infalling observer would encounter a seething firewall at the horizon. This conclusion is hard to swallow, since for a macroscopic black hole we would expect semiclassical theory to be trustworthy at the event horizon, and the black hole solution to the classical Einstein equation has a smooth horizon, not a firewall.

To find a way out of this quandary, it is helpful to contemplate the thermofield double (TFD) state of two boundary conformal field theories, which we'll refer to as the left and right boundary theories. The TFD is an entangled pure state of the left and right boundaries, with the property that the marginal state of the right boundary (with the left boundary traced out) is a thermal state with temperature $T$, and likewise the marginal state of the left boundary (with the right boundary traced out) is thermal with the same temperature. The corresponding bulk geometry is a two-sided black hole. Both the left black hole and the right black hole are in equilibrium with a radiation bath at temperature $T$, and both have smooth event horizons. Furthermore, the two black holes have a shared interior --- they are connected in the bulk by a non-traversable wormhole behind the horizon~\cite{Maldacena2003}. Here, the right black hole (let's call it $H$) is purified by another system (the left black hole $E$), and emits Hawking radiation, yet it has a smooth horizon. How can we reconcile this finding with the AMPS argument?

For the case of the two-sided black hole, there is an instructive answer \cite{ER=EPR}. The Hawking mode $B$ outside the right black hole can be purified by both $\tilde{B}$ behind the horizon and by a subsystem of $E$, because $E$ itself lies behind the horizon and $\tilde{B}$ is a subsystem of $E$! It is very tempting to suggest that a similar resolution of the AMPS puzzle applies to the case of a one-sided black hole $H$, which is entangled with a system $E$ outside its horizon. That is, we may regard the black hole interior and the exterior system entangled with the black hole as two complementary descriptions of one and the same system. Indeed, we might imagine allowing $E$ to undergo gravitational collapse, thereby obtaining a pair of entangled black holes, which, if we accept a conjecture formulated in \cite{ER=EPR}, would be connected through the bulk by a non-traversable wormhole. The boundary dual of this bulk state, up to a one-sided transformation acting on one of the two boundaries, is a TFD, to which our previous discussion of the entanglement structure of the two-sided black hole ought to apply.

The idea that, for the case of a black hole $H$ purified by the exterior system $E$, we may regard the black hole interior as related to $E$ by a complicated encoding map, has been advocated, discussed, and criticized in much previous work \cite{ER=EPR,Susskind2014,Papadodimas2013,Harlow2014,Bousso2014,Papadodimas2016}. We will revisit this issue in this paper, arguing that a proper resolution of the AMPS puzzle should invoke concepts that have received relatively short shrift in earlier discussions of the firewall problem, namely \textit{quantum error correction, computational complexity, and pseudorandomness}. 

The scenario described above, in which the black hole $H$ has become maximally entangled with the exterior system $E$, might arise because the black hole actually formed long ago, and since then has radiated away more than half of its initial entropy. In that case $E$ would be the Hawking radiation so far emitted during the black hole's lifetime, most of which is by now far away from the black hole. One could object that our proposal, that the black hole interior is related to $E$ by a complicated encoding map, is too wildly non-local to be credible \cite{Almheiri2013a,Bousso2014,Harlow2014}. Why can't an exterior agent who interacts with the Hawking radiation send instantaneous signals to the black hole interior in flagrant violation of causality? And why can't such an agent access the encoded system $\tilde{B}$, breaking the entanglement between $\tilde{B}$ and $B$ and hence creating excitations which can be detected by an observer who falls through the horizon? 

Our answer is that such non-local operations are in principle possible, but are not accessible to observers whose computational abilities are bounded (a notion we make precise in Section~\ref{sec:pr_decoupling}); the operations required to disturb the interior mode are far too complex to be realizable in practice for any realistic observer. Thus, in spite of the extreme non-locality of the encoding map, violations of the semiclassical causal structure of the black hole spacetime are beyond the reach of any realistic exterior observer. This statement is most conveniently expressed using the language of quantum error correction and computational complexity. We will use $|S|$ to denote the size of a physical system $S$; by size we mean the number of qubits, so that $2^{|S|}$ is the dimension of the Hilbert space of $S$. We regard $H$ as the Hilbert space of black hole microstates, and $E$ as the Hilbert space of the previously emitted radiation. For an old black hole $H$ which is nearly maximally entangled with $E$, we show that a quantum error-correcting code can be constructed, in a subspace of $EH$, which describes the black hole interior. The logical operators of this code, which preserve the code subspace, are operators acting on the interior. We will argue that a code exists with the following property: Any operation on the radiation $E$ that can be performed as a quantum computation whose size is polynomial in $|H|$ will commute with a set of logical operators of the code, up to corrections which are exponentially small in $|H|$. For this encoding, then, an observer outside a black hole can signal the interior only by performing an operation of super-polynomial complexity. Because the encoded interior is for all practical purposes invisible to the agent who roams the radiation system $E$, we call the code's logical operators \textit{ghost operators}.\footnote{The word ``ghost'' is sometimes used to describe unphysical degrees of freedom. That is not what we mean here. The ghost operators act on a system (the interior of a black hole) which is physical but \emph{inaccessible} to observers outside the black hole who have reasonable computational power.}

To reach this conclusion, we make a nontrivial but reasonable assumption --- that the radiation system $E$ is \textit{pseudorandom}. Note that if the state of $EH$ is pure, and $|H| \ll |E|$, then the density operator $\rho_E$ of $E$ is not full rank, so that $\rho_E$ is obviously distinguishable from the maximally mixed state $\sigma_E$. When we say that $\rho_E$ is pseudorandom, we mean that $\rho_E$ and $\sigma_E$ are not \textit{computationally} distinguishable. That is, suppose we receive a copy of $\rho_E$ (or even polynomially many copies) and we are asked to determine whether the state is maximally mixed or not using a quantum circuit whose size is polynomial in $|H|$. If $\rho_E$ is pseudorandom, then our probability of answering correctly exceeds $1/2$ by an amount which is exponentially small in $|H|$. Such pseuodrandom quantum states exist, %(see Section~\ref{sec:quantum-pseudo}), 
and furthermore it has recently been shown \cite{Ji2018} that they can be prepared by efficient quantum circuits, if one accepts a standard (and widely believed) assumption of post-quantum cryptography: That there exist one-way functions which are hard to invert using a quantum computer. Since black holes are notoriously powerful scramblers of quantum information \cite{Sekino2008}, we think the assumption that $\rho_E$ is pseudorandom is plausible, though undeniably speculative. Our main technical result shows that if $\rho_E$ is pseudorandom, then a code with ghost logical operators must exist. 

That the existence of quantum-secure one-way functions implies the hardness of \textit{decoding} Hawking radiation had been pointed out earlier in \cite{Harlow2013} and \cite{Aaronson2016}. But our statement goes further --- it indicates that causality is well respected from the viewpoint of computationally bounded observers (as long as $|H|$ is large). The semi-classical causal structure of the evaporating black hole spacetime can be disrupted by an observer with sufficient computational power, but not by an observer whose actions can be faithfully modeled by a quantum circuit with size polynomial in $|H|$. On the other hand, interior observers, who in principle have access to $H$ as well $E$, could plausibly perform nontrivial operations on the interior which are beyond the reach of the computationally bounded observer who acts on $E$ alone. 

Our main result can be regarded as a contribution to the theory of quantum error correction in a nonstandard setting. In the context of fault-tolerant quantum computation, where the goal is to protect a quantum computer from noise due to uncontrolled interactions with the computer’s environment, we usually consider noise which is weak and only weakly correlated. For example, we might model the noise using a Hamiltonian describing the interactions of the computer and environment, where each term in the Hamiltonian is small and acts on only a few of the computer’s qubits. In our setting the ``computer’’ is the system $EH$, and the ``noise’’ results from the interactions of the computationally bounded observer with system $E$, while $H$ is regarded as noiseless. In contrast to conventional quantum error correction, we allow the  noise to be strong, highly correlated, and adversarially chosen, yet the logical system $\tilde{B}$ encoded in $EH$ is well protected against this noise. To obtain this result, though, it is essential that the noise acts only on $E$ and not on $H$, a departure from the usual model of fault tolerance in which all qubits are assumed to be noisy. 

We note that the encoding of the black hole interior in $EH$ is state dependent; that is, the way the system $\tilde{B}$ is embedded in $EH$ depends on the initial state that underwent gravitational collapse to form a black hole. This state dependence of the encoding has sparked much discussion and consternation \cite{Harlow2014,Marolf2016}. What seems troubling is that operators which depend on the state to which they are applied are not linear operators acting on Hilbert space, and therefore can not be regarded as observables as described in the conventional quantum theory of measurement. Our view is that the tension arising from the state dependence of the encoded operator algebra signals that we do not yet have a fully satisfactory way to describe measurements performed inside black holes. We will not rectify this shortcoming in this paper.

Our argument about the robustness of the ghost logical operators makes no direct use of AdS/CFT technology. This may be viewed as either a strength or a weakness. The strength is that our results may be applicable to black holes in spacetimes which are asymptotically flat or de Sitter, and stand independently of any assumptions of holography. The weakness is that we have not presented evidence based on holographic duality which supports our conjecture. 

There has been great recent progress toward resolving the discrepancy between Hawking's semiclassical analysis \cite{Hawking1975} and the Page curve~\cite{Hayden2019, Penington2019,Almheiri2019b} of an evaporating black hole, including formulas for the entropy of the radiation supported by explicit computations \cite{Almheiri2019a,Penington2019a,Almheiri2019}. %involving replica wormholes
These results strengthen the evidence that black hole evaporation is a unitary process, and also point toward a resolution of the firewall problem in which the interior of a partially evaporated black hole is encoded in the Hawking radiation. This beautiful prior work, however, does not directly address how the profoundly nonlocal encoding of the interior in the radiation is compatible with the semiclassical causal structure of the black hole geometry. It is for that purpose that we hope our observations concerning the pseudorandomness of the radiation and the construction of ghost logical operators acting on the interior will prove to be relevant. Our main conclusion is that the encoded interior can be inaccessible to observers outside the black hole who have reasonable computational power.  Establishing closer contact between our work and these recent computations is an important open problem.

%Moreover, these calculations suggest that the interior of a black hole is encoded in the radiation. But if that's the case, how can the observer in the exterior, who has plenty of time to interact with the majority of the radiation, cannot perturb the interior? It is precisely this question that we have intended to answer in this paper; our main contribution lies on finding a possible mechanism by which the space-time locality between the interior and the exterior of a black hole remains intact. We believe that prior studies do not necessarily resolve this issue because, if we assume that the interior of a black hole is encoded in its radiation, the exterior observer can, in principle, perturb this interior mode from the outside. If such a task is possible, depending on what the observer does outside the black hole, i.e., whether to perturb the interior mode or not, he will see something different once he dives into a black hole. This is a flagrant violation of locality. The main technical result of this paper lies on providing a definition of the interior mode that can avoid such fallacies.

%\jp{A comment for here or elsewhere: Pseudorandomness of thermal radiation is expected to be a property of more general strongly chaotic quantum systems, not just black holes. Our observation about the existence of a code applies, but its significance for other systems besides black holes is unclear. Black holes are special because they have an event horizon and a firewall problem.}
%\IK{I think that discussion may be better suited for the conclusion.}

The rest of this paper is structured as follows. In Section~\ref{sec:overview} we provide a non-technical summary of the paper. In Section~\ref{sec:classical-pseudo} and~\ref{sec:quantum-pseudo}, we review the notion of pseudorandomness in both the classical and quantum setting; 
%These sections can be skipped if the readers are already familiar with the concept. 
in Section~\ref{sec:hawking-pseudo}, we argue that the Hawking radiation is a pseudorandom quantum state, and we explain in detail our computational model of the black hole.

In the remaining sections, we derive consequences of the pseudorandomness assumption, and explore their potential relevance to the black hole firewall problem. 
In Section~\ref{sec:pr_decoupling}, we show that it is computationally hard for an observer interacting with the early radiation $E$ to distill the interior mode $\tilde B$ and carry it into the black hole. In Section~\ref{sec:bh_qecc}, we show that the encoded system $\tilde B$ is protected against errors inflicted on $E$ by any agent who performs a quantum operation with $\mathrm{poly}(|H|)$ computational complexity and sufficiently small Kraus rank. In Section~\ref{sec:qec}, we describe the construction of ghost logical operators acting on the black hole interior; these operators commute with all low-complexity operations applied to $E$ by an agent $O$, provided that $O$'s quantum memory is not too large.
%\IK{I suggest the following change:  ... which commute with all low-complexity operations that a sufficiently small agent can apply on $E$.} 
If the observable properties of the black hole interior are described by such ghost operators, we infer that the interior cannot be affected or detected by computationally bounded agents who interact with the Hawking radiation. The theory of ghost operators, which can be constructed for any approximate quantum error-correcting code, may also be of independent interest. In Section~\ref{sec:efficient-manipulation}, we show that, if the state of the partially evaporated black hole has been efficiently generated, then an agent with access to both $E$ and $H$ can manipulate the encoded interior efficiently, and efficiently distill the encoded system $\tilde B$ to a small quantum memory. Section~\ref{sec:conclusion} contains our conclusions. Some technicalities are treated in the Appendices, and in Appendix \ref{sec:no_pseudorandomness_consequences} we discuss via an example how the construction of ghost logical operators may fail if the Hawking radiation is not pseudorandom.

%we that the ``low-complexity errors'' inflicted on encoded interior show that black hole is a quantum error correcting code that can correct ``low-complexity errors." we revisit and resolve the original firewall paradox \cite{Almheiri2013}.  In particular, we explain why it is difficult for an infalling observer to distill the interior mode before jumping into the black hole. Building upon the technical results of Section~\ref{sec:pr_decoupling}, in Section~\ref{sec:bh_qecc}, we show that black hole is a quantum error correcting code that can correct ``low-complexity errors." In Section~\ref{sec:qec}, we revisit the post-firewall paradoxes discussed in \cite{Almheiri2013a,Harlow2014,Bousso2014} and resolve them. In order to resolve these paradoxes, we develop a theory of \emph{ghost logical operators}, which is a general statement about quantum error correcting code that may be of an independent interest. We conclude in Section~\ref{sec:conclusion}.

\section{Probing the radiation}\label{sec:overview}

In this section we'll provide a somewhat more explicit explanation of our main result, still skipping over technical details which will be laid out in later sections. The situation we consider is depicted in Figure~\ref{fig:setup}. There, the unitary transformation $U_{\text{bh}}$ describes the formation and subsequent partial evaporation of a black hole formed from infalling matter in a pure state $|\phi_{\mathrm{matter}}\rangle$, where $E$ denotes the ``early’’ Hawking radiation which has been emitted so far, $H$ denotes the remaining black hole which has not yet evaporated, and $B$ denotes Hawking quanta of the ``late’’ radiation which has just been emitted from the black hole. We may assume for convenience that $B$ is a single qubit --- our conclusions would be the same if we considered $B$ to be any system of constant dimension, independent of the size of $E$ and $H$. The system $P$ denotes an ancillary system called the ``probe", which might represent, for example, ambient dust around the black hole. We will discuss the role of the probe in greater detail shortly, but for simplicity we may ignore its presence right now.

\begin{figure}[ht]
    \centering
    \includegraphics[width=0.4\columnwidth]{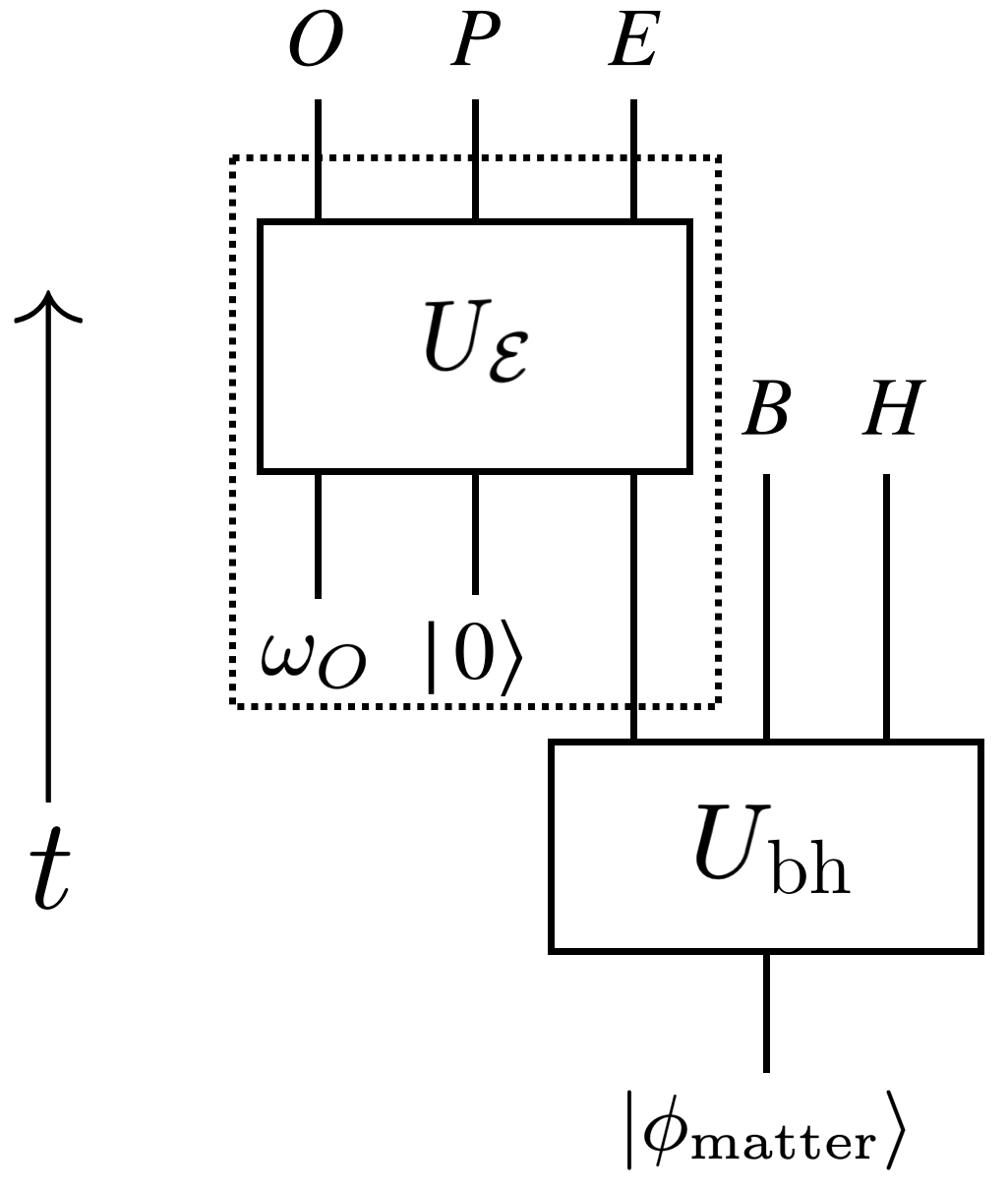}
    \caption{
    A black hole forms due to the gravitational collapse of an infalling state of matter. The black hole then evaporates for a while, emitting the ``early'' radiation $E$ and the ``late'' radiation $B$; the formation and partial evaporation of the black hole are described by the unitary transformation $U_{\text{bh}}$. An observer $O$ interacts with the early radiation and a probe system $P$, where the unitary transformation $U_{\mathcal{E}}$ (enclosed by the dotted line) has quantum complexity which scales polynomially with the size $|H|$ of the remaining black hole (essentially its entropy $S_{\text{bh}}$). If the radiation is pseudorandom, then $O$ is unable to distinguish $E$ from a perfectly thermal state.
    %Once the black hole is formed from an infalling state of matter $|\phi_{\mathrm{matter}}\rangle$, the early radiation interacts with the probe and the observer. The dotted line represents a physical process whose complexity scales at most polynomially with $S_{\text{bh}}.$ \jp{Change $U_{\text{BH}}$ to $U_{\text{bh}}$ in the figure. The notation $U_{\text{BH}}$ is confusing because one expects $BH$ to refer to systems $B$ and $H$, but it actually is supposed to mean ``black hole,'' I changed it in the text. Same goes for $S_{\text{BH}}$. Also, indicate in the figure that time flows upward.} \textcolor{blue}{(ET) I've changed this figure. Would you like time arrows on all the other figures as well? Or would this one suffice to set conventions?}
    }
    \label{fig:setup}
\end{figure}

In the case of an ``old’’ black hole $H$, which has already radiated away over half of its initial entropy and has become nearly maximally entangled with $E$, we have $|H| < |E|$. Because the lifetime of an evaporating black hole scales like the $3/2$ power of its initial system size, we may regard the unitary transformation $U_{\text{bh}}$ to be ``efficient,'' meaning that it can be accurately described by a quantum circuit whose size increases only polynomially with $|EHB|$. Our key assumption is that the efficient unitary $U_{\text{bh}}$ creates a pseudorandom state of $EB$ (see Section~\ref{sec:hawking-pseudo}). The notion of a pseudorandom quantum state will be further discussed in Section~\ref{sec:quantum-pseudo}.

In the context of the AMPS puzzle, the recently emitted system $B$ should be purified by a system $\tilde{B}$ behind the horizon. We will explore the idea that this system $\tilde{B}$ is actually encoded in $EH$, the union of the black hole system $H$ and the early radiation system $E$. 
Let us denote the state prepared by the unitary map $U_{\text{bh}}$ as $|\Psi\rangle_{EHB}$, and consider its expansion
\begin{align}
    |\Psi\rangle_{EHB} = \sum_{ijk} \Psi_{ijk} ~|i\rangle_E\otimes |j\rangle_H \otimes|k\rangle_B.
\end{align}
There is a corresponding map $V_\Psi: \tilde{B}\rightarrow EH$ defined by
\begin{align}
    V_\Psi = \sqrt{d_B} \sum_{ijk} \Psi_{ijk}|i\rangle_E \otimes|j\rangle_H\otimes\langle k|_{\tilde{B}},\label{eq:state_dependent_isometry}
\end{align}
where $d_B$ is the dimension of $B$. If $B$ is maximally mixed in the state $|\Psi\rangle$, then $V_\Psi$ is an isometric map embedding $\tilde{B}$ in $EH$. We interpret $V_\Psi$ as the encoding map of a quantum error-correcting code, which maps the interior system $\tilde{B}$ to the subspace of $EH$ with which $B$ is maximally entangled. 

If $\tilde T_{\tilde{B}}$ is any operator acting on $\tilde{B}$, there is a corresponding ``logical'' operator ${T}_{EH}$ acting on $EH$ defined by
\begin{align}\label{eq:tilde-X-define}
    V_\Psi \tilde T_{\tilde{B}} = {T}_{EH} V_\Psi.
\end{align}
This logical operator is not uniquely defined, because equation~(\ref{eq:tilde-X-define}) only specifies its action on the code space, the image of $V_\Psi$. We may say that ${T}_{EH}$ is the ``mirror operator'' of $\tilde T_{\tilde{B}}$ determined by $|\Psi\rangle$, whose defining property is that ${T}_{EH}$ and $\tilde T_{\tilde{B}}$ produce the same output when acting on the state $|\Psi\rangle$.

%\jp{JP: I tried to make the $\tilde T$ notation consistent with what we use later in the paper. }

We wish to investigate whether an agent who interacts with only the radiation system $E$ can manipulate the encoded system $\tilde{B}$. For that purpose we introduce an additional system $O$ to represent an observer outside the black hole who interacts with $E$. This interaction is modeled by a unitary transformation $U_{\mathcal{E}}$ acting on $OE$, possibly followed by a simple measurement performed on $O$; for example one might measure all the qubits of $O$ in a standard basis). The unitary transformation, but not the following measurement, is shown in Figure~\ref{fig:setup}.
%\footnote{While the circuit diagram in Figure~\ref{fig:setup} does not make the measurement explicit, the measurement can be modeled as a unitary process involving a larger system. One can view $P$ in Figure~\ref{fig:setup}, which represents the probe, as this larger system.} 
After the interaction, but before $O$ is measured, the joint state of $OEBH$ has evolved to 
\begin{align}
    |\Psi'\rangle_{OEBH} = \left(\left(U_{\mathcal{E}}\right)_{OE}\otimes I_{BH}\right) \left(|\omega\rangle_O \otimes |\Psi\rangle_{EBH}\right),
\end{align}
where $|\omega\rangle_O$ is the initial state of $O$ before $O$ and $E$ interact. 

Our notation $U_{\mathcal{E}}$ for the unitary transformation is motivated by a widely used convention in the theory of quantum channels, in which $\mathcal{E}$ denotes a quantum noisy channel (a trace-preserving completely positive map), with the letter $\mathcal{E}$ indicating an ``error'' acting on the input to the channel. A quantum channel always admits a dilation (also called a purification), a unitary transformation which acts on the input system and an ``environment,'' after which the environment is discarded. In our context, the noisy channel $\mathcal{E}$ acting on $E$ arises from the action of the observer, and we may regard the observer's system $O$ as the environment in the dilation of $U_{\mathcal{E}}$. In the following discussion, we will often omit the subscript $OE$ on $\left(U_\mathcal{E}\right)_{OE}$, leaving it implicit that $U_{\mathcal{E}}$ acts on the radiation system $E$ and observer $O$.

Now we can appeal to a standard result in the theory of quantum error correction. In $|\Psi\rangle_{EHB}$ we regard $B$ as a ``reference system'' which purifies the maximally mixed state of the encoded system $\tilde{B}$.
%and we regard $O$ as the ``environment'' of a noisy quantum channel whose dilation is the unitary transformation $U_{\mathcal{E}}$. 
Is there a recovery operator which can be applied to $EH$ to correct the error induced by this noisy channel? In fact a recovery operator that corrects the error \textit{exactly} exists if and only if the marginal state $\rho'_{OB}$ of $OB$ factorizes,
\begin{align}
    \rho'_{OB} = \rho'_O\otimes \rho'_B,
\end{align}
in which case we say the reference system $B$ ``decouples'' from the environment $O$. Heuristically, the error can be corrected if and only if no information about the state of $\tilde B$ leaks to the environment $O$. There is also an approximate version of this statement~\cite{BenyOreshkov10,Flammia2017limitsstorageof}. Roughly speaking (we will be more precise in Section~\ref{sec:pr_decoupling}), recovery with fidelity close to one is possible if and only if $O$ and $B$ are nearly uncorrelated after $O$ and $E$ interact.

Now consider the implications of our assumption that the Hawking radiation is pseuodorandom. As stated in Section~\ref{sec:intro}, the marginal state $\rho_{EB}$ is pseudorandom if $\rho_{EB}$ cannot be distinguished from a maximally mixed state by any circuit with size polynomial in $|H|$, apart from an error exponentially small in $|H|$. We will show in Section~\ref{sec:pr_decoupling} that, assuming $|O|\ll |H|$, if $\rho_{EB}$ is pseudorandom and $U_{\mathcal{E}}$ is any polynomial-size unitary transformation, then $O$ and $B$ approximately decouple up to an error exponentially small in $|H|$. Therefore, apart from an exponentially small error, a computationally bounded observer $O$ is unable to inflict an uncorrectable error on the encoded system $\tilde{B}$.

We can make a stronger assertion: It is possible to choose the logical operators acting on the encoded system to be robust \textit{ghost operators}, which (acting on the code space) nearly commute with any operation applied by the computationally bounded observer $O$. Returning now for simplicity to the setting of exact correctability, we claim that if the error induced by $U_{\mathcal{E}}$ is correctable, then for any operator $\tilde T_{\tilde{B}}$ acting on system $\tilde{B}$, it is possible to choose the corresponding logical operator ${T}_{EH}$ satisfying equation~(\ref{eq:tilde-X-define}) such that
\begin{align}
    {T}_{EH} U_{\mathcal{E}}\left(I_O\otimes  V_{\Psi}\right) =  U_{\mathcal{E}} {T}_{EH}\left(I_O\otimes V_{\Psi}\right).
\end{align}
In this sense, the correctable errors have no effect on the ghost logical algebra. This claim is a special case of a more general statement about operator algebra quantum error correction (OAQEC)~\cite{oaqec1,oaqec2}. Since we do not expect a black hole to provide an exact error-correcting code, we will need to analyze the case of approximate quantum error correction. Unfortunately, it does not seem straightforward to generalize the results of~\cite{oaqec1,oaqec2} to the approximate setting. Instead, we present a self-contained construction of exact ghost logical operators in Section~\ref{sec:exactghost}, without making direct use of known results from the theory of OAQEC, and then generalize the construction to the approximate setting in Section~\ref{sec:apxghost}.

%This claim is actually a special case of a general statement about quantum error correction, which has not been previously discussed as far as we know: it is always possible to choose logical operators of a quantum code whose action on the code space commutes with the action of all correctable errors. 

We will apply the approximate version of this result to the situation where $U_{\mathcal{E}}$ induces an approximately correctable error, thus inferring that the logical operators acting on the encoded system $\tilde{B}$ may be chosen so that they nearly commute with the actions of the computationally bounded observer $O$. We propose that these robust ghost operators are the logical operators acting on the black hole interior, and conclude that the interior is very well protected against the actions of any realistic observer who resides outside the black hole.

The statements about the indistinguishability of $\rho_{EB}$ from a maximally mixed state, the  decoupling of $O$ from $B$, the correctability of $U_{\mathcal{E}}$, and the commuting action of $\tilde{T}_{EH}$ and $U_{\mathcal{E}}$ on the code space, are all approximate relations with exponentially small corrections. Therefore, we need to be mindful of these corrections in constructing our arguments. Fortunately, many relevant features of approximate quantum error-correction have been previously studied, and we make use of results from \cite{BenyOreshkov10,Flammia2017limitsstorageof} in particular. %\jp{Does this still apply?}\IK{Maybe we don't use them heavily, but we do use them in Section 7.1.}

For the general argument sketched above we have assumed that the observer system satisfies $|O| \ll |H|$. But it is also instructive to consider a different scenario, in which the observer has access to an auxiliary probe system $P$. We now imagine that the probe $P$, which might have a size comparable to or larger than $E$, is prepared in a simple initial state and then interacts efficiently with $E$. After this interaction between $E$ and $P$, the observer (still satisfying $|O| \ll |H|$), interacts with $EP$, performing an efficient quantum computation that may be chosen adversarially. In this case, too, we can show under the same pseudorandomness assumption as before that the reference system $B$ decouples from $O$, and that robust ghost logical operators can be constructed. For example, the probe might cause all of the qubits of $E$ to dephase in a preferred basis, but the entanglement between $\tilde{B}$ and $B$ would still be protected. The modification from the previously considered case is that now $\tilde{B}$ will be encoded in $EHP$ rather than $EH$, and we conclude that the encoded black hole interior remains inaccessible to any computationally bounded observer $O$ who examines the radiation and the probe, as long as the size of the observer's memory satisfies $|O|\ll |H|$.

Our conclusion that $\tilde{B}$ is difficult to decode or manipulate follows from the pseudorandomness of the Hawking radiation if the observer is computationally bounded and has access only to the radiation system $E$ outside the black hole. But we might imagine that an observer who jumps into the black hole has access to the black hole degrees of freedom $H$ as well as $E$. We show in Section~\ref{sec:efficient-manipulation} that an observer who has access to $EH$ can efficiently manipulate and decode $\tilde B$, assuming only that the state $|\Psi\rangle_{EBH}$ was created by an efficient unitary process. In this sense, an interior observer can interact with the interior degrees of freedom, as one might expect. A similar remark applies to the fully evaporated black hole. If the final state after complete evaporation is a highly scrambled pure state of $EB$, where $|B|\ll |E|$, then the maximally mixed state of $B$ is purified by a code subspace of $E$. If $B$ has constant size, then the code state can be efficiently distilled and deposited in a small quantum memory, assuming only that the map from the infalling matter to the outgoing Hawking radiation is an efficient unitary process.

If an efficient measurement of $EB$ \emph{can} detect the correlation between $E$ and $B$, then we may expect that an observer acting on $E$ is able to interact efficiently with the black hole interior. In Appendix~\ref{sec:no_pseudorandomness_consequences}, we show that, if a product observable $M_E\otimes N_B$ has an expectation value in the state $|\Psi\rangle_{EHB}$ that differs significantly from its expectation value in a maximally mixed state of $EB$, then there cannot be a complete set of ghost logical operators on $EH$ commuting with $M_E$. It follows that, if $M_E$ can be realized efficiently, low-complexity operations acting on the Hawking radiation \emph{can} send a signal to the interior.

\section{Classical pseudorandomness}
\label{sec:classical-pseudo}
Our argument that the black hole interior is inaccessible to computationally bounded exterior observers hinges on the hypothesis that the Hawking radiation emitted by an old black hole is pseudorandom. In this section we'll provide background about the concept of pseudorandomness, which some readers might find helpful.

As discussed in Section~\ref{sec:intro}, we are interested in a black hole that is still macroscopic but has already been evaporating for longer than its Page time~\cite{Page1993}. The state of the previously emitted radiation system $EB$ is purified by the black hole system $H$, and by this time 
$EB$ is much larger than $H$; therefore the microscopic state $\rho_{EB}$ of $EB$ has far lower rank than a thermal state. It must then be possible, at least in principle, to distinguish $\rho_{EB}$ from a thermal state. But how, operationally, would an observer outside the black hole who interacts with the radiation be able to tell the difference?

To start with, it will be instructive to consider a simple classical model that captures some of the features of this setup --- after we understand how the classical model works we'll be better prepared to analyze an analogous quantum model.  Let's suppose that the emitted Hawking radiation is a classical bit string $x$ of length $n$, which our observer is permitted to read. But this bit string is not chosen deterministically; rather, when the observer reads the radiation he actually samples from a probability distribution governing $n$-bit strings. We'll say that the state of the black hole is ``thermal'' if this distribution is the uniformly random distribution $p_I(x)$, where
\begin{equation}
    p_I(x) = \frac{1}{2^n}, \qquad \forall x \in \{0,1 \}^n.
\end{equation}
But suppose the state of the black hole is described by a distribution that is in principle almost perfectly distinguishable from the uniform distribution. Can this state ``fool'' the observer, leading him to believe the distribution is uniform even though that is far from the case? See Figure~\ref{fig:imitation}.

%We can ask a similar question in a completely classical setup. Abstractly, these types of experiments can be modeled as a probe(e.g., an experimental apparatus) that interacts with the system of interest(Hawking radiation). Then an observer can read out the outcome of the experiment by observing the probe; see Figure\ref{fig:imitation}. Ultimately, we will be interested in a setup in which a quantum probe interacts with a quantum system, from which an observer extracts a quantum information. But for now, let us stick to the classical case. We will learn an interesting lesson whose ramification will carry over to the quantum setup.

%It will be useful to conceive a thought experiment that is in the spirit of the imitation game~\cite{Turing1950}. In the imitation game, a player interrogates another player to determine if the other player is a machine or a human. The other player is in a black box, so they can only communicate with each other through a note. The answer can be revealed by opening a box, but otherwise, if a machine can imitate the human's answer sufficiently accurately, the player will be unable to tell the difference. 

To be more concrete, let's suppose the observer is assured that he is sampling from a distribution which is either the uniform distribution $p_I(x)$, or a different distribution $p_S(x)$ which is uniform on the subset of $n$-bit strings $S$:
\begin{equation}
    p_S(x) =
    \begin{cases}
    2^{-\alpha n}, \quad & x \in S \\
         0, \quad & x \not\in S
         \end{cases}\ ,
\end{equation}
where $|S| = 2^{\alpha n}$ (with $0<\alpha < 1$) is the number of strings contained in $S$. The observer samples once from the distribution, receiving $x$, and then executes a classical circuit $C$ with $x$ as an input, finally producing either the output $1$ if he guesses that the distribution is $p_S$, or the output $0$ if he guesses that the distribution is $p_I$.\footnote{The conclusion we reach below would not change much if he were permitted to sample from the distribution a number of times polynomial in $n$.}

\begin{figure}
    \centering
   \includegraphics[width=0.7\columnwidth]{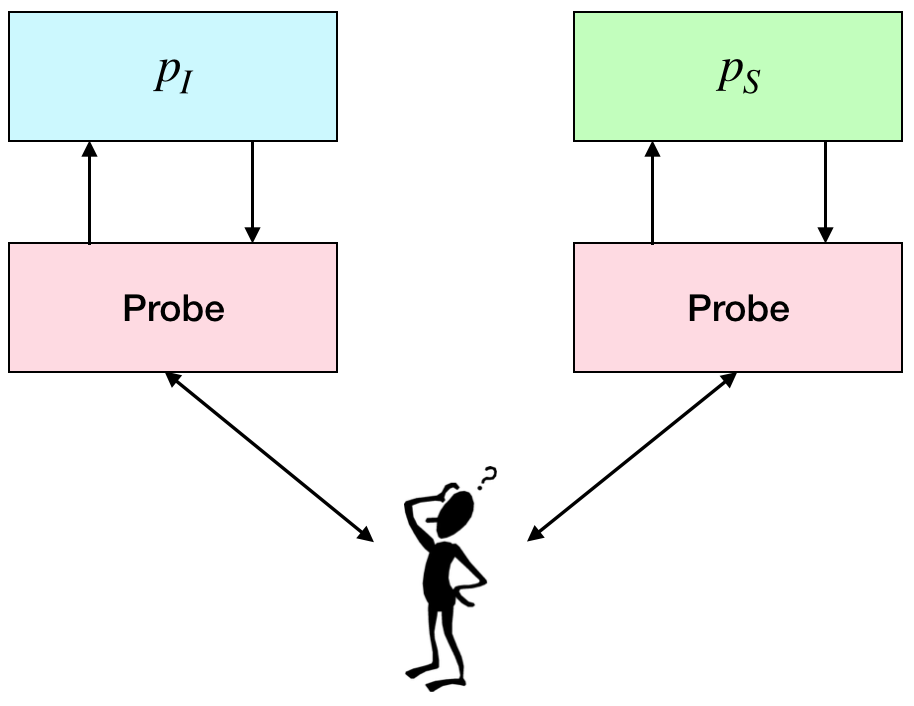}
    \caption{An observer samples from a distribution and attempts to decide whether the distribution is uniformly random or not.}
   \label{fig:imitation}
\end{figure}

For $n$ large, it is clear that $C$ can be chosen so that the observer guesses correctly with a high success probability. Suppose, for example, that he outputs $1$ if $x\in S$ and he outputs $0$ if $x\not\in S$. It the distribution is actually $p_S$, this guess is correct with probability $1$. If the distribution is actually $p_I$, then the guess is correct unless $x$ happens to lie in $S$ ``by accident,'' which occurs with probability $2^{-(1-\alpha) n}$. Therefore, for fixed $\alpha$ and large $n$, the probability of an incorrect guess is exponentially small in $n$.

%In our ``randomness imitation game", we will consider an observer who is attempting to tell whether a given distribution is uniformly random or not; see Figure\ref{fig:imitation}. We will denote the uniformly random distribution as $p_1(x)$, where
%\begin{equation}
%    p_1(x) = \frac{1}{2^n} \quad \forall x \in \{0,1 \}^n.
%\end{equation}
%The other distribution, which is subsuming the role of a machine attempting to imitate a human, is a %uniformly random distribution over a \emph{subset} of bit strings, which we denote as $S$:
%\begin{equation}
%    p_S(x) =
%    \begin{cases}
%    \frac{1}{2^s} \quad & x \in S \\
%         0 \quad & x \not\in S
%         \end{cases}.
%\end{equation}
%Just like in the ordinary imitation game, $p_1$ and $p_S$ are generally different. Let $s=\alpha n$ for some $0<\alpha <1$. In the $n\to \infty$ limit, it is clear that there is \emph{some} experiment that can distinguish $p_1$ and $p_S$ with high probability. To wit, simply output $1$ if a bit string is in $S$ and $0$ otherwise. This strategy will output $1$ with probability $1$ for $p_S$. In contrast, for $p_1$, it will output $1$ with probability $\frac{1}{2^{(1-\alpha)n}}$. Therefore, in the $n\to \infty$ limit, one should be able to distinguish the two distributions with an exponentially small failure probability. The analogue of ``opening the box" would be to apply this circuit and read out the outcome.

However, depending on the structure of the set $S$, the circuit $C$ that distinguishes $p_S$ and $p_I$  might need to be quite complex, making this strategy impractical if the observer has limited computational power. Suppose, for example, that the observer is unable to perform a computation with more than $\Lambda(n)$ gates, where $\Lambda(n)$ grows subexponentially with $n$ --- that is, $\Lambda(n) \leq \exp(f(n))$ where $f(n)$ scales sublinearly with $n$. Then we can show that the set $S$ can be chosen such that this computationally bounded observer has only an exponentially small chance of distinguishing $p_S$ and $p_I$; that is, his probability of guessing the distribution correctly is no better than $1/2 + 2^{-cn}$, where $c$ is a positive constant. In that case, we say that the distribution $p_S$ is \textit{pseudorandom}. 

%However, in practice, applying such a circuit may prove to be very difficult. This is because one often has a limited time to perform the experiment and there is only so much computation one can perform in a limited window of time. We will see that this computational limitation effectively prevents the observer from ``opening the box."

%To be concrete, let us introduce a cutoff on the experimental time, $\Lambda(n)$ and only consider experiments that can be performed in time $\Lambda(n)$.  We will find that, as long as $\Lambda(n)$ scales subexponentially in $n$,\footnote{This means that $\Lambda(n) \leq \exp(f(n))$ for some $f(n)$ that scales sublinearly with $n$.} there will always be a set $S$ for which $p_S$ and $p_1$ are practically indistinguishable. This means that no experiment can distinguish them with success probability greater than $\frac{1}{2^{cn}}$ for some $c>0$. Such distributions will be referred to as \emph{pseudorandom} distributions. 

To show that such a pseudorandom distribution $p_S$ exists, we argue in two steps. In the first step, we consider some fixed circuit $C$, and denote by $L_C$ the set of $n$-bit input strings for which $C$ outputs 1 (we say that $C$ ``accepts'' the strings in $L_C$). If the input $x$ is chosen by sampling from $p_I(x)$, then $C$ accepts $x$ with probability
\begin{equation}
    P_C(I) = \frac{|L_C|}{2^n}, 
\end{equation}
while if $x$ is chosen by sampling from $p_S(x)$, then $C$ accepts $x$ with probability
\begin{equation*}
    P_C(S) = \frac{|S\cap L_C|}{2^{\alpha n}}.
\end{equation*}
Now suppose that $S$ is chosen randomly from among all subsets of $n$-bit strings with cardinality $|S|=2^{\alpha n}$. We can envision the possible strings as $2^n$ balls, of which the balls accepted by $C$ are colored white, and the balls rejected by $C$ are colored black, while $S$ is a random sample containing $|S|$ of these balls. Suppose that the white balls constitute a fraction $f$ of all the balls. Then, for $n$ large, we expect that $S$ also contains a fraction of white balls which is close to $f$. This intuition can be made precise using Hoeffding's inequality, from which we derive
\begin{equation}
    \Pr\left(|P_C(S)-P_C(I)|\geq \epsilon\right) \leq e^{-2|S|\epsilon^2},\label{eq:hoeffding}
\end{equation}
where the probability is evaluated for the uniform distribution over all subsets with $|S|$ elements. Now we can choose $\epsilon$ to have the exponentially small value $\epsilon =|S|^{-1/4}$ (for example) to see that, if $S$ is sampled uniformly with $|S|$ fixed, the probability that $C$ accepts a sample from $p_S$ is exponentially close to the probability that $C$ accepts a sample from $p_I$. We conclude that, not only is it possible to choose the subset $S$ such that the fixed circuit $C$ can barely distinguish $p_S$ from $p_I$, but furthermore most choices for $S$ with $|S|=2^{\alpha n}$ have this property.

We have now completed the first step in our two-step argument. But so far we have only shown that $S$ can be chosen such that $p_S$ and $p_I$ are hard to distinguish for one fixed circuit $C$. We wish to make a much stronger claim, that there is a choice for $S$ such that $p_S$ and $p_I$ are nearly indistinguishable by \textit{any} circuit with a number of gates subexponential in $n$.

To prove this stronger claim we proceed with the second step in the argument. For a collection of circuits $\{C_1,\ldots, C_N\}$, what is the probability that $|P_{C_i}(S)- P_{C_i}(I)|\geq \epsilon$, for at least one $i$? An upper bound on this probability follows from the union bound, which asserts that 
\begin{equation}
    P(A_1\cup \cdots \cup A_N) \leq \sum_{i=1}^N  P(A_i),
\end{equation}
where $\{A_1, A_2, \dots, A_N\}$ is any set of events. Using equation \eqref{eq:hoeffding}, we conclude that the probability that at least one of the $N$ circuits distinguishes $p_S$ from $p_I$ with probability at least $\epsilon$ is no larger than $Ne^{-2|S| \epsilon^2}$.

How many possible circuits are there which act on the $n$-bit input $x$ and contain $m$ computation steps? In each step of the computation, we either input one of the bits of $x$ or we execute a gate which is chosen from a set of $G$ possible gates, where $G$ is a constant. Our claim will hold if each gate in $G$ has a constant number of input and output bits, so for simplicity let's assume that each gate has at most two input bits and generates a single output bit (like a NAND gate for example). Each two-bit gate acts on a pair of bits which are outputs from previous gates; this pair can be chosen in fewer than $m^2$ ways. Therefore, the total number $N(m)$ of size-$m$ circuits can be bounded as
\begin{align}
    N(m) \le \left((n+G)m^2\right)^m, 
\end{align}
which implies
\begin{align}
\log N(m) \le m\left(2 \log m + \log(n+G)  \right).
\end{align}

Even if we choose an exponentially large circuit size $m = 2^{\gamma n}$ and an exponentially small error $\epsilon = 2^{-\delta n}$, we find that $N(m) e^{-2|S|\epsilon^2}$ is doubly exponentially small in $n$ for $|S| = 2^{\alpha n}$ and $\alpha > \gamma + 2 \delta$. Hence, if $S$ is randomly chosen, it's extremely likely that the distributions $p_S$ and $p_I$ are indistinguishable by circuits of size $2^{\gamma n}$, up to an exponentially small error.

%In order to prove this claim, suppose we have circuits $\{C_1,\ldots, C_N\}$. What is the probability that $|p_{C_i}(S)[1] - p_{C_i}(I_n)[1]|\geq \epsilon(S)$ for at least one $i$? In order to answer this question, it is convenient to use Boole's inequality, also known as the union bound:
%\begin{equation}
%    P(A_1\cup \cdots \cup A_N) \leq \sum_{i=1}^N  P(A_i),
%\end{equation}
%where $A_i$ is any event. Therefore, for a uniformly randomly chosen $S$, the probability that at least one of the circuits can distinguish $p_1$ and $p_S$ with probability $\epsilon$ is bounded by $Ne^{-2|S| \epsilon^2}.$
%%%
%Suppose we use at most $m$ logic gates. We can specify any such circuit by specifying the type of the logic gate, its input, and its output. Therefore, for each of the $m$ logic gates, we need at most $O(\log m)$ bits of information. The total number of bits is $O(m\log m)$. Therefore, for a randomly chosen $S$, the probability that $|p_C(S)[1]-p_C(I_n)[1]|$ is greater or equal to $\epsilon$ against all circuit $C$ using $m$ logic gates is bounded by $e^{cm\log m - 2|S| \epsilon^2}$ for some constant $c$. It is interesting to consider the regime in which all the relevant parameters scales exponentially with $n$. Let $|S|= 2^{\alpha n}$, $m = 2^{\gamma n}$, and $\epsilon = 2^{-\delta n}$. We see that, if $\alpha > \gamma + 2\delta$, with overwhelming probability, $p_1$ and $p_S$ are indistinguishable up to an exponentially small error with an exponent $\delta$. 

To summarize, we've shown that the set $S$ can be chosen so that the probability distribution $p_S$ has these properties: 
%there is a probability distribution $p_S$ with the following properties. 
\begin{enumerate}
    \item %Its entropy density($\alpha$) deviates from the maximal entropy density by an $\Theta(1)$ amount.
    Its entropy per bit $\alpha$ is a positive constant less than 1.
    \item It is statistically distinguishable from $p_I$ with an exponentially small failure probability.
    \item If $\gamma < \alpha$, any circuit using at most $2^{\gamma n}$ gates almost always fails to distinguish $p_S$ and $p_I$. 
\end{enumerate}
If $n$ is macroscopic, the task of distinguishing $p_S$ from $p_I$ can be absurdly difficult, even if the entropy density $\alpha$ is quite small. Suppose, for example, that $n=10^{23}$ is comparable to Avogadro's number, and $\alpha = 10^{-12}$. Choosing $\gamma = \delta = 10^{-13}$ we conclude that a circuit with $m= 2^{10^{10}}$ gates can distinguish $p_S$ from $p_I$ with a success probability no larger than $\epsilon = 2^{-10^{10}}$. Even if we could perform one gate per unit of Planck time and Planck volume, an unimaginably large spacetime region would be required to execute so large a circuit. 

The existence of pseudorandom distributions was first suggested by Yao~\cite{Yao1982}, and the construction we have described was discussed by Goldreich and Krawcyzk~\cite{Goldreich1989}. % in which the authors refer to such distributions as sparse pseudorandom distributions. 
In our analysis we assumed that the observer executes a deterministic circuit, but it turns out that giving the observer access to a random number generator does not make his task any easier~\cite{Goldreich1989}.

%\subsection{More than a bit}
%So far, we have only considered a two-outcome measurement experiment. What if we have a $k$-bit outcome? With a minor modification, the analysis in the previous Section would go through. The only difference is that we should consider more colors. Specifically, we will have $2^k$ different colors, each of which can be represented by a bit string $x\in \{0,1\}^k$,

%\textcolor{red}{(IK) Here's my stab at this.}

%\textcolor{blue}{(ET) I can see the relevance now, but I'm still entirely convinced that this paragraph is necessary. I think the quantum pseudorandomness part is rather self-contained as a whole, and doesn't really require any more justification. But I'm OK with keeping this as well.}

Up until now, we assumed that the observer performs a computation whose output is a single bit. But what if he obtains a $k$-bit output instead? Can we choose the set $S$ so that for all circuits of bounded size the probability distribution governing the $k$ output bits is very similar for input strings drawn from $p_S$ and $p_I$?
%%%
%Suppose, for concreteness, that the observer obtains a $k$-bit output. Can we choose the set $S$ so that for all circuits of bounded size the probability distribution governing the $k$ output bits is very similar for input strings drawn from $p_S$ and $p_I$? 
Our previous reasoning does not have to be modified much to handle this case. Now for the fixed circuit $C$, we denote by $L_C[y]$ the set of $n$-bit input strings for which $C$ outputs the $k$-bit string $y$, and we denote by $P_C(I)[y]$, $P_C(S)[y]$ the probability that $C$ outputs $y$ when receiving as input a sample from $p_I$, $p_S$ respectively. Now we envision the $n$-bit input strings as balls which can be colored in $2^k$ possible ways, corresponding to the $2^k$ possible values of the output $y$. Applying the previous argument to each color, we find that when $S$ is chosen at random from among all subsets with cardinality $|S|$,
\begin{equation}
    \text{Prob}\left[|P_C(S)[y]-P_C(I)[y]|\geq \epsilon\right] \leq e^{-2|S|\epsilon^2},\label{eq:hoeffding-y}
\end{equation}
for each output $y$. From the union bound, the probability that $|P_C(S)[y] - P_C(I)[y]|$ exceeds $\epsilon$ for at least one value of $y$ is bounded above by $2^k e^{- 2|S|\epsilon^2}$, which also provides an upper bound on the probability that the total variation distance between $P_C(S)$ and $P_C(I)$ exceeds $\epsilon' = 2^k \epsilon$. Therefore the total variation distance will be no larger than $\epsilon'$ with high probability as long as $2^{\alpha n} 2^{-2k}\epsilon'^2$ is large, which means $\epsilon'$ can be exponentially small in $n$ provided $2k < \alpha n$. 
%%%
%Let $p_C(S)[x]$ be the probability that a circuit $C$ outputs $x$ if the input is sampled uniformly from $S$. Similarly, let $p_C(I_n)[x]$ be the probability that a circuit $C$ outputs $x$ if the input is uniformly sampled from $I_n$. Let $L_C(n)[x]$ be a set of $n$-bit strings for which $C$ outputs $x$. We have
%\begin{equation}
%    p_C(I_n)[x] = \frac{|L_C(n)[x]|}{2^n}
%\end{equation}
%and
%\begin{equation}
%    p_C(S)[x] = \frac{|S\cap L_C(n)[x]}{2^s}.
%\end{equation}
%%%
%Therefore, we see that $p_C(I_n)[x]$ can be viewed as a fraction of balls whose color is $x$ and $p_C(S)[x]$ can be viewed as a fraction of balls whose color is $x$ within the set $S$. Therefore, the probability that these two differ by at least $\epsilon$ is bounded by 
%\begin{equation}
 %   e^{-2|S| \epsilon^2}.
%\end{equation}
%Therefore, the probability that the two are different by at least %$\epsilon$ for at least one $x$ is smaller than 
%\begin{equation}
%    \left(1- e^{-2|S| \epsilon^2}\right)^{2^k}.
%\end{equation}
%Therefore, with probability $\sim 1-2^ke^{-2|S|\epsilon^2}$, there is a distribution for which the total variation distance between $p_C(S)[x]$ and $p_C(I_n)[x]$ is $2 \times 2^k \epsilon$ or less. Let $\epsilon'$ be the total variation distance. If $s> 2k$, then $\epsilon'$ can be made to be very small. Therefore, the pseudorandom distribution can continue to fool the observer as long as the observer is small compared to the system.
%%%
We conclude that the pseudorandom input distribution and the uniformly random input distribution will yield exponentially close output distributions as long as the observer's output register is small compared to the entropy of $S$.

The preceding argument shows that, indeed, there are probability distributions which are computationally indistinguishable from the uniformly random distribution. But can we make such a distribution \emph{efficiently}? It turns out that our argument for the computational hardness of distinguishing a pseudorandom distribution from a uniformly random distribution can be used to show that such distributions are typically hard to produce with polynomial-sized circuits.\footnote{We thank Adam Bouland for emphasizing this point.} 
%Fortunately, 
However, there are distributions that can be created using polynomial-sized circuits which, under reasonable complexity-theoretic assumptions, are difficult to differentiate from the uniformly random distribution,  for any polynomial-sized circuit. Such distributions can be generated by pseudorandom generators~\cite{Arora2009}. We will give a more complete description of such constructions for the quantum case in Section~\ref{sec:quantum-pseudo}.

Efficient sampling from a (classical) pseudorandom distribution is analogous to the formation and partial evaporation of a black hole. Pseudorandom number generators consult a random ``key'' which is hidden from the adversary, and then compute a function which depends on the key. This function is chosen so that an output drawn from the resulting family of outputs indexed by the key is computationally indistinguishable from the output of a truly random function. In the case of the partially evaporated black hole, the key becomes a black hole microstate, and the key-dependent function evaluation becomes the chaotic unitary evolution of the evaporating black hole. An adversary samples the Hawking radiation, and attempts to determine whether the sample is drawn from a thermal distribution or not. 

To properly discuss the evaporating black hole, we will need to consider the quantum version of pseudorandomness, to which we turn in the next two sections. But our simplified classical model of ``Hawking radiation'' is instructive. It teaches us that the (classical) adversary can interact with the (classical) radiation for a subexponential time (or even for the exponential time $2^{\gamma n}$ if $\gamma$ is sufficiently small), without ever suspecting that the radiation is far from uniformly random. On the other hand, that conclusion may no longer apply if the adversary collects $k$ bits of information where $k$ is sufficiently large ($k > \alpha n/ 2)$. Both of these features will pertain to the quantum version of our story.

\section{Quantum pseudorandomness\label{sec:quantum-pseudo}}
%In the classical example, we saw that there is a distribution that ``looks" random from the perspective of a computationally limited observer. The argument was that there is a pseudorandom distribution that can fool some circuit. Then we could bound the number of circuits with a limited number of gates. 

Now consider the quantum version of the task described in the Section~\ref{sec:classical-pseudo}. Our observer receives a quantum state $\rho$, and is challenged to guess whether $\rho$ is maximally mixed or not. For that purpose, he performs a quantum computation with $\rho$ as input, and he outputs a single bit: 0 if he guesses $\rho$ is maximally mixed and 1 otherwise. 

Following our analysis of the classical case, let's first suppose that the observer executes a particular fixed quantum circuit. That means the observer measures a particular Hermitian observable $A$ with unit operator norm. Suppose we try to fool the observer by providing as input a pure state $\rho = |\psi\rangle\langle\psi|$. How well can the observable $A$ distinguish this pure state from the maximally mixed state? 

Suppose that $|\psi\rangle$ is chosen uniformly at random from among all $n$-qubit pure states. Then Levy's lemma~\cite{Hayden2004a} says that 
\begin{equation}
    \Pr\left(\left|\langle\psi|A|\psi\rangle - \frac{\text{Tr}(A)}{2^n}\right|\geq \epsilon\right) \leq e^{-c\,2^n\epsilon^2 }\label{eq:levy}
\end{equation}
for some constant $c$, where the probability is evaluated with respect to the invariant Haar measure on the $n$-qubit Hilbert space. This means that, for $n$ large, the pure state $|\psi\rangle$ can be chosen so that $|\psi\rangle$ and the maximally mixed state are exponentially difficult to distinguish using the observable $A$. Furthermore, most pure states have this property. Even a pure quantum state can pretend to be maximally mixed, and the observer will not know the difference!\footnote{If two identical copies of $|\psi\rangle$ are available, then it is easy to distinguish the pure state $|\psi\rangle$ from the maximally mixed state by conducting a swap test. Here we assumed that only a single copy is available.}

As in the classical case we can strengthen this claim: The state $|\psi\rangle$ can be chosen so that $|\psi\rangle$ is hard to distinguish from the maximally mixed state not just for one fixed quantum circuit, but for \textit{any} quantum circuit of reasonable size. To carry out this step of the argument, we'll need an upper bound on the number of quantum circuits of specified size; here we confront the subtlety that quantum circuits, unlike classical ones, form a continuum, but this wrinkle poses no serious obstacle to completing the argument. If we settle for specifying the unitary transformation realized by a circuit with $m$ gates to constant accuracy, it suffices to specify each gate to $O(\log m )$ bits of precision. Therefore, as in the classical case, the complete circuit can be specified by $O(m\log m)$ bits. It follows that, if $m$ is subexponential in $n$, then the number $N(m) $ of circuits with size $m$ is the exponential of a function which is subexponential in $n$. In contrast, the right-hand side of equation~\eqref{eq:levy} is the exponential of an exponential function of $n$. Using the union bound, we conclude that if the pure state $|\psi\rangle$ is chosen uniformly at random, it will, with high probability, be hard to distinguish $|\psi\rangle$ from the maximally mixed state using \emph{any} circuit of size subexponential in $n$. 

On the other hand, if we were not concerned about the complexity of the observer's task then it would be easy to distinguish $|\psi\rangle$ from the maximally mixed state. The observer could perform a projective measurement with the two outcomes $\{E_0 = I - |\psi\rangle\langle \psi |,\  E_1=|\psi\rangle\langle \psi|\}$, guessing that the input state is $|\psi\rangle$ if he obtains the outcome $E_1$, and guessing that the input state is maximally mixed if he obtains the outcome $E_0$. This strategy always succeeds if the input is $|\psi\rangle$, and fails with the exponentially small probability $2^{-n}$ if the input is maximally mixed. The trouble is that, for a typical pure state $|\psi\rangle$, this measurement is far too complex to carry out in practice. 

A typical pure quantum state is somewhat analogous to the distribution $p_S$ we described in Section~\ref{sec:classical-pseudo}. In both cases, it is hard for an observer who is limited to performing polynomial-size computations to tell that the state is not uniformly random, even though an observer with unlimited computational power can tell the difference. Furthermore, both examples are subject to the same criticism --- it is computationally hard to sample uniformly from Haar measure (that is, to prepare a ``typical'' pure state), just as it is computationally hard in the classical setting to sample from the the distribution $p_S$. In the quantum setting, as for the classical setting, we may ask a more nuanced question: Can quantum states be prepared \textit{efficiently} which are hard to distinguish from maximally mixed states? This more nuanced question is the relevant one as we contemplate the properties of the radiation emitted by a partially evaporated black hole, because the formation and subsequent complete evaporation of a black hole can occur in a time that scales like $S_{\textrm{bh}}^{3/2}$, where $S_{\textrm{bh}}$ is the initial black hole entropy. Hence, the preparation of the Hawking radiation can be simulated accurately by an efficient quantum circuit.

%Because Haar-random states are typically computationally indistinguishable from the maximally mixed state, most states can be viewed as a quantum analogue of the pseudorandom distribution. However, we believe it is inadequate to view such states as an accurate model of black hole microstates. Black holes are formed in a time that is polynomial in its entropy while Haar-random states typically have an exponentially large complexity. Most of the properties ascribed to black holes through Haar-randomness only require the existence of unitary $k$-designs which are efficiently computable.

The answer is yes (under a reasonable assumption), as was shown recent by Ji, Liu, and Song~\cite{Ji2018}; pseudorandom quantum states \textit{can} be prepared efficiently. The assumption we need is the the existence of a family of \emph{quantum-secure pseudorandom functions} $\{\mathrm{PRF}_k\}_{k\in K}$. This means that each $\mathrm{PRF}_k$ can be efficiently computed, but it is difficult to distinguish a randomly sampled member of $\{\mathrm{PRF}_k\}$ from a truly random function with any efficient quantum algorithm. The set $K$ is called the \emph{key space} of the function family. The existence of such pseudorandom functions follows from the existence of  quantum-secure one-way functions, an assumption which is standard in cryptography. 

The key idea is that we can construct a pseudorandom quantum state as a superposition of computational basis states, where all basis states appear with equal weight except for a phase, and the phases appear to be random to a computationally bounded observer. Specifically, we consider a family of states $\{|\phi_k\rangle \}_{k \in K}$
\begin{equation}
    |\phi_k\rangle = \frac{1}{\sqrt{N}} \sum_{x \in X } \omega_N^{\text{PRF}_k(x)}|x\rangle,
\end{equation}
where $N=2^n$, $\omega_N = e^{2\pi i / N}$, $X = \{0,1,2,\dots, N{-1}\}$, and $\{\mathrm{PRF}_k: X \to X\}_{k\in K}$ is a family of quantum-secure pseudorandom functions. We can show that a uniform mixture of the states $\{|\phi_k\rangle \}$ is computationally indistinguishable from the maximally mixed state.\footnote{In fact, we can simplify the construction. It was shown in~\cite{brakerski2019} that the same family of states is still pseudorandom if we replace the root of unity $\omega_N$ by $-1$.} 

We may argue as follows. First we consider the family of \emph{all} functions $f_{k'}:X\to X\}_{k'\in K'}$ indexed by key space $K'$, and the corresponding family of pure states
\begin{equation}
    |f_{k'}\rangle = \frac{1}{\sqrt{N}} \sum_{x \in X} \omega_N^{f_{k'}(x)} |x\rangle. \label{eq:functionstate}
\end{equation}
The first thing to note is that $\{|f_{k'}\rangle \}$ is information-theoretically indistinguishable from Haar-random; we state this fact for the reader's convenience in Lemma \ref{lem:ji}.
\begin{lemma}[\cite{Ji2018}, Lemma 1]\label{lem:ji}
Let $\{|f_{k'}\rangle\}$ be the family of states defined in equation \eqref{eq:functionstate}. Then, for $m$ polynomial in $n$, the state ensemble $\{|f_{k'}\rangle^{\otimes m}\}$ is statistically indistinguishable from the ensemble $\{|\psi\rangle^{\otimes m}\}$ where $|\psi\rangle$ is Haar-random, up to a negligible error.
\end{lemma}

Furthermore, the ensemble $\{|\phi_k \rangle \}$ cannot be efficiently distinguished from the ensemble $\{|f_{k'}\rangle\}$. If it could be, then we could leverage this fact to efficiently distinguish $\{ \text{PRF}_k \}$ from a family of random functions~\cite{Ji2018}, contradicting our assumption that $\{ \text{PRF}_k \}$ is a quantum-secure pseudorandom function family. It now follows that the ensemble $\{|\phi_k \rangle \}$ cannot be efficiently distinguished from a maximally mixed state.

So far we have shown that a uniform mixture of the states $\{|\phi_k\rangle \}$ is pseudorandom; it remains to show that this mixture can be prepared efficiently. 
We start with a product of qubits, each in the state $|0\rangle$, and apply a Hadamard gate to each qubit to obtain the state
%\begin{equation}
%    \frac{1}{N\sqrt{|K|}} \sum_{x,y \in X} \sum_{k \in K} |x\rangle |k\rangle |y\rangle.
%\end{equation}
\begin{equation}
    \frac{1}{\sqrt{N|K|}} \sum_{x \in X} \sum_{k \in K} |x\rangle |k\rangle .
\end{equation}
%\jp{Fix: $y$ register should be initialized to $|1\rangle$.} \IK{I don't see an error here?} \jp{The quantum Fourier transform maps $|y\rangle$ to $N^{-1/2}\sum_z \omega^{zy}|z\rangle$, which is not the state we wanted here, unless we set $y=1$.}\IK{Got it, thanks!}
Next, we apply the quantum Fourier transform (which has complexity polynomial in $n$) to another $n$-qubit register that is initialized in the state $|00\dots 01\rangle$, obtaining
\begin{equation}
    \frac{1}{N\sqrt{|K|}} \sum_{x, y \in X} \sum_{k \in K} |x\rangle|k\rangle \omega_N^y |y\rangle.
\end{equation}
Now, we compute $\text{PRF}_k(x)$ and subtract modulo $N$ from the $y$ register. This computation can be done efficiently because by assumption the $\text{PRF}_k$ is an efficiently computable function; the resulting state is
\begin{equation}
    \frac{1}{N\sqrt{|K|}} \sum_{x, y \in X} \sum_{k \in K} |x\rangle|k\rangle \omega_N^y |y-\text{PRF}_k(x) \rangle.
\end{equation}    
After shifting the summation index $y$, we have, up to a global phase,
\begin{equation}
\begin{aligned}
    \frac{1}{N\sqrt{|K|}} \sum_{x, y \in X} \sum_{k \in K} \omega_N^{\text{PRF}_k(x)} |x\rangle|k\rangle \omega_N^y |y\rangle
    =\frac{1}{\sqrt{|K|}}\sum_k |\phi_k\rangle|k\rangle|\text{QFT}\rangle,
\end{aligned}
\end{equation}
where $|\text{QFT}\rangle = N^{-1/2}\sum_{y \in X}\omega_N^y |y\rangle$. After the key $|k\rangle$ is discarded,  the marginal state over the first register is the uniform mixture of $\{|\phi_k\rangle \}$. Thus, we have prepared this mixture efficiently.

%More precisely, we say that a keyed family of quantum states $\{|\psi_k\rangle\}_{k\in K}$ is pseudorandom if:
%\begin{enumerate}
%\item There exists a polynomial time quantum algorithm that generates $|\psi_k\rangle$ on input $k$.
    
%\item Given $m(\kappa)$ copies of $|\psi_k\rangle$ for any polynomial $m\in \mathrm{poly}(\kappa)$, the state $|\psi_k\rangle^{\otimes m}$ is computationally indistinguishable from $|\phi\rangle^{\otimes m}$, where $|\phi\rangle$ is Haar random. More precisely, given any polynomial time algorithm $\mathcal{A}$, we have
%\begin{align}
%\left|\Pr_{k\leftarrow K}\left[\mathcal{A}\left(|\psi_k\rangle^{\otimes m}\right)=1\right]-\Pr_{|\phi\rangle\leftarrow \mu}\left[\mathcal{A}\left(|\phi\rangle^{\otimes m}\right)=1\right]\right| = \epsilon(\kappa).
%    \end{align}
%\end{enumerate}

The definition of a pseudorandom quantum state in reference~\cite{Ji2018} is really overkill for our purposes. Those authors are concerned with cryptographic applications, and therefore consider a definition (as stated in Lemma \ref{lem:ji}) where, for each value $k$ of the key, $m$ identical copies of $|\phi_k\rangle$ are available where $m$ is polynomial in $n$. We will not encounter such scenarios in this paper. Therefore, we may instead adopt a simplified definition of pseudorandomness which is more suitable for the application to black hole physics. In Definition~\ref{def:pseudorandom} below, the size $|H|$ of the remaining black hole parametrizes how difficult it is to distinguish radiation emitted from the partially evaporated black hole from the maximally mixed state. In this sense, the remaining black hole $H$ serves as the key space of the pseudorandom radiation state. Even if $|H|$ is less than half of the initial black hole entropy, this task remains difficult so long as the remaining black hole $H$ is macroscopic. Our hypothesis that the Hawking radiation is pseudorandom provides a way to  formalize the idea that the Hawking radiation is effectively thermal even when the state of $E$ has relatively low rank because $|H| \ll |E|$.

\section{Is Hawking radiation pseudorandom?\label{sec:hawking-pseudo}}

We have now seen, in both the classical and quantum settings, that pseudorandom states exist. Though in principle these states are almost perfectly distinguishable from maximally mixed states, in practice no observer with reasonable computational power can tell the difference. Moreover, under standard cryptographic assumptions, there exist constructions of such states which can be efficiently prepared. But up to this point we have not addressed whether pseudorandom quantum states can be efficiently prepared in plausible physical processes like the evaporation of a black hole. 

In the case of a black hole which forms from gravitational collapse and then \emph{completely} evaporates, the resulting state of the emitted Hawking radiation, though highly scrambled, would \emph{not} be pseudorandom. We take it for granted that the time evolution of the quantum state can be accurately approximated by a quantum circuit, which has size polynomial in the initial black hole entropy $S_{\text{bh}}$ because the evaporation process takes a time $O\left(S_{\text{bh}}^{3/2}\right)$. We may consider a toy model of this process, in which the initial state $|\phi_{\text{matter}}\rangle$ of the collapsing matter is a product state of $n$ qubits $|\phi_{\text{matter}}\rangle=|0\rangle^{\otimes n}$, and the final state after complete evaporation is $|\Psi_{\text{fin}}\rangle = U|\phi_{\text{matter}}\rangle$, where $U$ is a unitary transformation constructed as a polynomial-size circuit. In this case, an observer could just execute this circuit in reverse, hence applying $U^\dagger$ to $|\Psi_{\text{fin}}\rangle$, and then measure the qubits in the standard basis, thus easily distinguishing $|\Psi\rangle$ from the maximally mixed state. We see that, if $\rho$ is a state that can be prepared by a polynomial-size quantum circuit, yet is hard to distinguish from maximally mixed by polynomial-size circuits, then $\rho$ cannot be pure.

Instead, we consider a \emph{partially} evaporated black hole as in Figure~\ref{fig:psi}.  We imagine that the $n$-qubit state $\rho_{EB}$ is prepared by applying a polynomial-size unitary circuit $U_{\text{bh}}$ to the initial state $|0\rangle^{\otimes n}|0\rangle^{\otimes k}$ of $EBH$, where $H$ is a $k$-qubit system, and then discarding $H$. In our toy model, $EB$ is the Hawking radiation that has been emitted so far, and $H$ is the remaining black hole. (Recall that $B$ is a small portion of the emitted Hawking radiation whose properties we will investigate later; for the purpose of the present discussion we are only interested in the state of $EB$, the full radiation system.) If our observer had access to $H$ as well as $EB$, he could easily tell that the state is not maximally mixed, but what if $H$ is inaccessible?

\begin{figure}[ht]
\centering
\includegraphics[width=0.5\columnwidth]{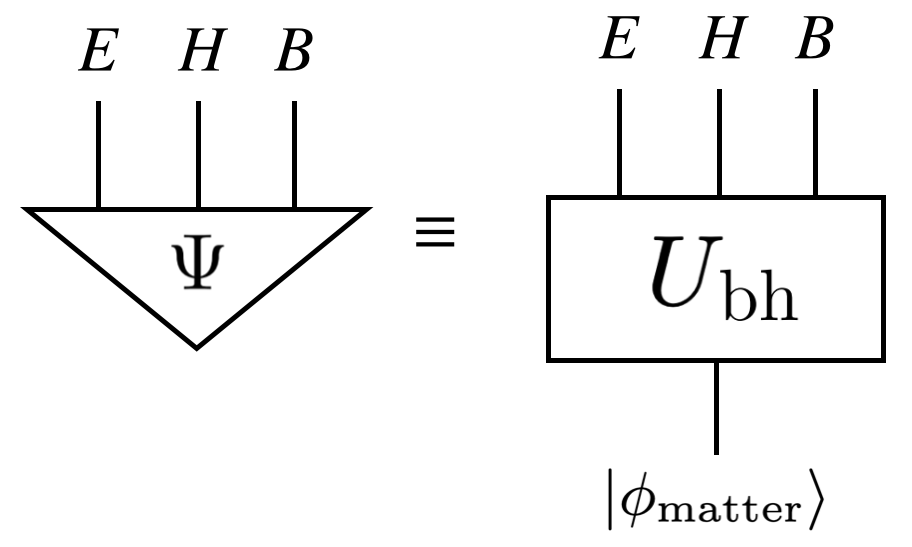}
\caption{Our toy model of a partially evaporated black hole, where $EB$ is the Hawking radiation emitted so far, and $H$ is the remaining black hole. The initial state $|\phi_{\text{matter}}\rangle$ of the gravitationally collapsing matter is modeled as a product state. We conjecture that the unitary black hole dynamics prepares a pseudorandom state of $EB$.}
\label{fig:psi}
\end{figure}

%For example, $H$ could be an old black hole which has been evaporating for longer than its Page time, and $EB$ could be the Hawking radiation that has been emitted since the black hole formed. Here, $B$ is also part of the exterior radiation that we will decide to inspect separately later, but for now one can simply view $EB$ as the entirety of the radiation emitted so far. 
We are particularly  interested in the case where $1 \ll k = |H| < n = |EB|$, so that $\rho_{EB}$ fails to have full rank, and must therefore be information-theoretically distinguishable from the maximally mixed state; this situation resembles the classical model discussed in Section~\ref{sec:classical-pseudo}, where the entropy of the distribution $p_S$ is substantial but not maximal. Could the state $\rho_{EB}$ of the Hawking radiation, which is prepared by unitary evolution of $EBH$ for a time which is polynomial in $|EB|$, be pseudorandom? 

This is a question about quantum gravity, and we don't know the answer for sure, but we can make a reasonable guess. We have already seen in Section \ref{sec:quantum-pseudo} that quantum circuits exist that efficiently prepare pseudorandom quantum states. Since black holes are believed to be particularly potent scramblers of quantum information, it is natural to conjecture that the internal dynamics of a black hole can produce pseudorandom states as well. Indeed, we may expect similar behavior for the radiation emitted by other strongly chaotic quantum systems aside from black holes. Only for the case of black holes, though, where we face the daunting firewall puzzle, will our constructions of robust logical operators acting on $EH$ seem to have a natural interpretation.  

To better understand why the state of $EB$ might be hard to distinguish from a maximally mixed state, we may suppose, for example, that $\rho_H$ is maximally mixed so that the pure state of $EBH$ has the form
\begin{equation}
    |\Psi\rangle_{EBH} = \frac{1}{2^{|H|/2}} \sum_i |\psi_i\rangle_{EB}\otimes |i\rangle_H ,
\end{equation}
where the states $\{|\psi_i\rangle_{EB}\}$ are orthonormal.
The marginal state of $EB$ is then
\begin{equation}
    \rho_{EB} = \frac{1}{2^{|H|}} \sum_i |\psi_i\rangle \langle \psi_i|.
\end{equation}
Suppose the observer receives a state which is either $\rho_{EB}$ or the maximally mixed state $\sigma_{EB} = I_{EB}/2^{|EB|}$. A natural test is as follows: The observer augments $EB$ with the maximally mixed state of $H$ (which is easy to prepare), and then measures the projection onto $|\Psi\rangle$. This can be done efficiently by applying $U_{\text{bh}}^{-1}$ and then measuring in the standard basis. If the input state is $\sigma_{EB}$, the projection onto $|\Psi\rangle$ succeeds with probability $2^{-(|EB| + |H|)}$, while if the input state is $\rho_{EB}$ the success probability is 
\begin{equation}\label{eq:test-mixed}
    \langle \Psi|  \frac{\sum_i |\psi_i\rangle\langle\psi_i|}{2^{|H|}} \otimes \frac{I_H}{2^{|H|}}  |\Psi \rangle= \frac{1}{2^{2|H|}}.
\end{equation}
Thus this test distinguishes $\rho_{EB}$ and $\sigma_{EB}$, but only with a probability that is exponentially small in $H$. 

To conduct a better test we would somehow need to exploit the structure of the ensemble $\{|\psi_i\rangle_{EB}\}$. But if as we expect black holes are especially effective information scramblers, it is reasonable to suppose that the ensemble lacks any special properties that can be exploited by an observer who is limited to performing a polynomial-time quantum computation. If so, the Hawking radiation is pseudorandom, and the test we have described may be nearly optimal. 

In the example above we have assumed that the radiation has infinite temperature. The actual behavior of a black hole evaporating in asymptotically flat spactime is more complicated --- the temperature is actually finite, and in fact becomes hotter and hotter as the evaporation proceeds. Conceptually, though, the situation is similar to the idealized case of an black hole evaporating at infinite temperature. At early times, when $|H| \gg |EB|$, we expect the radiation emitted at a specified time to be information-theoretically indistinguishable from precisely thermal radiation at the same temperature. At late times, when $|H| \ll |EB|$, the global state of the radiation is distinguishable in principle from a thermal state (with temperature varying according to the time of emission), but we assume that telling the difference is computationally hard because the radiation is highly scrambled.

We also note that the constructions in \cite{Ji2018} reinforce earlier observations concerning the computational hardness of \textit{decoding} the Hawking radiation \cite{Harlow2013,Aaronson2016}. These authors considered the quantum state $|\Psi\rangle_{EBH}$ of an old black hole, and analyzed the task of extracting from the early radiation $E$ the subsystem which is entangled with the recently emitted Hawking mode $B$. This task would be easy for an observer who has access to both $E$ and $H$, but one can argue that there are efficiently preparable states of $EBH$ for which this decoding task cannot be achieved by an observer who performs a polynomial-size quantum computation on $E$ alone. Here, too, the hardness of decoding cannot be proven from first principles, but it follows from plausible complexity assumptions which are standard in ``post-quantum'' cryptography \cite{Harlow2013,Aaronson2016}. Again, the existence of states that are hard to decode does not guarantee that a black hole creates such states, but we take it on faith that if efficient preparation of such states is possible, then a black hole will be up to the job. 

%Whether the actual Hawking radiation in nature is pseudorandom or not will be difficult to assess until we have a full-fledged quantum theory of gravity. In the meantime, it may be worthwhile to understand if the known fast scramblers~\cite{Lashkari2013,Brandao2016} are secretly capable of producing pseudorandom states.\footnote{As a side note, we should emphasize that the pseudorandomness of the Haar-random states is unrelated to the pseudorandomness of the radiation. Approximating the black hole microstate as a Haar-random state is acceptable if the purpose is to calculate quantities that can be obtained from low moments of the state, i.e., Renyi-$k$ entropies for finite $k$ because random unitary circuits form a polynomial $t$-design provided $t>k$. However, the indistinguishability of the given state from the maximally mixed state  against polynomial-time quantum computation does not appear to be computable from low moments of the state.}

To summarize, on the basis of these (admittedly speculative) considerations, we propose that for the quantum state $|\Psi\rangle_{EBH}$ of an old black hole, the state $\rho_{EB}$ of the Hawking radiation is pseudorandom. If $|H| < |EB|$, then the rank of $\rho_{EB}$ is not maximal, so that $\rho_{EB}$ is distinguishable from a thermal state. In fact, an observer with access to $H$ as well as $EB$ could efficiently check that $\rho_{EB}$ is not thermal. Furthermore, an observer without access to $H$ could check that $\rho_{EB}$ is not thermal by performing a quantum computation of exponential size on $EB$ alone. But an observer outside the black hole, who performs a polynomial-size quantum computation on $EB$ without access to $H$, will be able to distinguish $\rho_{EB}$ from a thermal state with a success probability that is at best exponentially small in $|H|$. Our analysis of the robustness of the encoded black hole interior in the remainder of this paper will rest on this assumption.

This discussion highlights the importance of distinguishing the von Neumann entropy of the Hawking radiation from its thermodynamic entropy. After the Page time, the Von Neumann entropy of $EB$ becomes far smaller than the von Neumann entropy of a perfectly thermal state, so one could in principle verify that the Hawking radiation is not perfectly thermal by measuring its von Neumann entropy. The existence of pseudorandom quantum states then implies that measuring the von Neumann entropy with a small error requires an operation of superpolynomial complexity \cite{gheorghiu2020estimating}. One could imagine trying to measure the entropy of the radiation by, for example, withdrawing its thermal energy to operate a heat engine. If the radiation is pseudorandom, though, the radiation would be indistinguishable from thermal radiation in any efficient process, despite its low von Neumann entropy.

Recalling the construction of the pseudorandom state recounted in Section~\ref{sec:quantum-pseudo}, we note \cite{cleve2000fast} that the quantum Fourier transform can be executed with circuit depth $O(\log n)$, and that under plausible cryptographic assumptions the function $\textrm{PRF}_k$ can be computed in depth polylog $n$ \cite{Ji2018}. Thus a pseudorandom state can be prepared in polylog $n$ time. Plausibly, the state preparation can be achieved in a time comparable to the $O(\log n)$ scrambling time of a black hole, as one might naively expect.

%\textcolor{red}{(IK)TODO: Explain why we cannot use Haar-random states. Such an approximation may good for computing things like Renyi entropies, but we are asking a more subtle question here. }

%\textcolor{red}{It is perhaps more convincing to say that given a Haar random state $|\psi\rangle$ on the total system, its reduced density matrix on a subsystem is necessarily computationally pseudorandom. Then our pseudorandomness hypothesis can be presented as a \emph{weakening} of the usual Haar random hypothesis for black holes.}

%{\color{red}(IK) I am not sure if that would be more convincing. While the reduced density matrix of a Haar random state would be pseudorandom, the complexity of Haar-random state is exponential. In contrast, the state generated by the ordinary time evolution is expected to be polynomial in the black hole entropy.}

%{\color{red}(IK) I want to add a short paragraph that succinctly summarizes our model.}

\section{Pseudorandomness and decoupling\label{sec:pr_decoupling}}

%\textcolor{red}{(ET) I have restructured a bit to make the discussion a little more streamlined, especially for the correctability of the channels. The actual code embedding is introduced earlier, etc. Still need to incorporate some of Isaac's discussions into the appropriate places.}

%\textcolor{red}{(IK) I have revised everything up to this point to make the notation/structure more coherent. I will do something similar with the remainder of this paper tomorrow. (on Sunday)}

In this section, we formalize our hypothesis that Hawking radiation is pseudorandom, and explore its implications regarding the firewall paradox~\cite{Almheiri2013}. Our analysis can be viewed as a refinement of the Harlow-Hayden argument~\cite{Harlow2013,Aaronson2016}. 

As formulated in \cite{Almheiri2013} and summarized in Section \ref{sec:intro}, the firewall paradox highlights a conflict between the unitarity of black hole evaporation and the monogamy of entanglement. A possible resolution is that the interior mode $\tilde B$ that purifies a recently emitted Hawking mode $B$ may actually be encoded in the radiation. On the face of it, this resolution flagrantly violates locality, and one wonders whether this violation of locality can be detected by an agent who first interacts with the radiation and then falls through the event horizon to visit the interior. We will argue that, provided the Hawking radiation is pseudorandom and the size of the observer is small compared to the black hole, the nonlocality is undetectable in practice because it would take an exponentially long time for the observer to distill the encoded interior mode before falling into the black hole.
%Therefore, locality remains safe.

%Let us first formalize our assumption on the pseudorandomness of the Hawking radiation and explore its consequences for the original firewall paradox~\cite{Almheiri2013}. In the original firewall paradox~\cite{Almheiri2013}, an infalling observer can extract the interior mode and fall into the black hole, thereby observing a violation of known physics. Harlow and Hayden~\cite{Harlow2013} argued that this would be operationally impossible. We provide an alternative argument to arrive at the same conclusion.

%We argue that any physical entity significantly smaller than the remaining black hole is decoupled from the radiation. The simplest setup concerns a small observer that interacts with the radiation. A more subtle thought experiment involves an observer together with a ``probe", i.e., a macroscopic sub-system (not limited to be small compared to the black hole) which is allowed to interact with the radiation, for example galactic dust from which the black hole radiation is scattered.

We will first present a sketch of the argument in a simplified setting where the radiation interacts with a single observer who is significantly smaller than the remaining black hole. Later on we will extend the argument to the case where the observer has access to a large probe outside the horizon, whose size may be comparable to or even larger than the remaining black hole.

Recall our conventions: Let $O$ denote the observer, $H$ the remaining black hole, $B$ the late outgoing mode, $E$ the early radiation, and $P$ the external probe. We will also refer to the joint system $EB$ as the exterior radiation. All subsystems can be decomposed in terms of qubits, and our statements about computational complexity concern the number of steps in a computation executed by a universal quantum computer.  Below and throughout the remainder of the paper, given an operator $A$, we will use $\|A\|_1 = \mathrm{Tr}(\sqrt{A^\dagger A})$ to denote the trace norm, $\|A\|_F = \sqrt{\mathrm{Tr}(A^\dagger A)}$ to denote the Frobenius norm, and $\|A\|$ to denote the operator norm.

First, we define %precisely 
what it means for the external radiation of the black hole to be pseudorandom. 
\begin{definition}\label{def:pseudorandom}
Let $|\Psi \rangle_{EBH}$ be the state of the black hole and the radiation. Let $\sigma_{EB}=I_{EB}/d_{EB}$ be the maximally mixed state of $EB$, and let $\rho_{EB} = \mathrm{Tr}_H\left(|\Psi\rangle\langle\Psi|\right)$. We say that the state $|\Psi\rangle_{EBH}$ is \textit{pseudorandom} on the radiation $EB$, if there exists some $\alpha > 0$ such that
\begin{equation}
    \left|\Pr\left(\mathcal{M}(\rho_{EB})=1\right) - \Pr\left(\mathcal{M}(\sigma_{EB})=1\right)\right| \leq 2^{-\alpha |H|}, \label{eq:pr}
\end{equation}
for any two-outcome measurement $\mathcal{M}$ with quantum complexity polynomial in $|H|$, the size of the remaining black hole. 
\end{definition}
\noindent This definition captures the notion that no feasible measurement can tell the difference between $\rho_{EB}$ and the maximally mixed state. A few remarks will help to clarify the definition. (1) When we say a measurement of $\rho_{EB}$ has polynomial quantum complexity, we mean it can be performed by executing a quantum circuit of polynomial size acting on $EB$, followed by a qubit measurement in the standard computational basis. Use of ancilla systems is also permitted in the measurement process, provided the ancilla is initialized in a product state.  (2) Of particular interest is the value of the constant $\alpha$ that makes the bound in equation (\ref{eq:pr}) tight for asymptotically large black holes. But because black holes are such effective information scramblers, we would expect a comparable value of $\alpha$ to apply also for black holes of moderate size. There is no obvious small parameter in the problem that would lead us to expect $\alpha$ to be small compared to 1. (3) This definition is appropriate for the case where the Hawking radiation has infinite temperature. As we remarked in Section \ref{sec:hawking-pseudo}, we expect the realistic case of finite-temperature radiation to be conceptually similar, and for similar conclusions to apply in that case. But we will stick with the infinite-temperature case for the rest of the paper to simplify our analysis.\footnote{In the finite temperature case, the entanglement between $B$ and the rest of the system is no longer maximal. This causes an extra complication when we use the quantum error-correction technology in Section~\ref{sec:bh_qecc}, because the encoding map $V$ from $B$ to $EH$ defined by the state $\Psi_{EBH}$ need not be exponentially close to an isometry. Instead we may replace $V$ by the approximate isometry $V\rho_B^{-1/2}$,  which slightly modifies the error bounds derived in Section~\ref{sec:bh_qecc} and \ref{sec:qec}. Similar techniques have been used in \cite{Papadodimas2013,Papadodimas2016,Harlow2014}.} 

Let us now deduce a consequence of Definition~\ref{def:pseudorandom}. We introduce an observer subsystem $O$ initialized in a state $\omega_O$, and  an ancilla subsystem $P$ initialized in the product state $|0\rangle_P$. The main result of this section is the following: Suppose that $|\Psi\rangle_{EBH}$ is pseudorandom, and let $\rho_{OPE}$ be any state of $OPE$ obtained by applying a unitary of polynomial complexity to $\omega_{O}\otimes |0\rangle_P\otimes |\Psi\rangle_{EBH}$. Then the correlation between the observer and the early radiation is exponentially small in $|H|$ for any such state; \emph{i.e.},
\begin{equation}
    \|\rho_{OB} - \rho_{O} \otimes \rho_B \|_1 \leq 6\cdot 2^{-(\alpha|H|-|O|)}, \label{eq:decoupling_main}
\end{equation}
where we have now assumed that $B$ is a single qubit.
We will call~\eqref{eq:decoupling_main} the \emph{decoupling bound}, because it states that the observer $O$ nearly decouples from the exterior radiation mode $B$, and therefore gains negligible information about the interior mode $\tilde B$ which is entangled with  $B$.
In Section~\ref{sec:qec} we leverage \eqref{eq:decoupling_main} to show that the interior mode $\tilde B$ can be regarded as an encoded subsystem of $EH$ which is protected against all ``low-complexity'' errors, where ``low-complexity'' is shorthand for polynomial complexity. 

Prior work~\cite{Harlow2013,Aaronson2016} has suggested that the decoupling bound holds when the size of the remaining black hole is an $O(1)$ fraction of the initial black hole entropy $S_{\text{bh}}$. However, our conclusion goes further. Even if the majority of the initial black hole has evaporated, so that $|H| \ll |EB|$, the observer $O$ and the late radiation $B$ remain decoupled as long as the remaining black hole $H$ is macroscopic and the observer's system $O$ obeys $|O|\ll |H|$.

To derive the decoupling bound, we apply the pseudorandomness assumption to the setup described in Figure~\ref{fig:setup}. %We consider an observer whose initial state is entangled with an auxiliary system which we denote as $O'$. There are $|O|$ Bell pairs between $O$ and $O'$ in the standard basis. The probe is assumed to be initialized in the product state $|0\rangle^{\otimes|P|}$. 
The unitary $U_{\mathcal{E}}$ is applied to the radiation, probe, and the observer. Because the evaporation time of the black hole is polynomial in its size, and $U_{\mathcal{E}}$ is applied before the evaporation is complete, we may assume that $U_{\mathcal{E}}$ is applied in a polynomial time and therefore has polynomial complexity. We  also assume that the initial state of the observer $\omega_O$ is of low complexity, although this assumption is not crucial; we may take the state $\omega_O$ to be arbitrary, at the cost of a slightly weaker decoupling bound. %by purifying the system $O$.

In order to bound the correlation between $B$ and $O$, we consider a complete set of operators acting on $OB$. A convenient choice is the set of Pauli operators $P_i$ acting on $OB$. By a Pauli operator acting on $n$ qubits we mean a tensor product of $n$ $2\times 2$ Pauli matrices; there are $4^n$ such operators $\{P_i, i = 0, 1, 2, \dots , 4^n-1\}$ (where $P_0 = I$) whose phases can be chosen so that each $P_i$ for $i\ne 0$ has eigenvalues $\pm 1$, and the $\{P_i\}$ are orthogonal in the Froebenius norm: $\text{Tr}\left(P_i P_j \right)= 2^n \delta_{ij}$. Here $n$ is the number of qubits in $OB$. Because measurement of $P_i$ is a low-complexity two-outcome measurement, it follows from the assumption that $\Psi_{EBH}$ is pseudorandom that 
\begin{equation}\label{eq:implied-by-pseudo}
    \left|\text{Tr}((\rho_{OB} - \sigma_{OB})P_i)\right| \leq 2^{-\alpha |H|}
\end{equation}
for any Pauli operator $P_i$, where $\sigma_{OB}$ is the state that results when the state $\rho_{EB}$ measured by the observer is replaced by the maximally mixed state; see Figure~\ref{fig:complete_picture}. 

To understand why equation~(\ref{eq:implied-by-pseudo}) follows from pseudorandomness, note that we are modeling a measurement of $EB$ by the observer $O$ as a low-complexity unitary interaction between $EB$ and $O$, followed by a simple measurement of the $O$ register. Strictly speaking, then, we should allow the Pauli operator $P_i$ to act only on $O$, not on $OB$. In effect, we are assuming that the observer's quantum memory contains $|OB|$ qubits rather than $|O|$ qubits, so that measuring a Pauli operator acting on $OB$ is permitted. In our formulation of the pseudorandomness assumption, there is no restriction on the size of the observer's memory, only on the complexity of his operation. Therefore, assuming that the state of $EB$ is pseudorandom, the observer's measurement will not distinguish $\rho_{OB}$ from $\sigma_{OB}$ even if the observer is permitted to measure $B$ as well as $O$.

Using the completeness and orthogonality of the Pauli operators, we can bound the Frobenius distance between the two states as
\begin{align}
    \|\rho_{OB} - \sigma_{OB}\|_F^2&=\mathrm{Tr}\left((\rho_{OB} - \sigma_{OB})^2\right)\\ 
    &= 2^{-(|OB|)}\sum_{i} \left|\text{Tr}((\rho_{OB} - \sigma_{OB})P_i)\right|^2.
\end{align}
Because there are $4^{|OB|}$ Pauli operators, the right hand side is bounded by $2^{-2\alpha |H|} 2^{|OB|}$. The trace distance is bounded by the Frobenius norm as
\begin{align}
    \|\rho \|_1 \le \sqrt{\mathrm{rank}(\rho)}\,\|\rho\|_F,
\end{align}
for any operator $\rho$. Therefore we have
\begin{equation}
    \|\rho_{OB} - \sigma_{OB} \|_1 \leq  2^{|OB|}2^{-\alpha|H|},  \label{eq:decoupling_elementary}
\end{equation}
because the rank of $\rho_{OB}$ can be no larger than $2^{|OB|}$.
From~\eqref{eq:decoupling_elementary},
%and the triangle inequality, and recalling the $\sigma_B$ is maximally mixed, 
one finds that 
\begin{equation}\label{eq:decoupling-derive}
    \begin{aligned}
    \|\rho_{OB} - \rho_O \otimes \rho_B \|_1 &\leq \|\rho_{OB} - \sigma_{OB} \|_1 + \|\sigma_{OB} - \rho_{O}\otimes \rho_B \|_1 \\
    &\leq 2^{|OB|-\alpha |H|} + \|\sigma_{OB} - \rho_O \otimes \rho_B \|_1 \\
    &=2^{|OB|-\alpha |H|} + \|\sigma_{O} \otimes \sigma_B - \rho_O \otimes \rho_B \|_1 \\
    &\leq 2^{|OB|-\alpha |H|} + \|\sigma_O -\rho_O \|_1 + \|\sigma_B - \rho_B \|_1 \\
    &\leq 3 \times 2^{(|OB|-\alpha |H|)}.
    \end{aligned}
\end{equation}
The first line is the triangle inequality. From the first line to the second line, we used \eqref{eq:decoupling_elementary}. From the second line to the third line we used the fact that $\sigma_{OB}$ is a product state over $O$ and $B$. From the third line to the fourth line, we used the fact that 
\begin{equation}
\begin{aligned}\label{eq:product-fact}
    \|\sigma_O \otimes \sigma_B - \rho_O \otimes \rho_B \|_1 &\le \|(\sigma_O-\rho_O)\otimes \sigma_B \|_1 + \|\rho_O \otimes (\sigma_B - \rho_B)\|_1 \\
    &\leq \|\sigma_O-\rho_O \|_1 + \|\sigma_B - \rho_B \|_1,
\end{aligned}
\end{equation}
where the first line of \eqref{eq:product-fact} follows from the triangle inequality, and the second from the property that tracing out a subsystem cannot increase the trace distance. 
%    \jp{How does one obtain the first equality in \eqref{eq:product-fact}?} \textcolor{blue}{(ET) I'm not sure it's an equality. It might be a typo. I've changed it to an inequality, which follows by adding and subtracting $\rho_O\otimes \sigma_B$ and applying the triangle inequality.}
To reach the last line of \eqref{eq:decoupling-derive}, we again used the property that tracing out a subsystem cannot increase the trace distance.
Finally, in the case where $B$ is a single qubit, so that $|OB| = |O| + 1$, \eqref{eq:decoupling-derive} becomes the decoupling bound~\eqref{eq:decoupling_main}. More generally, decoupling is satisfied whenever $|OB| \ll \alpha |H|$.
%recall, though, that here we have assumed that the observer has a memory of size $|OB|$, which must be at least as large as $|B|$.

%\jp{It's not clear to me why (\ref{eq:decoupling_main}) follows from (\ref{eq:decoupling_elementary}), and it is also not clear why we need (\ref{eq:decoupling_main}). Won't (\ref{eq:decoupling_elementary}) suffice for our purposes?} \textcolor{violet}{(IK) I think (\ref{eq:decoupling_elementary}) is equivalent to (\ref{eq:decoupling_main}).} \jp{Can we add a sentence explaining why?} \textcolor{violet}{(IK) Sorry, they are not exactly the same. Now I added an actual derivation. There is an extra multiplicative factor of $3$. I changed \eqref{eq:decoupling_main} accordingly.} \textcolor{blue}{(ET) For the most part, \ref{eq:decoupling_elementary} does suffice. But I think \ref{eq:decoupling_main} is more natural and easier to remember, and the slightly weaker bound doesn't really affect much. A note to self that we need to change all the bounds affected by the extra factor of $3$.}

\begin{figure}[ht]
    \centering
    \includegraphics[width=0.65\columnwidth]{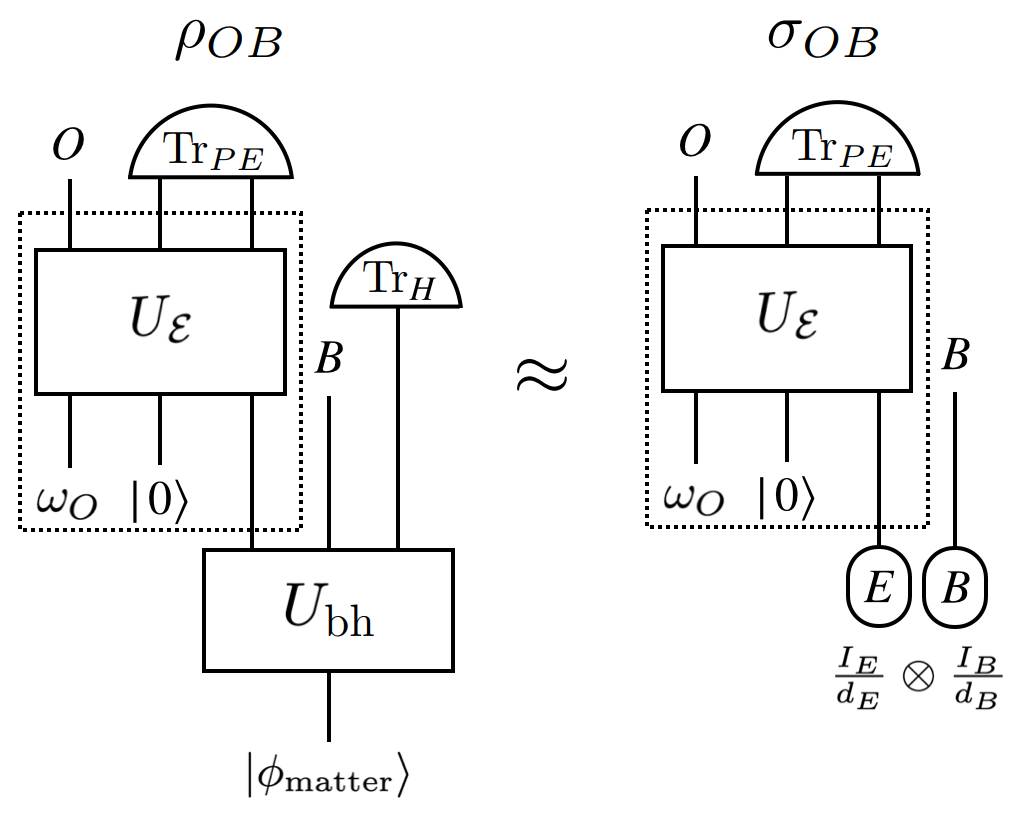}
    \caption{Graphical depiction of the decoupling bound, which follows from the pseudorandomness of the Hawking radiation emitted by an old black hole. On the left, a black hole forms from collapse and partially evaporates; the emitted radiation is $EB$ and the remaining black hole is $H$. Then an observer $O$ and probe $P$ interact with the radiation subsystem $E$ for a time that scales polynomially with the initial black hole entropy $S_{\text{bh}}$. On the right, the unitary transformation describing the interaction of $OPE$ is the same as on the left, but the state of the Hawking radiation is replaced by a maximally mixed state of $EB$. The decoupling bound asserts that the final state of $OB$ is the same in both cases, up to an error that is exponentially small in $|H|$, the size of the remaining black hole, provided that $|O| \ll |H|$.
    %A black hole is formed from an initial state $|\phi_{\text{matter}}\rangle$. The radiation interacts with an observer ($O$) and a probe ($P$), whose joint state is initialized in $\omega_O \otimes |0\rangle\langle 0|_P$. The dotted box represents a unitary process whose complexity scales at most polynomially with $S_{\text{bh}}$. On the right side,  $\sigma_{OB}$ is created by replacing the reduced density matrix of $U_{\text{bh}}|\phi_{\text{matter}}\rangle$ over $EB$ by a maximally mixed state. In particular, $B$ and $O$ are decoupled from each other. These statements hold approximately, and the approximation error is exponentially small in $|H|$ provided that $|O| \ll |H|$.
    }
    \label{fig:complete_picture}
\end{figure}

Because two states close in trace distance cannot be distinguished well by any measurement, the decoupling bound implies that the state $\rho_{OB}$ cannot be distinguished from the state $\sigma_{OB}$ assuming that $|O| \ll |H|$. 
%%
%Note that by definition $\sigma_{OB}$ is factorized between $O$ and $B$. Therefore, $\rho_{OB}$ is approximately factorized between $O$ and $B$, provided that $|O| \ll |H|$. 
%%%
We thus conclude that any subsystem small compared to the remaining black hole $H$, even after interacting with the early radiation $E$, cannot be correlated with $B$; see Figure~\ref{fig:complete_picture}. In particular, an observer outside the black hole who interacts with $E$ for a polynomially bounded time remains decoupled from $B$, assuming that the Hawking radiation is pseudorandom. 

This conclusion about the hardness of decoding follows from the pseudorandomness assumption for any computationally bounded observer who can access only system $E$. However, the decoding becomes easy if the observer has access to both $E$ and $H$, as long as the state $|\Psi\rangle_{EBH}$ has polynomial complexity. For this case, we will describe an explicit decoding protocol in Section \ref{sec:efficient-manipulation}.

\section{Black hole as a quantum error-correcting code \label{sec:bh_qecc}}
In this section, we recast the findings in Section~\ref{sec:pr_decoupling} in the language of quantum error correction. The quantum error correction point of view will prove to be useful in understanding more subtle thought experiments studied in Section~\ref{sec:qec}. We will see that an old black hole, together with its previously emitted Hawking radiation, is a quantum error-correcting code with exotic properties that have not been noted in previous discussions of holographic quantum error-correcting codes~\cite{Almheiri2015,Pastawski2015}. These properties hold if the Hawking radiation is pseudorandom. That a black hole can be viewed as a quantum error-correcting code is not new~\cite{Pastawski2015,Hayden2019,Verlinde2013,Yoshida2019}. What's new is that a black hole can protect quantum information against seemingly pernicious errors; we refer to these as ``low-complexity errors,'' meaning errors inflicted by a malicious agent who performs a quantum computation on the Hawking radiation with complexity scaling polynomially in the size of the remaining black hole. 

To explain this claim, it is useful to view the state of the black hole and the radiation as an encoding map from the interior mode $\tilde{B}$ into $EH$. That is,  $|\Psi\rangle_{EHB}$ defines an isometric embedding of $\tilde B$ into $EH$. Recall that $E$ denotes the early radiation, $H$ denotes the remaining black hole, and $B$ denotes a late outgoing mode. For simplicity, we assume that $B$ is a single qubit, but the following results remain essentially unchanged for $B$ of any constant size (small compared to $H$). The encoded system $\tilde B$ describes the mode in the black hole interior that is entangled with $B$.

We can define an (approximate) isometric embedding $V_{\Psi}:\mathcal{H}_{\tilde{B}}\rightarrow \mathcal{H}_{EH}$ of a single qubit $\tilde{B}$ into the subspace $EH$ by
\begin{equation}
    V_{\Psi}|i\rangle_{\tilde{B}} = 2(I_{EH}\otimes\langle\omega|_{B\tilde{B}})(|\Psi\rangle_{EHB}\otimes |i\rangle_{\tilde{B}}),\label{eq:embedding}
\end{equation}
where $|\omega\rangle_{B\tilde{B}}=2^{-1/2}(|00\rangle_{B\tilde{B}} + |11\rangle_{B\tilde{B}})$ denotes an EPR pair on $B\tilde{B}$; see Figure~\ref{fig:encoding}. While $V_\Psi$ itself is not precisely an isometric embedding, it is exponentially close to one under the assumption that $|\Psi\rangle_{EBH}$ is pseudorandom on the exterior system $EB$, as specified in Definition \ref{def:pseudorandom}. In Appendix \ref{sec:embedding}, we show that there exists an isometric embedding $V$ such that
\begin{align}\label{eq:embedding-error}
    \|V - V_{\Psi}\| \le 2\cdot 2^{-\alpha|H|},
\end{align}
where $\|\cdot\|$ denotes the operator norm.
The isometry $V$ then defines a code subspace that encodes $\tilde{B}$. For macroscopic observers (\emph{i.e.}, $|O|\gg 1$), the error in~\eqref{eq:embedding-error} is negligible compared to the error in the decoupling bound~\eqref{eq:decoupling_main}. 
%\jp{Is it necessary to remark on the error here being expressed in the operator norm, versus the trace norm in the decoupling bound?} \textcolor{blue}{(ET) Maybe not strictly necessary, but I think reminding the readers of notation which hasn't appeared in a while might generally be a good idea. But certainly the bound/notation is very standard.} \jp{I meant to ask about the logic, not the notation. We are comparing the errors in \eqref{eq:embedding-error} and \eqref{eq:decoupling_main}. This seems a little bit misleading since the norms are different in the two equations. Maybe we should explain why this is a fair comparison.} \textcolor{blue}{Using the operator norm is mostly out of convenience I think. In the proof that $V_\Psi$ is an approximate embedding, we naturally bound the largest singular value. The two norms will differ by a constant factor (the dimension of $B$), so the two are mostly equivalent for our purposes.} 
Although the norm in equation~(\ref{eq:embedding-error}) is the operator norm rather than the trace norm, that distinction need not concern us if $|B|$ is sufficiently small compared to $|H|$.
Therefore we can ignore any differences between $V$ and $V_{\Psi}$ and use them interchangeably.

\begin{figure}[ht]
\centering
\includegraphics[width=0.5\columnwidth]{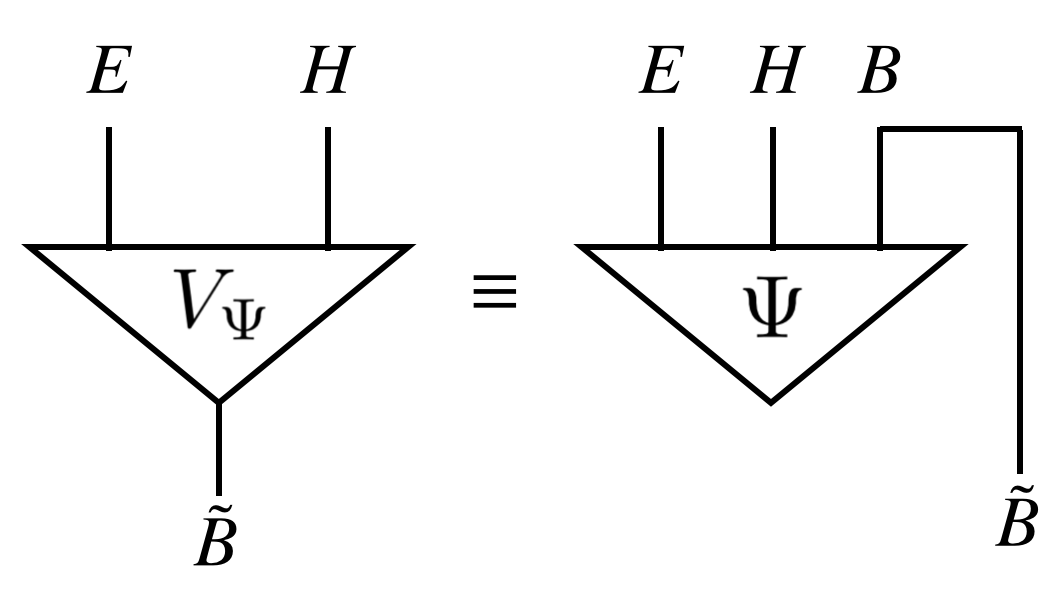}
\caption{The definition of the encoding of $\tilde{B}$ into $EH$, with $\Psi$ defined as in  Figure~\ref{fig:psi}.}
\label{fig:encoding}
\end{figure}

We will now show that the isometry $V_{\Psi}:\mathcal{H}_{\tilde{B}} \rightarrow \mathcal{H}_{EH}$ defined above embeds $\tilde{B}$ into $EH$ as a code subspace for which any low-complexity noise model acting on $E$ is (approximately) correctable. By low-complexity error, we mean that the unitary process $U_{\mathcal{E}}$ in Figure~\ref{fig:bh_code} has complexity at most polynomial in $|H|$. Here the external observer $O$ plays the role of the ``environment" for the noise process acting on $E$ and the probe $P$.
\begin{figure}[ht]
    \centering
    \includegraphics[width=0.3\columnwidth]{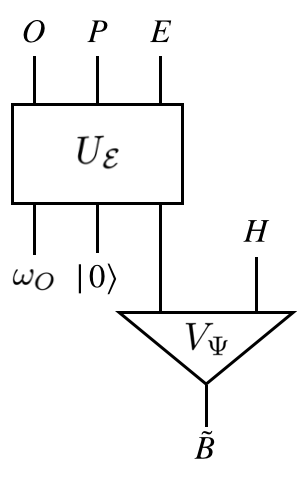}
    \caption{A black hole can be viewed as a quantum error-correcting code. By tracing out the observer $O$, we obtain a ``noise model" $\mathcal{E}$ on the early radiation and the probe.}
    \label{fig:bh_code}
\end{figure}

%As in Section~\ref{sec:pr_decoupling}, let us ignore the probe in Figure~\ref{fig:bh_code} at the moment. The observer is initialized in a state $\omega_O$ and then interacts with the radiation(without interacting with the remaining black hole). This is precisely what the environment does in the standard theory of quantum error correction. Any noise model can be thought as a unitary process between a system and an environment which is traced out at the end of the interaction. In Figure~\ref{fig:bh_code}, by tracing out the observer, we can obtain an effective error model that is induced on the radiation.

The error model depicted in Figure~\ref{fig:bh_code} is rather exotic compared to error models that are typically considered in discussions of quantum gravity and fault-tolerant quantum computing. For example, one widely studied error model is the ``erasure model,'' wherein each qubit may be removed with some probability, and we know which qubits are removed. The performance of quantum codes against erasure errors arises, in particular, in studies of the holographic AdS/CFT dictionary ~\cite{Almheiri2015,Pastawski2015}; if a logical bulk operator can be ``reconstucted'' on a portion of the boundary, that means that erasure of the complementary portion of the boundary is correctable for that logical operator. By the no-cloning theorem, no code can tolerate erasure of more than 50\% of the qubits in the code block. In contrast, in our setup, erasure is correctable even if most of the qubits are removed. The catch is that the erased qubits must lie in $E$; removal of qubits in $H$ is not allowed. 
%the fraction of qubits that are strongly affected by the error can approach $1$. This is because $E$ can constitute, say, more than $99\%$ of the initial black hole entropy. 

In studies of fault-tolerant quantum computing, the noise afflicting the physical qubits is usually assumed to be weak and weakly correlated. In a Hamiltonian formulation of the noise model, this means that each qubit in the computer is weakly coupled to a shared environment \cite{preskill2013sufficient}. In contrast, for the noise model described by $U_{\mathcal{E}}$, the noise may act strongly on all the qubits in $E$; the only restriction is that the noise has quantum complexity scaling polynomially with $|H|$. Furthermore, how the noise acts depends on the initial state $\omega_O$ of the observer $O$, which may be chosen adversarially. Again, what makes successful error correction possible is that the subsystem $H$ is assumed to be noiseless, an assumption that would be unrealistic for typical quantum computing hardware. 

Codes that can protect against this malicious typle of noise are central to our proposed resolution of the firewall paradox. An old black hole provides such a code if its previously emitted radiation is pseudorandom. The code corrects errors successfully if the noise acting on $E$ has low complexity and the remaining black hole $H$ is noiseless, provided that the observer $O$ is small compared to $H$. 

%This is possible because there are two additional constraints on the observer: The observer is significantly smaller than $|H|$, and the joint evolution of the observer and the radiation must be an efficient unitary. These two constraints place non-trivial limitations on the observer's actions. From these limitations, we conclude that $V_{\Psi}$ can protect quantum information from the observer.  

\subsection{Correcting low-complexity errors}
For simplicity, we will first consider a scenario without the probe $P$ shown in Figure~\ref{fig:bh_code}. We will see that the error applied to the radiation system $E$ is (approximately) correctable. In Section~\ref{sec:probe} we will explain how our conclusion changes when the probe $P$ is included.

A central result in the theory of quantum error correction is the \emph{information-disturbance relation}, which states that a code can protect quantum information from noise if and only if the ``environment'' of the noise channel $\mathcal{E}$ learns nothing about the logical information. More precisely, there is a physical process $\mathcal{R}$, the recovery process, which reverses the error:
\begin{equation}
    \mathcal{R} \circ \mathcal{E} \approx \mathcal{I}, \label{eq:recovery}
\end{equation}
where $\mathcal{I}$ is the identity operation, if and only if the ``reference system'' that purifies the  quantum error-correcting code decouples from the environment . In Figure~\ref{fig:bh_code} (neglecting the probe $P$), the environment of the noise channel $\mathcal{E}$ acting on $E$ is the observer $O$, and $B$ is the reference system that purifies the encoded interior mode $\tilde B$. Therefore the necessary and sufficient condition for (approximate)  correctability is the (approximate) decoupling of $B$ and $O$,
\begin{equation}
    \rho_{OB} \approx \rho_{O} \otimes \rho_{B} \label{eq:decoupled_OB}
\end{equation}
where $\rho_{OB}$ is the reduced density operator for $OB$. Here the approximation errors of equation~\eqref{eq:recovery} and equation~\eqref{eq:decoupled_OB} are related to each other by a constant factor. Therefore, using the decoupling bound~\eqref{eq:decoupling_main}, we can conclude that there exists a recovery process $\mathcal{R}$ that reverses $\mathcal{E}$ up to an error exponentially small in $|H|$, as long as $|O| \ll |H|$ and assuming that the Hawking radiation is pseudorandom. 

Various formal statements that imply the existence of $\mathcal{R}$ in equation~\eqref{eq:recovery} are known; for the reader's convenience, we reproduce some of these results below. First, let us properly define what it means for a channel to be approximately correctable with respect to some code subspace --- a more comprehensive discussion can be found in \cite{BenyOreshkov10}. Let $S(\mathcal{H})$ denote the set of states on a Hilbert space $\mathcal{H}$. Suppose that we are given channels $\mathcal{E}, \mathcal{N}: S(\mathcal{H})\rightarrow S(\mathcal{H})$. Fixing a state $\rho \in S(\mathcal{H})$, we define the \emph{entanglement fidelity} between $\mathcal{E}$ and $\mathcal{N}$ with respect to $\rho$ to be
\begin{align}
    F_\rho(\mathcal{E},\mathcal{N}) = f\left[(\mathcal{E}\otimes \mathcal{I})(|\psi\rangle\langle\psi|),(\mathcal{N}\otimes \mathcal{I})(|\psi\rangle\langle\psi|)\right],
\end{align}
where $|\psi\rangle$ is a purification of $\rho$, and where
\begin{align}
    f(\rho,\tau) = \mathrm{Tr}\left(\sqrt{\sqrt{\tau}\rho\sqrt{\tau}}\right)
\end{align}
is the usual fidelity between states $\rho$ and $\tau$. To quantify the  closeness of two channels, we use the worst-case entanglement fidelity to define  the \emph{Bures distance}, given by
\begin{align}\label{eq:bures-define}
    \mathfrak{B}(\mathcal{E},\mathcal{N}) = \max_{\rho}\sqrt{1-F_\rho(\mathcal{E},\mathcal{N})};
\end{align}
we sometimes define a more restricted notion of the Bures distance, where we maximize over states in some specified subspace. In discussions of error correction, we say that a noise channel $\mathcal{E}$ is $\epsilon$-correctable with respect to a code subspace $\mathcal{C}\subseteq \mathcal{H}$ if there exists a recovery channel $\mathcal{R}$ such that
\begin{align}\label{eq:bures-epsilon}
    \mathfrak{B}(\mathcal{R}\circ\mathcal{E},\mathcal{I}) \le \epsilon,
\end{align}
where the maximization in the Bures metric is over all code states $\rho$ with support on $\mathcal{C}$. 

The Bures metric is bounded above and below by the trace norm as
\begin{align}
    2\mathfrak{B}^2(\mathcal{E},\mathcal{N}) &\le \max_{\rho}\left\|(\mathcal{E}\otimes \mathcal{I})(|\psi\rangle\langle\psi|) - (\mathcal{N}\otimes \mathcal{I})(|\psi\rangle\langle\psi|)\right\|_1\le 2\sqrt{2}\mathcal{B}(\mathcal{E},\mathcal{N}).\label{eq:bures_inequality}
\end{align}
The norm in the middle is essentially the diamond-norm distance between the channels $\mathcal{E}$ and $\mathcal{N}$~\cite{MikeIke}, except that for the purpose of characterizing error correction the maximization is over code states only. Applying this inequality and tracing out the purifying system, the $\epsilon$-correctability of a channel $\mathcal{E}$ implies that we have
\begin{equation}
    \max_{\rho} \|(\mathcal{R} \circ \mathcal{E})(\rho) - \rho \|_1 \leq 2\sqrt{2}\epsilon,\label{eq:correctable}
\end{equation}
where again the maximization is over code states.

As mentioned previously, an important result characterizing approximate correctability is the information-disturbance trade-off, which we now state quantitatively. Let $\mathcal{E}:S(\mathcal{H}_A)\rightarrow S(\mathcal{H}_A)$ be a noise channel acting on a system $A$, and let $V:\mathcal{H}_A\rightarrow \mathcal{H}_F\otimes \mathcal{H}_A$ be an isometry which purifies $\mathcal{E}$; i.e.,
\begin{align}
    \mathcal{E}(\rho) = \mathrm{Tr}_{F}(V\rho V^\dagger).
\end{align}
Hence $F$ is the environment of the channel; we have resisted the temptation to denote the environment by $E$ to avoid confusion with our convention that $E$ denotes a subsystem of the Hawking radiation.
Then the \emph{complementary channel} $\widehat{\mathcal{E}}:S(\mathcal{H}_A)\rightarrow S(\mathcal{H}_F)$ is defined by
\begin{align}
    \widehat{\mathcal{E}}(\rho) = \mathrm{Tr}_A(V\rho V^\dagger).
\end{align}
A special case of interest is the identity channel $\mathcal{I}$. Taking the environment to be $1$-dimensional, the complementary channel to the identity channel is simply the (partial) trace
\begin{align}
    \widehat{\mathcal{I}}(\rho) = \mathrm{Tr}_A(\rho).
\end{align}
Then the information-disturbance trade-off states the following:

\begin{theorem}[\cite{BenyOreshkov10}, Theorem 1]\label{thm:decoupling-tradeoff}
Let $\mathcal{C} \subseteq \mathcal{H}_A$ be a code subspace. Let $\mathcal{E}:S(\mathcal{H}_A)\rightarrow S(\mathcal{H}_A)$ be an error channel. Then
\begin{align}
    \inf_{\mathcal{R}}\mathfrak{B}\left(\mathcal{R}\circ \mathcal{E}, \mathcal{I}\right) 
    = \inf_{\mathcal{R}'} \mathfrak{B}\left(\widehat{\mathcal{E}}, \mathcal{R}'\circ \mathrm{Tr}\right),\label{eq:info_dist_tradeoff}
\end{align}
where the infimums are taken over all channels $\mathcal{R}:S(\mathcal{H}_A)\rightarrow S(\mathcal{H}_A)$, and $\mathcal{R}':\mathbb{R}\rightarrow S(\mathcal{H}_F)$.
\end{theorem}

Note that a channel $\mathcal{R}':\mathbb{R}\rightarrow S(\mathcal{H}_F)$ is just state preparation on the channel environment $\mathcal{H}_F$, i.e., every such channel $\mathcal{R}'$ is uniquely identified with a state $\sigma_F \in S(\mathcal{H}_F)$ such that
\begin{align}
    (\mathcal{R}'\circ\mathrm{Tr})(\rho) = \mathrm{Tr}(\rho)\,\sigma_F,
\end{align}
so we can equivalently write
\begin{align}
     \inf_{\mathcal{R}'} \mathfrak{B}\left(\widehat{\mathcal{E}}, \mathcal{R}'\circ \mathrm{Tr}\right) = \inf_{\sigma_F} \mathfrak{B}\left(\widehat{\mathcal{E}}, \sigma_F\otimes \mathrm{Tr}\right).
\end{align}

Now let's see what equation~\eqref{eq:info_dist_tradeoff} tells us in the context of the black hole error-correcting code defined by the (approximate) isometry $V_\Psi$. Let $\tilde{\rho}_{\tilde{B}}$ be a logical state and let $\tilde{\rho}_{\tilde{B}B}$ be a purification. The isometry $V_{\Psi}$ then embeds $\tilde{\rho}_{\tilde{B}B}$ as a (purified) code state $\rho_{EHB}$:
\begin{align}
    \rho_{EHB} = V_{\Psi}\tilde{\rho}_{\tilde{B}B}V^{\dagger}_{\Psi}.
\end{align}
Let $\mathcal{E}:S(\mathcal{H}_{E}) \rightarrow S(\mathcal{H}_{E})$ be an arbitrary channel acting on $E$ such that some purification $U_\mathcal{E}$ of $\mathcal{E}$ has low-complexity (see the set-up described in Figure~\ref{fig:setup}). Let
\begin{align}
    \sigma_{OEHB} = U_{\mathcal{E}}(\omega_O\otimes\rho_{EHB})U^{\dagger}_{\mathcal{E}}\label{eq:post_evol}
\end{align} 
denote the overall post-evolution state. To apply Theorem~\ref{thm:decoupling-tradeoff}, let us consider the error channel $\mathcal{E}\otimes \mathcal{I}_H$. Then the environment of the channel $\mathcal{E}\otimes \mathcal{I}_H$ is the observer subsystem $O$, and the complementary channel $\widehat{\mathcal{E}\otimes \mathcal{I}_H}$ maps $S(\mathcal{H}_{EH})$ to $S(\mathcal{H}_{O})$. From~\eqref{eq:post_evol}, the state obtained from $\rho_{EHB}$ after the application of $\widehat{\mathcal{E}\otimes \mathcal{I}_H}$ is precisely given by
\begin{align}
    \left(\widehat{\mathcal{E}\otimes \mathcal{I}_H}\otimes \mathcal{I}_B\right)(\rho_{EHB}) = \mathrm{Tr}_{EH}\left(\sigma_{OEHB}\right) = \sigma_{OB}.
\end{align}
Since $\sigma_{OEHB}$ was a state obtained through acting on the black hole code state %\IK{Did we define what a black hole code is? We formally explain what black hole code is in Theorem~\ref{thm:controlled_unitary}, but that's \emph{after} this sentence. We need to either remove this word or explain what it is.} \ET{I think it should be OK. We've already defined and talked about the embedding $V$ in equation~\eqref{eq:embedding-error}. I think it should be clear here that we're talking about the code defined through that embedding. We can add a brief reminder if you feel it's necessary.} 
$\rho_{EHB}$ with a low-complexity unitary, it follows by the pseudorandom hypothesis that the decoupling bound~\eqref{eq:decoupling_main} holds. Therefore we have
\begin{align}
    \|\sigma_{OB} - \sigma_O\otimes \sigma_B\|_1 \le 6\cdot 2^{-(\alpha|H|-|O|)}.
\end{align}
Finally, since $U_\mathcal{E}$ is supported away from $B$, we have $\sigma_B = \rho_B$, and so
\begin{align}
    \left\|\left(\widehat{\mathcal{E}\otimes \mathcal{I}_H}\otimes \mathcal{I}_B\right)(\rho_{EHB}) - \sigma_O\otimes \rho_B\right\|_1 \le 6\cdot 2^{-(\alpha|H|-|O|)}.\label{eq:channel_decouple}
\end{align}

%\jp{I got confused here. It seems that you want the complementary channel to map $EH$ to $O$, while I had assumed it maps $E$ to $O$. Or did we intend for the channel input to be $\rho_{EB}$ here rather than $\rho_{EHB}$? Aside from that, I have lost the thread of the logic. Why does this follow from the decoupling bound? Why are we writing $\approx$ instead of writing out the error explicitly. How do we get from here to (7.19)? I'm lost. Can we fill in the intermediate steps? }
%\textcolor{blue}{(ET) Can you see if the preceding discussion is now a bit more comprehensible? I think a figure here might also help a bit, so I will add one later.} \jp{Better, thanks. But shouldn't those be trace norms in the last two equations?} \textcolor{blue}{Yes, that's right.}
%\textcolor{blue}{(ET) Sorry, I can see how it's confusing. For the purposes of applying Theorem~\ref{thm:decoupling-tradeoff}, we need to regard $\mathcal{E}$ as a channel which maps $EH$ to $EH$ (but of course, it acts by identity on $H$), so that the complementary channel maps $EH$ to $O$. I will fill in a bit more detail later and rewrite this part slightly in a bit.}\jp{We could say we are applying the theory to the channel $\mathcal{E}\otimes \mathcal{I}_H.$ I don't want to depart from our use of the symbol $\mathcal{E}$ elsewhere in the paper. } \textcolor{blue}{Yes, I think that's a better idea.}
This holds for all code states, so~\eqref{eq:channel_decouple}, together with the first inequality in~\eqref{eq:bures_inequality}, implies that we have
\begin{align}
    \inf_{\sigma_O}\mathfrak{B}\left(\widehat{\mathcal{E}\otimes \mathcal{I}_H}, \sigma_O\circ \mathrm{Tr}_{EH}\right) \le \sqrt{3}\cdot 2^{-(\alpha|H|-|O|)/2}.
\end{align}
Therefore, the channel $\mathcal{E}$ is approximately correctable by Theorem~\ref{thm:decoupling-tradeoff}. We state this as a Lemma.

\begin{lemma}\label{lem:correctable}
Let $V_{\Psi}$ be the approximate isometric embedding defined by the state $\Psi_{EHB}$. Let $\mathcal{E}$ be an error channel on $E$ with purification $U_{\mathcal{E}}$. Suppose that the decoupling bound~\eqref{eq:decoupling_main} holds. Then $\mathcal{E}$ is $\epsilon$-correctable for $V_{\Psi}$, where
\begin{align}
    \epsilon = \sqrt{3}\cdot 2^{-(\alpha|H| - |O|)/2},
\end{align}
if $B$ is a single qubit. For general $|B|$, we have
\begin{align}
    \epsilon = \sqrt{\frac{3}{2}}\cdot 2^{-(\alpha|H| - |OB|)/2}.
\end{align}
\end{lemma}

\noindent 
Note that the recovery operator $\mathcal{R}$ acts on $EH$ rather than $E$. The same will be true for the ghost logical operators we construct in Section~\ref{sec:qec}.

\subsection{Including the probe\label{sec:probe}}
%\jp{This subsection seemed to me rather long and rambling. I've reorganized it and condensed it a bit.}

We would now like to consider a modified scenario in which both the observer $O$ and a probe $P$ interact with the Hawking radiation system $E$, as indicated in Figure~\ref{fig:setup}. We cannot simply absorb $P$ into $O$, because we will continue to insist that $O$ is small compared to $H$, while we wish to allow $P$ to be comparable to $H$ in size, or even larger. In this modified scenario, the unitary purification $U_\mathcal{E}$ of the noise model acts on $OPE$ rather than $OE$. This change does not alter the conclusion that $O$ and $B$ decouple if $U_\mathcal{E}$ has low complexity. Therefore, just as before, there is a recovery map that reverses the effect of the noise on the encoded state. What changes is that now the recovery map acts on $PEH$ rather than $EH$.

%By including the probe, we can ask a more refined set of questions. Compared to the previous analysis, the key difference is that the ``noise model" $\mathcal{E}$ maps a state in $EH$ to a state in $EHP$. Equivalently, one may view the noise model as a map from $EHP$ to a state in $EHP$ wherein the state of the probe is initialized to some fixed state. \jp{Don't call it $\mathcal{E}$ when you really mean $\mathcal{E}\otimes \mathcal{I}_H$.}

%What remains to be true is that $O$ and $B$ are decoupled from each other. Therefore, there is a recovery map that reverses the effect of $\mathcal{E}$. In fact, one may even consider a subsystem that is somewhat scrambled across $PE$ by some unitary $U'$. As long as the complexity of this unitary is modest, we can absorb this unitary into $U_{\mathcal{E}}$ to ``unscramble" this subsystem. Then we would be able to apply the same argument to conclude that this scrambled subsystem is decoupled from $B$. \jp{I'm not following. I don't know what you mean by ``consider a subsystem that is somewhat scrambled across $PE$ \dots'' Consider it for what purpose or in what context?}

We emphasize that if the probe $P$ is sufficiently large, then $P$ need not decouple from $B$, even if $U_\mathcal{E}$ has low complexity.  To understand why not, suppose $P$ has the same size as the system $E$ and that the channel $\mathcal{E}$ swaps $P$ and $E$. Before this swap, $B$ is entangled with the code space embedded in $EH$; therefore after the swap (a low-complexity operation), $B$ is entangled with $PH$. More realistically, we might imagine that $P$ is a cloud of dust surrounding the black hole, and that $|P|\gg |E|$. After the dust interacts with the Hawking radiation, the encoding of $\tilde B$ will be modified, so that $B$ is entangled with a code subspace of $PEH$ rather than a subspace of $EH$~\cite{Bousso2014}. 

However, any subsystem of $OP$ which is small compared to $H$ will decouple from $B$, as long as $U_{\mathcal{E}}$ has low complexity, and assuming that the Hawking radiation is pseudorandom. The only way to distill the encoded state into a small subsystem is to perform a high complexity operation. Hence, if only low-complexity operations are allowed, we need not worry about a scenario in which the encoded version of $\tilde B$ outside the horizon is decoded into a small system, and then falls into the black hole to meet its twin in the interior.
This is essentially the observation of Harlow and Hayden~\cite{Harlow2013}, later extended by Aaronson \cite{Aaronson2016}. Our analysis goes further by clarifying that the encoded state is hard to distill even when the remaining black hole $H$ is much smaller than $E$, as long as $H$ is macroscopic and assuming that the Hawking radiation is pseudorandom. %It has also been suggested that

One might wonder whether the encoded mode can be easily extracted if the probe system $P$ is prepared in a carefully chosen state~\cite{Oppenheim2014}. Our conclusion is that any such initial state of $P$ would need to have exponential complexity, an unlikely property for the dust surrounding an evaporating black hole. One might also ask what happens if all the qubits in the early radiation system $E$ are measured in the standard basis by the observer. Surely this \emph{would} disrupt the encoded interior of the black hole. But in our model the number of radiation qubits that can be measured is limited by the size $|O|$ of the observer's memory, and the interior will stay well protected as long as $|O|$ is much smaller than $|H|$.

It is also instructive to view the system $O$ in a different way. Up to now we have regarded $O$ as a potentially malicious agent who attempts to damage the encoded interior of the black hole by acting on its exterior. More prosaically, we can think of $O$ as an abstract purifying space which is introduced for convenience so that we can describe the noise channel $\mathcal{E}$ using its purification, the unitary transformation $U_{\mathcal{E}}$. From that point of view, limiting the size $|O|$ of the ``observer'' $O$ is just a convenient way of restricting the form of the quantum channel $\mathcal{E}$. Specifically, the rank of the marginal density operator $\rho_O$ after $U_{\mathcal{E}}$ is applied is called the Kraus rank (or simply the rank) of the channel $\mathcal{E}$. This rank can be no larger than the dimension of system $O$, namely $2^{|O|}$, which we have assumed to be small compared to the dimension $2^{|H|}$ of the Hilbert space of black hole microstates.
Thus our conclusion can be restated: If the Hawking radiation is pseudorandom and $H$ is macroscopic, then the quantum error-correcting code protects the encoded version of $\tilde B$ against any noise channel acting on $PE$ that has both low complexity \emph{and} low rank.

%We may also consider an alternative point of view. So far, we have been considering the subsystem $O$ as representing a physical observer which interacts with the radiation through some unitary $U_{\mathcal{E}}$. Alternatively, we may interpret the subsystem $O$ by regarding it, not as a physical observer, but rather just an abstract purifying space for some channel $\mathcal{E}$ which is applied on the radiation (i.e., as an ancillary system indexing the Kraus operators of the channel). In this sense, the size of $O$ is not a limitation on the size of any physical observer interacting with the radiation, but rather a limitation on the Kraus rank of the operations that the black hole code can protect against; in this interpretation, the black hole code will be protected against \emph{low-rank, low-complexity} channels acting on the radiation and probe. 

%Note that considering $O$ as an abstract purifying system may be advantageous in several ways. Firstly, it is closer in spirit to the traditional setup considered in the high-energy literature, wherein an operation is simply applied without regard for the underlying physical mechanism in which it is achieved. Viewed this way, our limitation of low-rank channels is actually not very restrictive, since all of the operations commonly considered in the literature involve operators (i.e., rank-$1$ channels). It is unclear whether the traditional firewall paradox should extend to highly complicated channels. 

An advantage of this viewpoint is that one might otherwise be misled into interpreting $|O|$ as the physical size of an actual observer. More accurately, it can be regarded as the effective size of the quantum memory of a physical object. This distinction is significant. For an object of specified mass, the largest possible quantum memory is achieved by a black hole of that mass, but the memory size of a quantum computer typically falls far short of that optimal value, because most of its mass is locked into the rest mass of atomic nuclei and unavailable for information processing purposes. Furthermore, the mass per unit volume of a typical quantum computer is far smaller than a black hole's. Therefore it is reasonable to expect that the effective Hilbert space dimension of system $O$ (and hence the Kraus rank of the channel $\mathcal{E}$) is far smaller than the Hilbert space dimension of a black hole with the same circumference as system $O$.

Up until now we have mostly focused on the hardness of decoding the black hole interior mode by acting on the Hawking radiation outside the black hole, concluding that distilling the encoded system to a small quantum memory is computationally hard if the remaining black hole is macroscopic. In section \ref{sec:qec} we will turn to a more subtle question: Can a low-complexity operation acting on the Hawking radiation system $E$ create an excitation near the black hole horizon that could be detected by an infalling observer who falls into the black hole? Here too we will argue that the answer is no. This is a nontrivial extension beyond what we have found so far --- on the face of it, perturbing a quantum state is a far easier task than depositing the state in a compact quantum memory. 

Bousso emphasized that if the interior mode $\tilde B$ is encoded in $EH$, and if effective quantum field theory on curved spacetime is a good approximation in regions of low curvature, then the vacuum near the black hole horizon would need to be ``frozen'' ~\cite{Bousso2014}. That is, neither a small agent $O$ acting on $E$ nor a large probe $P$ interacting with $E$ could disrupt the entanglement of $\tilde B$ with $B$ and hence create an excitation localized near the horizon. We agree with this conclusion, provided that $|H| \gg 1$ and that the interactions of $OP$ with $E$ have quantum complexity scaling polynomially with $|H|$. Interactions with the large probe may alter how the black hole interior is encoded in the radiation and probe, but they do not disrupt the frozen vacuum. 

Once $|H|$ is O(1), large corrections to effective field theory may be expected. Furthermore, the semiclassical structure of spacetime may no longer be applicable in the regime where operations of superpolynomial complexity are allowed; these high-complexity operations could tear spacetime apart. In particular, our expectation that an agent acting on $E$ should be unable to influence the black hole interior might be flagrantly violated if the agent can perform high-complexity operations. We should grow accustomed to the notion that for effective field theory to be an accurate approximation, we require not only geometry with low curvature and states with low energy, but also operations with low complexity and low Kraus rank.

To investigate whether the semiclassical causal structure is robust with respect to low-complexity operations we will need to develop some additional formalism, specifically the theory of \emph{ghost logical operators}; in the context of an old black hole, these may be viewed as operators which act on the black hole interior. We would like to understand, given that the interior is encoded in the Hawking radiation outside the black hole, why low-complexity operations acting on the Hawking radiation produce no detectable excitations inside the black hole. We turn to that task next. 

%Are we done then? Not quite, because there is a more subtle thought experiment that the observer can do. He may be content with simply creating a firewall without extracting the interior mode. The observer may hope to achieve this goal by applying an operator or even a family of operators to the radiation prior to jumping into the black hole. After all, perturbing quantum information seems to be much eaiser than extracting quantum information in general! Such a thought experiment involves an observer applying an operator of his choice to the radiation and then jumping into the black hole to see if he observes the vacuum or not. The challenge is to come up with a definition of the operator acting on the interior mode whose measurement statistics is consistent with what is expected from the vacuum, independent of what operation the observer chooses to apply. In order to study this problem, we develop a theory of \emph{ghost logical operator} below.

\section{Theory of ghost logical operators}\label{sec:qec}
%\jp{In preceding sections I showed my edits in red, because I had previously edited these sections, and wanted to highlight further changes. I won't be marking all the edits in Section 8, only those that I especially want to call to your attention.}

So far, we have argued that the late radiation system $B$ remains decoupled from any sufficiently small subsystem of the early radiation $E$ and the probe $P$, when the observer $O$ performs a low-complexity operation on $EP$. Therefore an infalling observer with reasonable computational power is prevented from extracting the encoded interior mode before jumping into the black hole. But what if the observer settles for the seemingly easier task of disrupting the interior rather than decoding it? In this section we will show that an algebra of \emph{ghost logical operators} can be constructed acting on the interior mode, with the property that low-complexity operations performed outside the black hole nearly commute with the ghost algebra. Hence, if these ghost operators are regarded as operations that can be performed by an observer inside the black hole, we may conclude that the interior is well protected against the actions of malicious agents outside the black hole. 

%However, instead of extracting the interior mode, the infalling observer may decide to simply perturb the mode. 

Following arguments from~\cite{Almheiri2013a}, consider an operator $T$ which acts on the interior mode. Because the corresponding encoded operator acting on the Hawking radiation is highly scrambled, the commutator of this encoded operator with a generic simple operator acting on the radiation has no reason to be small. It seems, then, that an external observer should be able to perturb the interior mode easily \cite{Almheiri2013a,Bousso2014}. Can this conclusion be evaded by constructing the encoded operators suitably? For two-side black holes in AdS/CFT, 
Papadodimas and Raju argued that ``mirror operators'' with the desired properties can be constructed ~\cite{Papadodimas2013,Papadodimas2016}, but no satisfactory construction is known for evaporating black holes.% \jp{Is this true?}. \IK{I think so. In 2018, I asked Raju how their construction works in evaporating black holes. He said that it's unclear. (But there may have been new work since then that I am not aware of.)}

Within our simple toy model of evaporating black hole, we can construct analogues of the mirror operators. Assume that the decoupling bound~\eqref{eq:decoupling_main} holds. Then, as we will see, for every operator $\tilde{T}_B$ acting on some outgoing mode $B$, there exists a ``mirror operator" $T_{EH}$ acting on $EH$ which satisfies the following conditions:
\begin{equation}
\begin{aligned}
    &\tilde{T}_{B}|\Psi\rangle \approx T_{EH}|\Psi \rangle, \\
    &[T_{EH}, E_a]|\Psi\rangle \approx 0,
\end{aligned}\label{eq:mirror}
\end{equation}
where $\{E_a\}$ is a set of operators that a computationally bounded external observer can apply on the radiation, and $|\Psi\rangle$ is the state of the radiation and the black hole. The equations~\eqref{eq:mirror} hold up to an error exponentially small in $|H|$. The first line implies that one can (in principle if not in practice) certify entanglement between an outgoing radiation mode and an abstract subsystem specified by the operators $\{T_{EH}\}$.\footnote{For example, one could perform Bell tests using the Pauli operators acting on $B$ and its mirror.} Therefore, these operators satisfy the right measurement statistics expected for sensibly defined interior operators. The second line implies that these operators approximately commute with all the operators that the external observer can apply. The fact that $T_{EH}$ commutes with $\{E_a\}$ holds as an operator equation on all the states in the code subspace. Therefore, the subsystem specified by the mirror operators $\{T_{EH}\}$ is fully entangled with the late outgoing radiation modes while also being effectively ``space-like separated'' from the external observer. That is, the external observer can disrupt the semiclassical causal structure of the black hole only by applying operations of superpolynomial complexity to the radiation. 

In our construction, it is important to properly characterize the set $\{E_a\}$ of operators that the exterior observer can apply to the radiation. If we view the observer, the black hole, and the exterior radiation as a closed system, we ought to model the entire evolution as a unitary process. In order to enforce the unitarity of this process, the operator applied by the observer to the radiation should depend on the initial state of the observer, as in Figure~\ref{fig:action_radiation}.
\begin{figure}[ht]
    \centering
    \includegraphics[width=0.25\columnwidth]{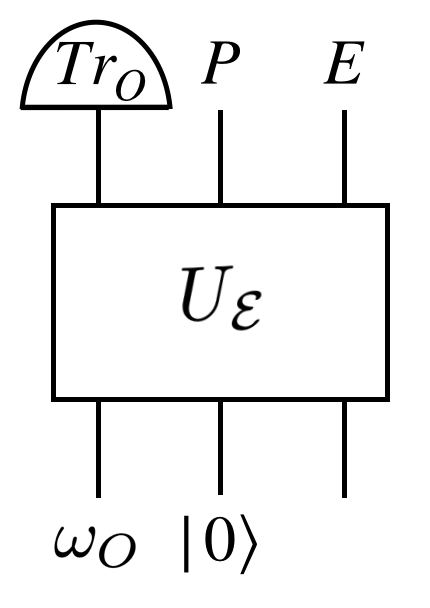}
    \caption{The operator applied by an exterior observer to the Hawking radiation depends on the observer's  initial state $\omega_O$, the probe's initial state $|0\rangle$, and the joint unitary transformation $U_{\mathcal{E}}$.}
    \label{fig:action_radiation}
\end{figure}

%While the discussion leading up to Figure~\ref{fig:action_radiation} is somewhat tautological, there is an important lesson that needs to be emphasized: 
In this scenario, the set of operations that the observer can apply to the radiation is not completely arbitrary. Specifically, any such operation must be of the following form:
\begin{equation}
    \rho_{PE} \mapsto \mathrm{Tr}_O \big( U_{\mathcal{E}}(\omega_O \otimes \rho_{PE}) U_{\mathcal{E}}^{\dagger} \big), \label{eq:action_radiation}
\end{equation}
wherein the only freedom available to the observer is the choice of the initial state $\omega_O$. Because the observer is part of a system that is governed by the laws of physics, the observer's actions are determined entirely by that initial state, not by the global unitary process. 
%For some readers, it may be more intuitive to use the following complementary picture in which the observer can actively apply some operator. 
One may view equation~\eqref{eq:action_radiation} as a quantum channel that acts on $PE$ with a Kraus representation and corresponding dilation given by %, and any such channel admits the following Kraus representation:
\begin{equation}
\begin{aligned}
    \rho_{PE} &\mapsto \sum_a E_a\rho_{PE} E_a^{\dagger} \\
    &= \text{Tr}_O\left(\sum_{a,b}(|a\rangle_O \otimes E_a)\,\rho_{PE}\,(\langle b|_O \otimes E_b^{\dagger})\right),
\end{aligned}   
\end{equation}
where $\sum_{a}E_a^{\dagger} E_a = I$, and $\{|a\rangle\}$ is an orthonormal basis for $O$. Therefore, $E_a\rho_{PE}E_a^{\dagger}$ can be thought as a (subnormalized) post-selected state in which the state of the observer after interacting with the radiation is $|a\rangle_O$. Up to normalization, the operator that the observer applied on the radiation would be $E_a$ in that case. While we do not know the exact details about $\{E_a\}$, within our model we have the following non-trivial constraints:
\begin{enumerate}
    \item The cardinality of the set $\{E_a \}$ is bounded above by $d_O$, where $d_O=2^{|O|}$ is the dimension of the observer's Hilbert space.
    \item The global unitary evolution $U_{\mathcal{E}}$ has a complexity polynomial in the black hole entropy $\sim |H|$.
\end{enumerate}

The construction of the mirror operators rests on the observation that $V_{\Psi}$ defines the embedding map of a quantum error-correcting code that can protect quantum information against ``environmental noise" caused by the observer $O$; see Figure~\ref{fig:bh_code}. %One can immediately see that 
The error model induced by the observer is different from conventional error models that are typically considered in discussions of fault-tolerant quantum computing. For one, the error is applied only on the radiation $E$ and probe $P$, not the remaining black hole $H$. Secondly, $U_{\mathcal{E}}$ can apply any  operation to the radiation with complexity polynomial in $|H|$. In contrast,  more conventional noise models such as the depolarizing channel or the amplitude damping channel typically result from a brief interaction between the environment and the system of interest. 

%Understanding why $V_{\Psi}$ can protect quantum information against this exotic error model constitutes the half of the puzzle in deriving equation~\eqref{eq:mirror}. That this error model is correctable follows from the analysis of Section~\ref{sec:pr_decoupling}.

We have already seen in Section~\ref{sec:bh_qecc} that the encoding map $V_\Psi$ protects quantum information against this exotic error model; this conclusion follows from the decoupling condition, which in turn is a consequence of the pseudorandomness of Hawking radiation as discussed in Section~\ref{sec:pr_decoupling}. Our next task is to relate this robustness against low-complexity noise to the claim in equation~\eqref{eq:mirror}.
%The remaining piece of the puzzle lies on proving a fact that relates the correctability of the aforementioned error model to equation~\eqref{eq:mirror}. 
The formalization and proof of this statement is the main technical contribution of this section. 

Before diving into details in the following subsections, let us summarize the conclusion. Consider an error model in which one applies either a channel $\mathcal{E}(\cdot) = \sum_a E_a (\cdot) E_a^{\dagger}$ or the identity channel, each occurring with nonzero probability. If a quantum error-correcting code $V_{\Psi}$ can correct such errors, then there is a complete set of logical operators that commutes with all the errors $\{E_a \}$ when acting on the code space; see Figure~\ref{fig:ghost}. That is, for any operator $\tilde{T}$ acting on the abstract logical space, there exists a corresponding logical operator $T$ acting identically on the code subspace such that $T$ satisfies the following intertwining condition for all $E_a$:
\begin{equation}\label{eq:commute-code-space}
    T E_a V_{\Psi} \approx E_a T V_{\Psi} \approx E_aV_{\Psi}\tilde{T}.
\end{equation}
These logical operators are special because the commutation relation holds as an operator equation acting on all the states in the code subspace. By mapping the isometry $V_{\Psi}$ back to the state $|\Psi\rangle$, we arrive at equation~\eqref{eq:mirror} and Figure~\ref{fig:ghost}. Note that this is a stronger statement than saying that the commutator of $T$ and $E_a$ has a vanishing expectation value in the code subspace, \emph{i.e.},
\begin{equation}
    V_{\Psi}^{\dagger}T E_a V_{\Psi} \approx V_{\Psi}^{\dagger}E_a T V_{\Psi}.
\end{equation}
In Section~\ref{sec:exactghost}, we will prove \eqref{eq:commute-code-space} in the exactly correctable setting. We will then generalize the construction to the approximate case in Section~\ref{sec:apxghost}. 

Equation~\eqref{eq:commute-code-space} also arises in the theory of Operator Algebra Quantum Error-Correction (OAQEC)~\cite{oaqec1,oaqec2}. However, in that context, one normally considers a logical operator $T$ which annihilates the orthogonal complement of the code space. A novelty of our discussion is that we will allow $T$ to have support extending beyond the code space. In that case, it is delicate to ensure that the action of $T$ on states outside the code space is consistent with \eqref{eq:commute-code-space}. More importantly, OAQEC was formulated in \cite{oaqec1,oaqec2} for the case of exact quantum error-correction. Our discussion in Section~\ref{sec:exactghost} is self-contained, and generalizes readily to the approximate setting, as we show in Section~\ref{sec:apxghost}.

%Note that for the case of exact error-correction, our results can be obtained as a special case of  While this is true, our construction in~\ref{sec:exactghost} is both self-contained and readily generalizable to the setting of approximate error-correction. As such, we will not use the language of OAQEC here.
\begin{figure}[ht]
    \centering
    \includegraphics[width=0.6\columnwidth]{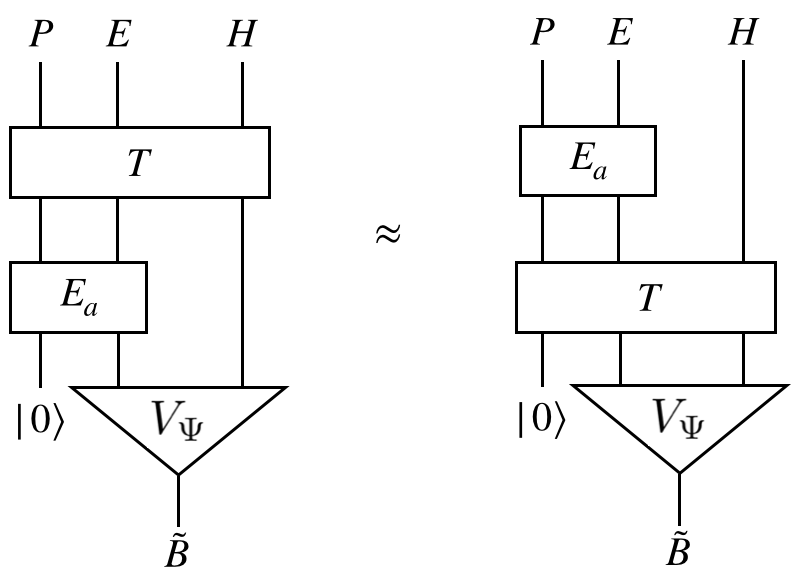}
    \caption{Acting on any code state, the ghost logical operator $T$ (approximately) commutes with any ``error'' in the set $\{E_a\}$.}
    \label{fig:ghost}
\end{figure}

\subsection{Exact ghost operators}\label{sec:exactghost}

Let $\tilde{\mathcal{H}}$ be an abstract logical Hilbert space, and consider an encoding $V:\tilde{\mathcal{H}} \rightarrow \mathcal{C}\subseteq \mathcal{H}$, where $\mathcal{C}$ denotes the code subspace embedded within the larger physical Hilbert space $\mathcal{H}$. Given a Hilbert space $\mathcal{H}$, we will let $S(\mathcal{H})$ denote the state space of $\mathcal{H}$, i.e., the set of all density operators supported on $\mathcal{H}$. Let $\mathcal{E}$ be a correctable error channel for $\mathcal{C}$, which we can write in a Kraus representation as
\begin{align}
    \mathcal{E}(\rho) = \sum_{a=1}^{|K|}E_a\rho E^\dagger_a,
\end{align}
where we denote the set of Kraus operators as $K = \{E_a\}$. A given channel will of course have many different Kraus representations; the choice of representation will not matter in the exact case, since the set of exactly correctable errors is closed under linear combinations, but we will have to be careful in the analysis of the approximate case in section~\ref{sec:apxghost}. In this section, we will fix an arbitrary Kraus representation $K$ for $\mathcal{E}$.

As a convention, we will denote quantities in $\tilde{\mathcal{H}}$ with tildes, and quantities in $\mathcal{H}$ without. Let
\begin{align}
    \tilde{T} = \sum_{k=1}^r \lambda_k\tilde{P}_k
\end{align}
be a normal operator on $\tilde{\mathcal{H}}$, with (distinct) eigenvalues $\{\lambda_k\}$, where each $\mathcal{P}_k$ is the spectral projector onto the corresponding eigenspace. For ease of notation, given any projector $P$, we will denote the corresponding range subspace as $[P]$, \emph{i.e.}, $[P] = \mathrm{Im}(P)$. 

Consider the encoded subspace $F_k = \mathrm{Im}(V\tilde{P}_k)$ of each eigenspace, and define
\begin{align}
    [P_k] =  \mathrm{span}\left\{E_a|\phi\rangle\ \bigg|\  E_a\in K,\ |\phi\rangle \in F_k\right\}.
\end{align}
Note that $[P_k]$ is the subspace generated by the set of all correctable errors, i.e., the span of $K$, acting on the encoded eigenspace $F_k$. These subspaces are well-defined since linear combinations of correctable errors remain correctable, and the Knill-Laflamme conditions \cite{knill1997theory} imply that subspaces corresponding to distinct eigenvalues will be orthogonal. We can then define a normal operator $T:\mathcal{H}\rightarrow \mathcal{H}$ by
\begin{align}
    T = \sum_{k=1}\lambda_kP_k,
\end{align}
where each $P_k$ is the corresponding projector onto $[P_k]$.
\begin{definition}
Given any normal operator $\tilde{T}:\tilde{\mathcal{H}}\rightarrow \tilde{\mathcal{H}}$, we will call the operator $T:\mathcal{H}\rightarrow \mathcal{H}$ obtained through the above construction the \emph{pseudo-ghost operator} corresponding to $\tilde{T}$.
\end{definition}

For any $|\chi_j\rangle \in F_j$ and any error operator $E_a \in K$, the action of the pseudo-ghost operator $T$ is such that
\begin{align}
    TE_a|\chi_j\rangle = \sum_{k=1}^r\lambda_kP_kE_a|\chi_j\rangle = \lambda_jE_a|\chi_j\rangle = E_a\hat{T}|\chi_j\rangle.
\end{align}
Here $\hat T$ can be any logical operator for $\tilde T$, which therefore satisfies $\hat T|\chi_j\rangle= \lambda_j |\chi_j\rangle$.
%\jp{JP: I  changed $|\lambda_j\rangle$ to $|\chi_j\rangle$ here and below, because I found it confusing to label the eigenvectors by their eigenvalues rather than their eigenspaces. And don't we need to say that $|\chi_j\rangle$ is in the $F_j$? (You had said it is in $[P_j]$.) Otherwise we can't be sure that $E_a|\chi_j\rangle$ is in $[P_j]$.} \ET{Yes, you are right.}

These pseudo-ghost operators satisfy $TE_a = E_a\hat T$ acting on the code space, and so do the ghost operators that we wish to construct. However, note that a pseudo-ghost operator $T$ will not necessarily act as a logical operator for $\tilde{T}$ since we might not have $F_k \subseteq [P_k]$ if the identity is not among the Kraus operators. The operator $T$ will not act correctly on the code subspace unless each of the encoded eigenspaces for $\tilde{T}$ are contained within the corresponding eigenspace for $T$. Our definition of a ghost logical operator should stipulate that $T$ is logical, as well as requiring $[T,E_a]=0$ acting on the code space.

\begin{definition}
Let $T:\mathcal{H}\rightarrow \mathcal{H}$ be a logical operator for $\tilde{T}$. We say that $T$ is a \emph{ghost logical operator} for $\tilde{T}$ if 
\begin{align}\label{eq:ghost_def}
    TE_a|\psi\rangle = E_aT|\psi\rangle
\end{align}
for all $E_a\in K$ and $|\psi\rangle \in \mathcal{C}$.
%%%
Given a pseudo-ghost operator $T$, we say that $T$ is \emph{extensible} if it admits an extension onto $\mathcal{H}$ such that it becomes a logical operator for $\tilde{T}$.
\end{definition}

Clearly the extension of any extensible pseudo-ghost operator will define a corresponding ghost logical operator. With the above definitions, it is simple to give a concise criterion for when pseudo-ghost operators extend to ghost logical operators in the exact setting.

\begin{lemma}\label{lem:ghost_exact1}
Let $T$ be a pseudo-ghost operator. Then $T$ is extensible if and only if
\begin{align}\label{eq:extensible}
    \langle \chi_j|E|\chi_i\rangle = 0,\qquad (i\neq j)
\end{align}
for all $E\in K$, $|\chi_i\rangle \in F_i$, $|\chi_j\rangle \in F_j$.
\end{lemma}
\begin{proof}
To see necessity, suppose that $T$ is extensible, and let $T'$ denote its logical extension. Because $T'$ is a logical operator for $\tilde{T}$, it must satisfy
\begin{align}
%T|\chi_{i}\rangle=\sum_{k=1}^{r}\lambda_{k}P_{k}|\chi_{i}\rangle=\lambda_{i}|\chi_{i}\rangle.
T'|\chi_{i}\rangle=\lambda_{i}|\chi_{i}\rangle.
\end{align}
Let $E\in K$ be arbitrary. Left multiplying by $\langle\chi_{j}|E^{\dagger}$, we get
\begin{align}
\lambda_{i}\langle\chi_{j}|E^{\dagger}|\chi_{i}\rangle=\langle \chi_j|E^\dagger T'|\chi_i\rangle = \sum_{k=1}^{r}\lambda_{k}\langle\chi_{j}|E^{\dagger}P_{k}|\chi_{i}\rangle=\lambda_{j}\langle\chi_{j}|E^{\dagger}|\chi_{i}\rangle.
\end{align}
Here we have used $E|\chi_j\rangle \in [P_k]$, and noted that $T$ and $T'$ have the same action on $[P_k]$; we also used $P_k E|\chi_j\rangle = \delta_{kj} \lambda_j E|\chi_j\rangle$.
If $\lambda_{i}\neq\lambda_{j}$, then we must have $\langle\chi_{j}|E^{\dagger}|\chi_{i}\rangle=0$. Taking the complex conjugate, we obtain equation~\eqref{eq:extensible}.

Conversely, suppose that for all $i\neq j$ and all correctable errors we have $\langle\chi_{i}|E|\chi_{j}\rangle=0$. We must extend the action of $T$ to each encoded eigenvector $|\chi_i\rangle \in \mathcal{C}$. The relations $\langle\chi_{i}|E|\chi_{j}\rangle=0$ imply that $|\chi_i\rangle$ is orthogonal to the subspaces $[P_j]$ for $j\neq i$. There are two possible cases, either $|\chi_i\rangle \in [P_i]$, for which $T|\chi_i\rangle = \lambda_i|\chi_i\rangle$ is already well-defined and we are done, or else there exists a component of $|\chi_i\rangle$ lying in the subspace orthogonal to $\bigoplus_{k=1}^r[P_k]$. 

Let $|\chi_i^\perp\rangle$ denote the normalized component of $|\chi_i\rangle$ orthogonal to $[P_i]$. Then we extend the subspace $[P_i]$ to $[P'_i]$ by defining the projector
\begin{align}
    P'_i = P_i + |\chi_i^\perp\rangle\langle\chi_i^\perp|.
\end{align}
Note that the new subspace $[P'_i]$ contains within it $[P_i]$ and remains orthogonal to $[P_j]$ for $j\neq i$. Moreover, we have $|\chi_i\rangle \in [P'_i]$. We can now define an extension of $T$ with the projector $P'_i$ in place of $P_i$. Then the extension $T'$ satisfies
\begin{align}
    T'|\chi_i\rangle = \lambda_i|\chi_i\rangle.
\end{align}
We may repeat this procedure with an orthogonal basis $\{|\chi_k\rangle\}$ for $\mathcal{C}$ until we are left with an extension which acts as a logical operator for $\tilde{T}$.
\end{proof}

We will be primarily interested in the case where there exists a full set of ghost logical operators. We say that there exists a \emph{complete set} of ghost logical operators if for every normal operator $\tilde{T}:\tilde{\mathcal{H}}\rightarrow \tilde{\mathcal{H}}$, there exists a corresponding ghost logical operator $T$. In what follows, given a channel $\mathcal{E}$, we will let $\mathcal{E}_\mathcal{I}$ denote the channel 
\begin{align}
\mathcal{E}_\mathcal{I} = \mathcal{I}/2 + \mathcal{E}/2, 
\end{align}
where $\mathcal{I}$ is the identity channel. That is, in the channel $\mathcal{E}_\mathcal{I}$, with probability $1/2$ $\mathcal{E}$ is applied, and with probability $1/2$ nothing happens. 

\begin{theorem}\label{thm:ghost_exact2} 
Let $\mathcal{E}$ be a correctable channel with Kraus operators $K$. Then a complete set of ghost logical operators for $\mathcal{E}$ exists if and only if $K\cup \{I\}$ is a correctable set, i.e., if and only if $\mathcal{E}_\mathcal{I}$ is a correctable channel.
\end{theorem}
\begin{proof}
Suppose that $K\cup\{I\}$ is a correctable set. Then the Knill-Laflamme conditions for $K\cup\{I\}$ imply that the hypotheses of Lemma~\ref{lem:ghost_exact1} are satisfied so that every pseudo-ghost operator is extensible to a ghost logical operator. It follows that there exists a complete set of ghost logical operators.

Conversely, suppose that there exists a complete set of ghost logical operators. Let $|\psi\rangle, |\phi\rangle\in \mathcal{C}$ be two mutually orthogonal code states, and let $|\tilde{\psi}\rangle=V^{\dagger}|\psi\rangle$ and $|\tilde{\phi}\rangle=V^{\dagger}|\phi\rangle$ be the corresponding pre-images in $\tilde{\mathcal{H}}$. Define the operators
\begin{align}
\tilde{T}_{1}=|\tilde{\phi}\rangle\langle\tilde{\phi}|-|\tilde{\psi}\rangle\langle\tilde{\psi}|,
\end{align}
and
\begin{align}
\tilde{T}_{2}=|\tilde{\phi}+\tilde{\psi}\rangle\langle\tilde{\phi}+\tilde{\psi}|-|\tilde{\phi}-\tilde{\psi}\rangle\langle\tilde{\phi}-\tilde{\psi}|,
\end{align}
where $|\tilde{\phi}\pm \tilde{\phi}\rangle = 2^{-1/2}(|\tilde{\phi}\rangle\pm |\tilde{\psi}\rangle)$. By assumption, there exist ghost logical operators $T_{1}$ and $T_{2}$ corresponding to $\tilde{T}_{1}$ and $\tilde{T}_{2}$. Now let $E_{a},E_{b}\in K\cup\{I\}$. Then we have
\begin{align}
\langle\psi|E_{a}^{\dagger}E_{b}|\phi\rangle &=\langle\psi|E_{a}^{\dagger}E_{b}T_{1}|\phi\rangle\\
&=\langle\psi|T_{1}E_{a}^{\dagger}E_{b}|\phi\rangle\\
&=-\langle\psi|E_{a}^{\dagger}E_{b}|\phi\rangle,
\end{align}
where the first line follows due to the fact that $|\phi\rangle$ is an eigenvector for $T_1$ with eigenvalue $1$, the second line follows from the defining equations~\eqref{eq:ghost_def} for the ghost operators, together with the fact that $T_1$ is self-adjoint, and the last line follows from the fact that $|\psi\rangle$ is an eigenvector for $T_1$ with eigenvalue $-1$. This implies that $\langle\psi|E_{a}^{\dagger}E_{b}|\phi\rangle=0$.

Repeating the same argument for $T_2$, we have
\begin{align}
\langle\phi-\psi|E_{a}^{\dagger}E_{b}|\phi+\psi\rangle &=\langle\phi-\psi|E_{a}^{\dagger}E_{b}T_{2}|\phi+\psi\rangle\\
&=\langle\phi-\psi|T_{2}E_{a}^{\dagger}E_{b}|\phi+\psi\rangle\\
&=-\langle\phi-\psi|E_{a}^{\dagger}E_{b}|\phi+\psi\rangle,
\end{align}
which implies that
\begin{align}
0=\langle\phi|E_{a}^{\dagger}E_{b}|\phi\rangle-\langle\psi|E_{a}^{\dagger}E_{b}|\psi\rangle.
\end{align}
Since $\phi$ and $\psi$ were arbitrary, this holds for any pair of orthogonal states. 

Let $\{|j\rangle\}$ be an orthonormal basis for $\mathcal{C}$ and define $\lambda_{ab}=\langle\psi|E_{a}^{\dagger}E_{b}|\psi\rangle$ for an arbitrary state $|\psi\rangle \in \mathcal{C}$. Then it follows that we have
\begin{align}
\langle i|E_{a}^{\dagger}E_{b}|j\rangle=\lambda_{ab}\delta_{ij},
\end{align}
so that the Knill-Laflamme conditions for $K\cup\left\{ I\right\}$ are satisfied. Therefore, $K\cup\left\{ I\right\}$ is a correctable set of errors.
\end{proof}

\subsection{Approximate ghost operators}\label{sec:apxghost}

In this section, we discuss how the ghost logical operators can be constructed for \emph{approximate} quantum error-correcting codes. We need to consider this case because we inferred in Section~\ref{sec:bh_qecc} that the errors due to low-complexity operations on the radiation system $E$ are correctable approximately (with a residual error exponentially small in $|H|$) rather than exactly. Although the uncorrected error is exponentially small, the Hilbert space is exponentially large, so we need to do a careful analysis to check that the ghost logical operators commute with the errors apart from exponentially small effects. 

%The ensuing analysis is motivated from the fact that black hole is not an exact error-correcting code: the polynomial-time unitaries we've mentioned earlier, which can be viewed as errors on the code, will inflict some uncorrectable error on the logical qubit. While the ``amount" of uncorrectable error is exponentially small in the entropy of the remaining black hole, there is an exponentially large factor that depends on the number of qubits the observer is carrying. In order to understand the combined effect of these two competing exponentials, we need to carry out a more careful analysis.

It turns out the strategy that we pursued in the exact setting also works in the approximate setting. To get started, we will construct \emph{approximate ghost projectors} $\{\mathcal{P}_i\}$ that play the same role as the $\{P_i\}$ in the previous section.  
%\jp{JP: Before we said ``pseudo-ghost projectors,'' but we never used that terminology again, so I changed it to ghost projectors as in the Definition.}

\begin{definition}
Let $\{|\tilde{i}\rangle\}$ be an orthonormal basis for $\tilde{\mathcal{H}}$, and suppose $V$ is an encoding isometry. We define (approximate) \emph{ghost projectors} with respect to this basis, denoted $\mathcal{P}_i$,  to be the orthogonal projectors onto the \emph{positive} eigenspace of
\begin{equation}
    \mathcal{E}(|i\rangle\langle i| - \rho_{i, \perp}),
\end{equation}
where $|i\rangle = V|\tilde{i}\rangle$ for $|\tilde{i}\rangle \in \tilde{\mathcal{H}}$, and where \begin{align}
\rho_{i, \perp} = \frac{1}{\dim\tilde{\mathcal{H}}-1}\sum_{j \neq i} |j\rangle \langle j|.
\end{align}
\end{definition}

The motivation behind this definition follows from the fact that $\mathcal{P}_i$ is an operator that can optimally distinguish $\mathcal{E}(|i\rangle\langle i|)$ from $\mathcal{E}(\rho_{i, \perp})$, according to the Holevo-Helstrom theorem \cite{MikeIke}. Because the effect of the channel $\mathcal{E}$ can be reversed up to a small error, it nearly preserves the orthogonality of $|i\rangle\langle i|$ and $\rho_{i,\perp}$; therefore, $\mathcal{P}_i$ can distinguish the two states almost perfectly. This suggests that $\mathcal{P}_i$, up to a small error, projects $\mathcal{E}(|i\rangle\langle i|)$ to a state close to $\mathcal{E}(|i\rangle\langle i|)$  and nearly annihilates $\mathcal{E}(\rho_{i,\perp}$). In the following two lemmas, we prove these claims rigorously. In Lemma~\ref{lem:ghostapx1}, we show that $\mathcal{P}_iE_a|i\rangle \approx E_a|i\rangle$, and in Lemma~\ref{lem:ghostapx2} we show that $\mathcal{P}_iE_a|j\rangle \approx 0$, for $i\neq j$, where $E_a$ is any Kraus operator of the channel $\mathcal{E}$.

If $\mathcal{E}$ is an $\epsilon$-correctable channel then we have 
\begin{align}
    \max_{\rho} \|(\mathcal{R} \circ \mathcal{E})(\rho) - \rho \|_1 \leq 2\sqrt{2}\epsilon := \tilde{\epsilon},
\end{align}
as given by equation~\eqref{eq:correctable}. Let us define $\tilde{\epsilon} = 2\sqrt{2}\epsilon$ to minimize factors of $2\sqrt{2}$. Then we can obtain the following bound:

%Note that the state $E_a|i\rangle$ is not normalized in general. With the appropriate normalization, we get the following bound: \jp{I don't understand. What does the bound have to do with the normalization?} \IK{I think we can remove this comment. Here's the background. The original bound was similar to Lemma \ref{lem:ghostapx1}, except that the states were normalized. One may have been tempted to use the normalized version because, even if $\epsilon$ is small, if the state $E_a|i\rangle$ has a small norm, a normalized version of the statement $E_a|i\rangle \approx \mathcal{P}_iE_a|i\rangle$ can be false. However, I think the current (unnormalized) version is better because (1) it's cleaner and (2) if the norm of the state is small, the probability that we end up in that state is so small that we probably don't care too much. (Recall that $\epsilon$ is exponentially small.)}

\begin{lemma}\label{lem:ghostapx1}
Let $\mathcal{E}$ be an $\epsilon$-correctable channel and let $\mathcal{P}_i$ be the corresponding ghost projector with respect to some basis. Then we have
\begin{equation}\label{eq:lemma86}
\left\|E_a|i\rangle - \mathcal{P}_i E_a|i\rangle\right\|_2^2 \leq 2\sqrt{2}\epsilon := \tilde{\epsilon},
\end{equation}
where $\| |\phi\rangle \|_2 := \sqrt{\langle \phi|\phi\rangle}$.
\end{lemma}
%\noindent \jp{Is it correct to use the subscript 2 on the norm? This is the norm of a vector, not an operator. This occurs a few more times below. I guess it's fine.} \IK{I added the definition.}
\begin{proof}
Note that, by the monotonicity of the trace norm, we have 
\begin{equation}
    \|\mathcal{E}(|i\rangle\langle i| - \rho_{i, \perp}) \|_1 \geq \| (\mathcal{R} \circ \mathcal{E})(|i\rangle\langle i| - \rho_{i, \perp})\|_1.
\end{equation}
We can use the fact that the recovery map $\mathcal{R}$ nearly succeeds in recovering the original state. By the triangle inequality, 
\begin{equation}
\begin{aligned}
    2=\||i\rangle \langle i| - \rho_{i, \perp} \|_1 &\le \|(\mathcal{R}\circ \mathcal{E})(|i\rangle\langle i| - \rho_{i, \perp}) \|_1\\ 
    & \quad + \||i\rangle \langle i| - (\mathcal{R}\circ \mathcal{E})(|i\rangle\langle i|) \|_1\\ 
    & \quad + \|\rho_{i,\perp} - (\mathcal{R}\circ \mathcal{E})(\rho_{i, \perp}) \|_1. \label{eq:proof_approx_ghost_triangle}
    \end{aligned}
\end{equation}
Therefore,
\begin{equation}\label{eq:mathE-epsilon-tilde}
    \|\mathcal{E}(|i\rangle\langle i| - \rho_{i, \perp}) \|_1 \geq 2 - 2\tilde{\epsilon}.
\end{equation}
Moreover, we have
\begin{equation}
\begin{aligned}
    \|\mathcal{E}(|i\rangle\langle i| - \rho_{i, \perp}) \|_1 &= \mathrm{Tr}(2\mathcal{P}_i \mathcal{E}(|i\rangle\langle i| - \rho_{i, \perp})) \\
    & \leq 2 \mathrm{Tr}(\mathcal{P}_i \mathcal{E}(|i\rangle\langle i|)).\label{eq:ghost1}
\end{aligned}
\end{equation}
The first line above follows by decomposing $\mathcal{E}(|i\rangle\langle i| - \rho_{i, \perp})$ into its positive and negative parts. Because the operator is traceless, the trace of the positive part is equal to the trace of the negative part, up to a minus sign. Since the trace distance is equal to the sum of the absolute value of the positive and negative trace, and because these values are the same, we arrive at the first identity. The second line then follows from the fact that $\mathrm{Tr}(\mathcal{P}_i\mathcal{E}( \rho_{i,\perp})) \geq 0$.

Therefore, we get the following bound:
\begin{equation}\label{eq:epsilon-tilde-bound}
\begin{aligned}
    1-\tilde \epsilon & \le \mathrm{Tr}(\mathcal{P}_i\mathcal{E}(|i\rangle\langle i|))\\
    &= \sum_a \langle i|E_a^{\dagger}\mathcal{P}_i E_a |i \rangle \\
    & = \sum_a q_{ia} \langle\psi_{ia}| \mathcal{P}_i|\psi_{ia}\rangle \\
    & = 1-\sum_a q_{ia} (1-\langle \psi_{ia}| \mathcal{P}_i |\psi_{ia}\rangle),
\end{aligned}
\end{equation}
where we define
\begin{align}
|\psi_{ia}\rangle = \frac{E_a|i\rangle}{ \sqrt{\langle i| E_a^{\dagger}E_a|i\rangle}},
\end{align}
and $q_{ia} = \langle i|E_a^{\dagger}E_a |i\rangle$. Note that $\sum_a q_{ia}=1$ since $\mathcal{E}$ is trace-preserving. Therefore, we get
\begin{equation}
    1-\langle \psi_{ia}| \mathcal{P}_i |\psi_{ia}\rangle \leq \frac{\tilde{\epsilon}}{q_{ia}}
\end{equation}
by noting that the last line of equation~\eqref{eq:epsilon-tilde-bound} contains a sum of non-negative terms. Since the sum is $\le\tilde\epsilon$, each individual term must be $\le \tilde \epsilon$ as well. Substituting in the expressions for $q_{ia}$ and $|\psi_{ia}\rangle$, this inequality becomes
\begin{equation}
    \langle i|E_a^\dagger E_a|i\rangle - \langle i|E_a^\dagger\mathcal{P}_i E_a|i\rangle \le \tilde \epsilon,
\end{equation}
which is equivalent to equation~\eqref{eq:lemma86}.
\end{proof}

\begin{lemma}\label{lem:ghostapx2}
Under the same hypothesis as Lemma~\ref{lem:ghostapx1}, if $i\neq j$, then
\begin{equation}\label{eq:lemma87}
    \left\|\mathcal{P}_i E_a|j\rangle\right\|_2^2 \leq (\dim\mathcal{C})\,\tilde{\epsilon}.
\end{equation}
\end{lemma}
\begin{proof}
Note that
\begin{equation}
\begin{aligned}
    2-2\tilde{\epsilon} &\leq \|\mathcal{E}(|i\rangle\langle i| - \rho_{i,\perp}) \|_1 \\
    &= \mathrm{Tr}(2\mathcal{P}_i\mathcal{E}(|i\rangle \langle i| - \rho_{i,\perp})) \\
    &\leq 2- 2\mathrm{Tr}(\mathcal{P}_i \mathcal{E}(\rho_{i,\perp})),
\end{aligned}
\end{equation}
where we've used equation~\eqref{eq:mathE-epsilon-tilde} in the first line, and equation~\eqref{eq:ghost1} in the second. The last line follows from the fact that $\mathcal{P}_i \leq I$. It follows that
\begin{equation}
 \frac{1}{\dim\mathcal{C} -1}\sum_{j\ne i}  \mathrm{Tr}(\mathcal{P}_i\mathcal{E}(|j\rangle\langle j|))
 =\mathrm{Tr}(\mathcal{P}_i\mathcal{E}(\rho_{i,\perp})) \leq \tilde{\epsilon},
\end{equation}
and therefore, we have 
\begin{equation}
    \mathrm{Tr}(\mathcal{P}_i \mathcal{E}(|j\rangle\langle j|)) \leq (\dim\mathcal{C})\,\tilde{\epsilon},
\end{equation}
for all $j\neq i$. 
%Let us define $|\psi_{ja}\rangle$ and $q_{ja}$ as in Lemma~\ref{lem:ghostapx1}. Then we get
%\begin{equation}
%\langle \psi_{ja} | \mathcal{P}_i |\psi_{ja} \rangle \leq \frac{(\dim\mathcal{C})\,\tilde{\epsilon}}{ q_{ja}},
%\end{equation}
Expanding in terms of the Kraus operators of the channel $\mathcal{E}$, this becomes
\begin{equation}
    \sum_a \mathrm{Tr}(\mathcal{P}_iE_a|j\rangle\langle j| E_a^\dagger\mathcal{P}_i) 
    =\sum_a  \left\|\mathcal{P}_i E_a|j\rangle\right\|_2^2
    \le (\dim\mathcal{C})\tilde\epsilon,
\end{equation}
where %we can write the first expression 
the first equality holds because $\mathcal{P}_i$ is a projector. Equation~\eqref{eq:lemma87} then follows.
\end{proof}

At this point, we can follow the construction we used for the case of exact ghost operators. Let $\mathcal{C}$ be a code subspace and let $\mathcal{E}$ be an error channel such that $\mathcal{E}_\mathcal{I}$ is $\epsilon$-correctable. Then by Lemmas~\ref{lem:ghostapx1} and~\ref{lem:ghostapx2}, we have
\begin{align}
    \|E_a|i\rangle - \mathcal{P}_iE_a|i\rangle\|_2^2 \le 2\tilde{\epsilon},\quad\text{and}\quad \|\mathcal{P}_iE_a|j\rangle\|_2^2 \le 2(\dim\mathcal{C})\,\tilde{\epsilon},
\end{align}
where each $E_a$ is a Kraus operators for $\mathcal{E}$, or the identity. Note that the extra factor of $2$ comes from the fact that the Kraus operators for $\mathcal{E}_{\mathcal{I}}$ are given by $\{E_a/\sqrt{2}\}\cup \{I/\sqrt{2}\}$, where each $E_a$ is a Kraus operator for $\mathcal{E}$.

Given a normal operator $\tilde{T}:\tilde{\mathcal{H}} \rightarrow \tilde{\mathcal{H}}$ defined by
\begin{align}
    \tilde{T} = \sum_k \lambda_k |\tilde{k}\rangle\langle \tilde{k}|,
\end{align}
we define the operator
\begin{align}
    T = \sum_k \lambda_k \mathcal{P}_k,\label{eq:ghost_op}
\end{align}
where each $\mathcal{P}_k$ is a ghost projector with respect to the given eigenbasis for $\tilde{T}$.
Then the operator $T$ satisfies
\begin{equation}
\begin{aligned}
    \|TE_{a}|j\rangle - \lambda_{j}E_{a}|j\rangle\|_2
     &= \left\Vert \sum_{k}\lambda_{k}\mathcal{P}_{k}E_{a}|j\rangle-\lambda_{j}E_{a}|j\rangle\right\Vert_2\\
    &\le \left\Vert \sum_{k\neq j}\lambda_{k}\mathcal{P}_{k}E_{\ell}|j\rangle\right\Vert_2 +\left\Vert \lambda_{j}\mathcal{P}_{j}E_{\ell}|j\rangle-\lambda_{j}E_{a}|j\rangle\right\Vert_2\\
    &\le \sum_{k\neq j}\left|\lambda_{k}\right|\left\Vert_2 \mathcal{P}_{k}E_{a}|j\rangle\right\Vert_2 +\left|\lambda_{j}\right|\left\Vert \mathcal{P}_{j}E_{a}|j\rangle-E_{a}|j\rangle\right\Vert_2 \\
    &\le \sqrt{2(\dim\mathcal{C})\tilde{\epsilon}}\sum_{k\neq j}\left|\lambda_{k}\right|+\left|\lambda_{j}\right|\sqrt{2\tilde{\epsilon}}\\
    &\le \sqrt{2(\dim\mathcal{C})\tilde{\epsilon}}\|\tilde{T}\|_{1},
\end{aligned}
\end{equation}
%%
%\jp{My pet peeve is unnumbered equations, because it complicates things if you want to direct attention to a particular equation. You probably did that here because you were worried about a page break in the middle of the multi-line equation? I added numbers here and in the following equation.}
%%
where $\|\tilde{T}\|_1$ is the trace norm of $\tilde{T}$. 
%\jp{and we have now dropped the subscript 2 on the Hilbert space norm. [I'm not sure why. Maybe we should remove in in earlier equations, too?]} \IK{Added the $_2$ subscripts for the $2-$norms.}
Now, let $\hat{T} = V\tilde{T}V^\dagger$, where $V$ is the code embedding. Then for a general code state $|\psi\rangle=\sum_{j}c_{j}|j\rangle$, we have
\begin{equation}
\begin{aligned}
\left\Vert TE_{a}|\psi\rangle-E_{a}\hat{T}|\psi\rangle\right\Vert_2 &\le \sum_j|c_j|\cdot \left\Vert TE_{a}|j\rangle-\lambda_jE_{a}|j\rangle\right\Vert_2\\
&\le\sqrt{2(\dim\mathcal{C})\tilde{\epsilon}}\|\tilde{T}\|_{1}\sum_{j}\left|c_{j}\right|\\
&\le(\dim\mathcal{C})\|\tilde{T}\|_{1}\sqrt{2\tilde{\epsilon}}.
\end{aligned}
\end{equation}
A slightly weaker, but more convenient bound in terms of the operator norm of $\tilde{T}$ can be given as
\begin{align}
\left\Vert TE_{a}|\psi\rangle-E_{a}\hat{T}|\psi\rangle\right\Vert_2 \le (\dim\mathcal{C})^{2}\|\tilde{T}\|\sqrt{\left(4\sqrt{2}\right)\epsilon},
\end{align}
which we can also express as a bound on the difference of two operators in the operator norm:
\begin{align}
\left\Vert TE_{a}V-E_{a}V\tilde{T}\right\Vert \le 2^{5/4}(\dim\mathcal{C})^{2}\|\tilde{T}\|\sqrt{\epsilon},\label{eq:apxghostbound}
\end{align}
where we have used $\tilde \epsilon = 2\sqrt{2} \epsilon$. Note that the bound~\eqref{eq:apxghostbound} holds for any Kraus representation $\{E_a\}$ of $\mathcal{E}$.
The bound~\eqref{eq:apxghostbound} motivates the following definition:

\begin{definition}
Let $(\mathcal{E}, K)$ be a noise channel $\mathcal{E}$ equipped with a given Kraus representation $K=\{E_a\}$. Let $\tilde{T}:\tilde{\mathcal{H}}\rightarrow \tilde{\mathcal{H}}$ be a normal operator. We say that $T$ is a $\delta$-approximate ghost operator for $\tilde{T}$ with respect to $(\mathcal{E},K)$ if we have
\begin{align}
\|TE_{a}V - E_{a}V\tilde{T}\| &\le \|\tilde{T}\|\delta,\label{eq:ghostapx}
\end{align}
where $E_{a} \in K\cup\{I\}$ is either a Kraus operator for $\mathcal{E}$, or the identity.
%%%
We say that a ghost operator $T$ is \emph{universal} if equation~\eqref{eq:ghostapx} holds for every Kraus representation of $\mathcal{E}$.
\end{definition}

Now we are ready to prove the analog of Theorem~\ref{thm:ghost_exact2} in the approximate setting. As before, we say that there exists a complete set of $\epsilon$-approximate ghost operators if there exists an $\epsilon$-approximate ghost logical operator for every normal operator on $\tilde{\mathcal{H}}$. 

\begin{theorem}\label{thm:ghostapx}
Let $\mathcal{C}$ be a code subspace and suppose that $\mathcal{E}_\mathcal{I}$ is $\epsilon$-correctable for $\mathcal{C}$. Then there exists a complete set of $\delta$-approximate universal ghost operators, where
\begin{align}\label{eq:thm-89}
\delta = 2^{5/4}(\dim\mathcal{C})^2\sqrt{\epsilon}.
\end{align}
\end{theorem}
\noindent For the sake of completeness, we also prove a converse of this result (Theorem~\ref{thm:ghost_converse}) in Appendix~\ref{sec:ghost_converse}. These results collectively can be seen as a generalization of the standard theorems of operator algebra quantum error-correction~\cite{oaqec1,oaqec2} to the approximate setting.

\begin{proof}
Suppose that $\mathcal{E}_\mathcal{I}$ is $\epsilon$-correctable for $\mathcal{C}$. Then equation~\eqref{eq:apxghostbound} shows that $T$ defined by equation~\eqref{eq:ghost_op} is a $\delta$-approximate ghost operator for any normal operator $\tilde T$, where $\delta=2^{5/4}(\dim\mathcal{C})^2\sqrt{\epsilon}$. The construction of the ghost projector $\mathcal{P}_k$, and therefore also the construction of $T$, depends only on the channel $\mathcal{E}$ and not on any particular Kraus representation; it follows that $T$ is universal. 
%\jp{I don't know what you mean by ``the preceding discussion.'' Please be more explicit, or omit.}\ET{I had meant the entire discussion leading up to equation~\eqref{eq:apxghostbound}, but it's really just equation~\eqref{eq:apxghostbound} that we need, so I've included it explicitly.}
\end{proof}

%\textcolor{red}{(ET) Are you very attached to the following part? I don't find it particularly helpful. We can keep it (perhaps moved somewhere else) or remove it, up to you.}

%Physically, the state $E_a|i\rangle$ can be interpreted as a post-selected state in which the environment decides to apply an error $E_a$. To see why, one can purify the map $\mathcal{E}(\rho) = \sum_a E_a \rho E_a^{\dagger}$ by the isometric extension
%\begin{equation}
%    U_{\mathcal{E}} = \sum_a E_a \otimes  |a\rangle_E,
%\end{equation}
%where $E$ is an auxiliary environment \footnote{In the black hole physics, $E$ should be thought of as representing an external observer.}. One can verify that
%\begin{equation}
%    U_{\text{ext}}\rho U_{\text{ext}}^{\dagger} = \sum_{a, b} E_a \rho E_{b}^{\dagger} \otimes |a\rangle \langle b|_E
%\end{equation}
%Tracing out $E$, we observe that
%\begin{equation}
%    \text{Tr}_E\left(U_{\text{ext}}\rho U_{\text{ext}}^{\dagger}\right) = \mathcal{E}(\rho).
%\end{equation}
%Furthermore, by projecting onto a particular state $|a\rangle_E$ on the environment, we obtain the following (subnormalized) state:
%\begin{equation}
%    _E\langle a|  U_{\text{ext}}\rho U_{\text{ext}}^{\dagger} |a\rangle_E = E_a\rho E_a^{\dagger}.
%\end{equation}

%Therefore, $\mathcal{P}_i$ has a remarkable property. Independent of what state the environment is in, the projector approximately annihilates $E_a|j\rangle$ for $j\neq i$, and projects $E_a|i\rangle$ back onto $E_a|i\rangle$.

\subsection{Firewall revisited}\label{sec:firewall_revisited}

We have now seen that, by assuming that the state of the Hawking radiation system $EB$ is pseudorandom, we may infer that low-complexity operations on $E$ are approximately correctable; the code space $\tilde B$ that purifies the late radiation system $B$ is protected against low-complexity operations on $E$. Correctability in turn implies that a complete set of ghost logical operators acting on $EH$, which nearly commute with all low-complexity operations on $E$, can be constructed.

Let us now reconsider the potential implications of the existence of ghost logical operators in the context of the black hole firewall problem. First, we assemble the results we have derived thus far to determine the value of $\delta$ for which the ghost logical operators are $\delta$-approximate. Under the pseudorandomness assumption equation~\eqref{eq:pr}, we saw in Lemma \ref{lem:correctable} that low-complexity operations are $\epsilon$-correctable for $\epsilon = \sqrt{3/2}\cdot 2^{-(\alpha|H| - |OB|)/2}$. Since the code space dimension is $\dim \mathcal{C} = 2^{|B|}$, equation~\eqref{eq:thm-89} says that the ghost operators are $\delta$-approximate for
\begin{equation}\label{eq:delta-from-pseudo}
    \delta = 2^{5/4} 2^{2|B|} \sqrt{\epsilon}
    = 2\cdot 3^{1/4} ~ 2^{2|B|} ~2^{-(\alpha |H| - |OB|)/4}
     = 2\cdot 3^{1/4} ~2^{-(\alpha |H| - |O| - 9|B|)/4}.
\end{equation}
Thus, $\delta$ becomes exponentially small for asymptotically large $|H|$, $|O|$, and $|B|$, provided $|O|, |B| \ll  |H|$. We could, for example, consider an encoded interior and an observer with size scaling linearly with $|H|$, and still have a complete set of ghost logical operators commuting with all low-complexity operations on $E$, up to exponentially small errors. 

This conclusion followed only from the assumption that the state of $EB$ is pseudorandom --- we needed no other special properties of black holes to derive it. We might, in fact, expect the same pseudorandomness assumption to hold not just for black holes but also for other strongly chaotic quantum systems. But a black hole \emph{is} special, because it has an event horizon, and it is because of the event horizon that we expect the late radiation system $B$ to be entangled with modes behind the horizon as well as with a subspace of $EH$; thus arises the black hole firewall problem. To ease the firewall problem, we propose using the ghost logical operators to describe (a portion of) the black hole interior. We would not make such a proposal for describing the ``interior'' of a burning lump of coal. 

Pleasingly, under this proposal, it is hard for an agent who acts on the radiation to create a firewall, or to otherwise influence the black hole interior apart from exponentially small effects. To create an excitation behind the horizon, the agent outside the black hole must perform an operation of superpolynomial complexity. 

We might want to allow the observer to perform a quantum computation on $E$ which is chosen from a long list of possible unitary transformations. The observer's freedom to choose can be encoded in the observer's initial state $\omega_O$, as depicted in Figure~\ref{fig:action_radiation}. If there are multiple observers $\{O_1, O_2, \dots O_m\}$, all interacting with $E$, we can group them all together into a collective observer $O=O_1O_2\dots O_m$. We may construct a complete set of ghost logical operators acting on the encoded black hole interior, consistently shared by all the observers, provided that $|O|, |B| \ll |H|$.

To be more concrete, suppose we want the black hole interior to be protected against any unitary transformation acting on $E$ chosen from amongst a collection of $N$ possible unitaries $\mathcal{U}=\{U_a\}_{a=1}^N$. We can model this situation by considering a conditional unitary transformation, controlled by an ancilla register in the observer's possession. To ensure that we can apply Theorem \ref{thm:ghostapx} we will add the identity transformation $U_0 = I_E$ to the list of possibilities, and envision that the observer applies 
\begin{equation}\label{eq:U_U-controlled}
    U_{\mathcal{U}} = \sum_{a=0}^{N} |a\rangle\langle a|_O\otimes (U_a)_E,
\end{equation}
where each $|a\rangle_O$ is a computational basis state and $2^{|O|} = N{+}1$.
Thus $U_a$ is applied by fixing the initial state of the $O$ register to be $|a\rangle_O$; see Figure~\ref{fig:controlledU}.

\begin{figure}[ht]
\centering
\includegraphics[width=0.6\columnwidth]{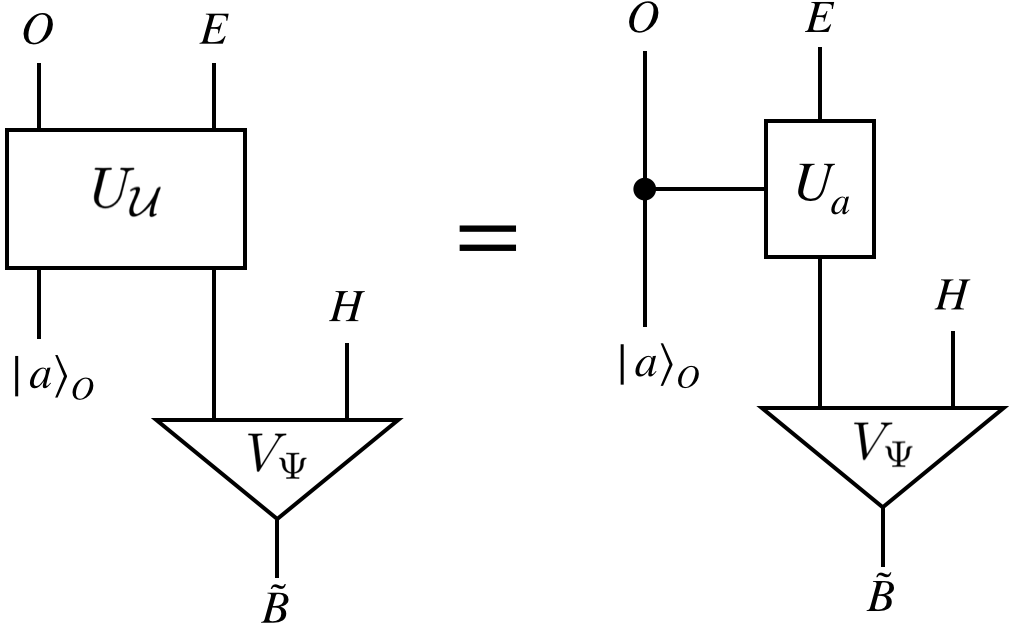}
\caption{The action of the observer as a controlled unitary tranformation. 
%\jp{JP: I changed the name of the unitary to $U_\mathcal{U}$, and used $\mathcal{U}$ to denote the set of unitaries the observer can apply to $E$. I'm not sure that was such a great idea, but one way or another the figure should be consistent with the equations. (That is, either change it to $U_\mathcal{U}$ in the figure, or change it back to $U_\mathcal{E}$ in the equations.} \IK{I changed the figure.}
}
\label{fig:controlledU}
\end{figure}

Our construction of a complete set of ghost logical operators applies --- assuming the Hawking radiation is pseudorandom --- if $U_\mathcal{U}$ has complexity polynomial in $|H|$. This will be assured if the cardinality $N$ of the list of unitaries is polynomial in $|H|$. The unitary
\begin{equation}
    \Lambda_a(U_a) = |a\rangle\langle a|_O \otimes (U_a)_E
\end{equation}
for which a non-trivial unitary acting on $E$ is triggered only by the basis state $|a\rangle_O$, has polynomial quantum complexity if $U_a$ does --- we show in Lemma \ref{lem:controlU} that, if we fix the complexity of $U_a$, then $\Lambda_a(U_a)$ can be implemented to precision $\epsilon$ with a circuit of $O(N^2 \log^4(1/\epsilon))$ two-gubit gates. Furthermore, the overall operator
\begin{align}
U_{\mathcal{U}} = \prod_{a=0}^{N}\Lambda_a(U_a)
\end{align}
is a product of $N+1$ such unitaries, and thus has complexity at worst a factor of $N+1$ larger. Therefore, if $N= \textrm{poly}(|H|)$, then $U_\mathcal{U}$ can be executed to exponential precision with a circuit of size $\textrm{poly}(|H|)$.

The unitary transformation $U_\mathcal{U}$ is a dilation of the quantum channel
\begin{equation}\label{eq:channel-E-U}
    \mathcal{E}_\mathcal{U}(\rho) = \frac{1}{N+1}\sum_{a=0}^N U_a \rho U_a^\dagger 
\end{equation}
acting on $E$, with Kraus operators $\{U_a\}_{a=0}^N$. Because $U_\mathcal{U}$ has polynomial complexity, under the pesudorandomness assumption a complete set of $\delta$-approximate ghost logical operators can be constructed, with $\delta$ given by equation~\eqref{eq:delta-from-pseudo}. In other words, for each unitary $U_a$ that the observer might apply, $U_a$ commutes with all ghost logical operators up to an exponentially small error. Hence no matter which low-complexity unitary the observer applies, the encoded black hole interior is hardly affected at all.

This conclusion is summarized by the following theorem:

\begin{theorem}\label{thm:controlled_unitary} 
Suppose that the decoupling bound~\eqref{eq:decoupling_main} holds. Let $V:\mathcal{H}_{\tilde{B}}\rightarrow \mathcal{H}_{EH}$ denote the black hole code embedding. Let $\mathcal{U} = \{U_a\}_{a=1}^N$ denote an arbitrary set of $N = \mathrm{poly}(|H|)$ unitaries acting on the early radiation $E$, where each unitary has complexity $\mathrm{poly}(|H|)$. 
%%%
Then there exists a complete set of logical operators $\mathcal{L} \subseteq \mathcal{B}(\mathcal{H}_{EH})$ for the black hole code such that for all $T \in \mathcal{L}$, and all $U_a \in \mathcal{U}$, we have
\begin{align}
    \|[U_a, T]V\| \le 2\delta' \|\tilde{T}\|,\label{eq:theorem810a}
    \end{align} 
    and
    \begin{align}
    \|TV - V\tilde{T}\| \le \delta' \|\tilde{T}\|, \label{eq:theorem810b}
\end{align}
where $\Tilde T$ is the operator on $\tilde B$ corresponding to $T$, and
\begin{align}\label{eq:delta-prime}
\delta' = 8 \cdot 6^{1/4}~ 2^{-\alpha |H|/4}(N+1)^{3/4}
\end{align}
if $\tilde B$ is a single qubit $(|B|=1)$.
\end{theorem}

\begin{proof}
Let us model the observer $O$ on the Hilbert space $\mathcal{H}_O = \mathbb{C}^{N+1}$, so that $2^{|O|}=N+1$. In Lemma~\ref{lem:controlU}, we show that the conditional unitary $U_\mathcal{U}$ defined in equation~\eqref{eq:U_U-controlled} can be approximated to exponential accuracy with a circuit of size $\mathrm{poly}(|H|)$ if each $U_a$ has complexity $\mathrm{poly}(|H|)$ and $N=\mathrm{poly}(|H|)$; therefore, under the pseudorandomness assumption, $U_\mathcal{U}$ is $\epsilon$-correctable with 
\begin{equation}
    \epsilon = \sqrt{\frac{3}{2}}\cdot 2^{-(\alpha |H| - |OB|)/2},
\end{equation}
and hence there exists a complete set of $\delta$-approximate ghost logical operators for $U_\mathcal{U}$ with
\begin{equation}
    \delta = 2\cdot 3^{1/4} ~2^{-(\alpha |H| - |O| - 9|B|)/4},
\end{equation}
or
\begin{equation}
    \delta = 8 \cdot 6^{1/4}~ 2^{-\alpha |H|/4}2^{|O|/4} = 8 \cdot 6^{1/4}~ 2^{-\alpha |H|/4}(N+1)^{1/4}
\end{equation}
if $|B|=1$. The Kraus operators for the channel $\mathcal{E}_\mathcal{U}$ in equation~\eqref{eq:channel-E-U} are $\{U_a/\sqrt{N{+}1}\}$; hence 
\begin{align}
\|TU_{a}V - U_{a}V\tilde{T}\| &\le \|\tilde{T}\|\delta\sqrt{N+1}=\|\tilde{T}\|\delta'.
\end{align}
This, together with Lemma~\ref{lem:ghostcommute}, gives the desired results equation~\eqref{eq:theorem810a} and equation~\eqref{eq:theorem810b}.

\end{proof}

Note that, although $\delta'$ in equation~(\ref{eq:delta-prime}) could be exponentially small even for superpolynomial $N$, we required $N = \textrm{poly}(|H|)$ because only in that case have we shown that the conditional unitary $U_\mathcal{U}$ has complexity $\textrm{poly}(|H|)$; we needed this property for the pseudorandomness assumption to imply that the observer is unable to distinguish the state of $EB$ from a maximally mixed state. 

We have inferred the existence of ghost logical operators which act on $EH$. It should also be possible to realize a non-trivial logical operator as a physical operator acting on $E$ alone, but only if that operator is computationally complex to construct. For instance, suppose that  
$W: \mathcal{H}_E \rightarrow \mathcal{H}_E$ is a unitary logical operator that can be accurately approximated by a quantum circuit of polynomial size.
%and $P$ is a product of single-qubit unitary operators acting on $E$. 
Then there exists a ghost logical operator $T$ that fails to commute with $W$ acting on the code space. Since $W$ has polynomial complexity, this contradicts Theorem \ref{thm:controlled_unitary}, and we conclude that no such $W$ can exist. 
%\ET{I don't think I follow the paragraph above. What's the purpose of splitting $U$ into $W$ and $P$?} 
%%%
This conclusion resonates with the observations of Bouland, Fefferman, and Vazirani, who argued that in the context of AdS/CFT duality, the dictionary relating the black hole exterior and interior should be computationally complex \cite{bouland2019,susskind2020horizons}.

On the other hand, if a quantum circuit is allowed to act on $H$ as well as $E$, and if $B$ has constant size, then any logical operator on the code space \emph{can} be realized efficiently.  We show this in Section \ref{sec:efficient-manipulation}.

\subsection{State dependence %is not enough 
\label{sec:comments}}
The (approximate) encoding isometry $V_\Psi: \mathcal{H}_{\tilde B}\rightarrow  \mathcal{H}_{EH}$ is determined by the pure quantum state $\Psi_{EHB}$ of the black hole $H$ and its emitted Hawking radiation $EB$. This state, and hence the encoding map, depends on the initial microstate of the infalling matter that collapsed to form the black hole. Therefore, the encoded interior of the black hole is said to be ``state dependent'' \cite{Papadodimas2013,Papadodimas2016}.

If black hole evaporation is unitary, and the event horizon is smooth because the black hole interior is encoded in the radiation, then state dependence of the encoding seems to be unavoidable; if the quantum information encoded in the initial state is preserved in the final state of the fully evaporated black hole, then how the late radiation emitted after the Page time is entangled with the early radiation emitted before the Page time must depend on that initial state.  This state dependence of the encoding is nonetheless troubling \cite{Harlow2014,Bousso2014,Marolf2016}. If the experiences of observers who fall through the event horizon are described by the logical operators of the code, and these logical operators are state dependent, then the observers inside the black hole seem to be capable of measuring nonlinear operators acting on $\Psi_{EHB}$, rather than linear operators as in the standard theory of quantum measurement. This ability to measure nonlinear properties of the state could lead to inconsistencies. We regard this as an unresolved issue, reflecting our incomplete understanding of how to describe measurements conducted behind black hole horizons. 

But the state-dependent encoding of the black hole interior is not sufficient by itself to solve the black hole firewall problem.\footnote{We thank Raphael Bousso for raising this issue.} %\IK{To be clear, Raphael did not claim that state-dependent encoding is insufficient to solve the firewall problem. On the contrary, he (incorrectly) surmised that state-dependence alone would resolve the firewall paradox (modulo the nonlinear quantum mechanics issue). The current phrasing may incorrectly represent Raphael's stance.} \jp{Well, it seems fair to me to say that he ``raised the issue.''} 
If the Hawking radiation is thoroughly scrambled, then we expect that the interior mode that purifies $B$ can be decoded by acting on $E$ alone after the Page time \cite{Hayden2007}, and therefore that the logical operators of the code may also be chosen to act on $E$ alone. If $T$ and $S$ are two noncommuting logical operators, where $S$ acts on $E$, then an observer (Bob) outside the black hole who applies $S$ could in principle alter the outcome of a measurement of $T$ performed by an observer (Alice) inside the black hole. Thus Bob can send an instantaneous message to Alice, in apparent violation of relativistic causality.

While we agree that such acausal signaling is possible in principle, we insist that the computational complexity of the task should be considered. Under the assumption that the Hawking radiation is pseudorandom, we have found that, in order to signal Alice, Bob must apply an operation to $E$ with complexity superpolynomial in $|H|$, if Alice's observables are the ghost logical operators we have constructed. Though possible, such an operation is infeasible in practice if the black hole $H$ is macroscopic; therefore the semiclassical causal structure of the spacetime is respected. 

\section{Inside the black hole}\label{sec:efficient-manipulation}
%\ET{One point of confusion I see in this section is that we use $V$ elsewhere for the black hole encoding, and here for a generic unitary on the code subspace. Should we address this?}
%\IK{Good point. I've changed the $V$s to $v$s. Let me know if you want to change that. I can change the figures accordingly. We still use $V$ in the oblivious amplitude amplification part, but that's probably not too confusing.}

Under our pseudorandomness assumption, an observer who acts on the early radiation system $E$ can affect the encoded interior of a black hole only by applying an operation with superpolynomial complexity. However, an agent who has access to the black hole system $H$ as well as $E$ can manipulate the interior efficiently. Here we construct an efficient unitary circuit $\bar U_{EH}$, acting on $EH$, that perturbs the encoded interior. Our construction makes use of an efficient quantum circuit that realizes the unitary $U_{\text{bh}}$ that describes the formation and partial evaporation of a black hole. This unitary creates a state in which $B$ is maximally entangled with a subspace of $EH$; if the circuit that implements $U_{\text{bh}}$ is efficient, then $\bar U_{EH}$ can be implemented efficiently as well. We will also see that an agent with access to $EH$ can efficiently decode the interior, distilling the code subspace of $EH$ to a small quantum memory.

Suppose we are given a unitary operator $U_{BEH}$ which realizes the map %\ET{We seem to denote $U_{BEH}$ as $U_{bh}$ in the introduction and in the figures below. We should try to be uniform. Edit: OK I see that we apply the case $U_{BEH} = U_{bh}$ below. But that's a little confusing as you're reading through I think.} 
\begin{equation}
    U_{BEH}|0\rangle_B |0\ldots 0\rangle_{EH} = \frac{1}{\sqrt{2}}(|0\rangle_B |\psi_0\rangle_{EH} + |1\rangle_B  |\psi_1\rangle_{EH}),
\end{equation}
where $B$ is a single qubit, and $EH$ is $n$ qubits. By applying the circuits that implement $U_{BEH}$ and  $U_{BEH}^{\dagger}$ on an ancillary register, together with some additional gates acting on the ancilla and $EH$, we will apply a unitary operator $\bar{U}_{EH}$ acting on $EH$ with the property that
\begin{equation}
    \begin{aligned}
    & \bar{U}_{EH}|\psi_0\rangle_{EH} = v_{00} |\psi_0\rangle_{EH} + v_{10}|\psi_1\rangle_{EH},\\
      & \bar{U}_{EH}|\psi_1\rangle_{EH} = v_{01} |\psi_0\rangle_{EH} + v_{11}|\psi_1\rangle_{EH},
    \end{aligned}
\end{equation}
where 
\begin{equation}
v = \left(
\begin{array}{cc}
  v_{00}   & v_{01} \\
 v_{10}    & v_{11}
\end{array}
\right)
\end{equation}
is some chosen $2\times 2$ unitary matrix. That is, $\bar{U}_{EH}$ applies an arbitrary ``logical'' unitary transformation on the two-dimensional ``code space'' spanned by $\{ |\psi_0\rangle_{EH}, |\psi_1\rangle_{EH}\}$.

The protocol is explained in two steps. First, we describe a probabilistic protocol which applies $\bar{U}_{EH}$ with success probability $\frac{1}{4}$. Next, using the probabilistic protocol, we build  a deterministic protocol which applies $\bar{U}_{EH}$ with probability $1$. The first protocol applies a unitary $U_{a_1a_2}$ and $U^\dagger_{a_1a_2}$ once each. Here, $U_{a_1a_2}$ is a unitary acting on an ancillary register $a=a_1a_2$ and can be realized by applying the circuit that implements $U_{BEH}$ on register $a_1$ and $a_2$. The register $B$ is replaced with $a_1$ and the register $EH$ is replaced with $a_2$. The second protocol applies $U_{a_1a_2}$ and $U_{a_1a_2}^\dagger$ three times each. We also use some additional gates, which are also efficient.

For the probabilistic protocol, consider the following sequence of operations:
\begin{enumerate}
    \item Initialize $a$ in the $|0\ldots 0\rangle$ state.
    \item Apply $U_{a_1a_2}$.
    \item Apply a swap between $a_2$ and $EH$.
    \item Apply the single-qubit operation $v^T$ to $a_1$.
    \item Apply $U_{a_1a_2}^{\dagger}$.
    \item Measure the $a$ register in the computational basis, and postselect on measuring the all-$0$ bit string.
\end{enumerate}
\noindent Applying this protocol for $U_{BEH}=U_{\text{bh}}$, and taking the initial state %, without loss of generality, 
to be $|\phi_\mathrm{matter}\rangle = |00\dots 0\rangle$, we obtain the circuit diagram in Figure~\ref{fig:bh_inside_probabilistic}.

\begin{figure}[ht]
    \centering
    \includegraphics[width=0.7\columnwidth]{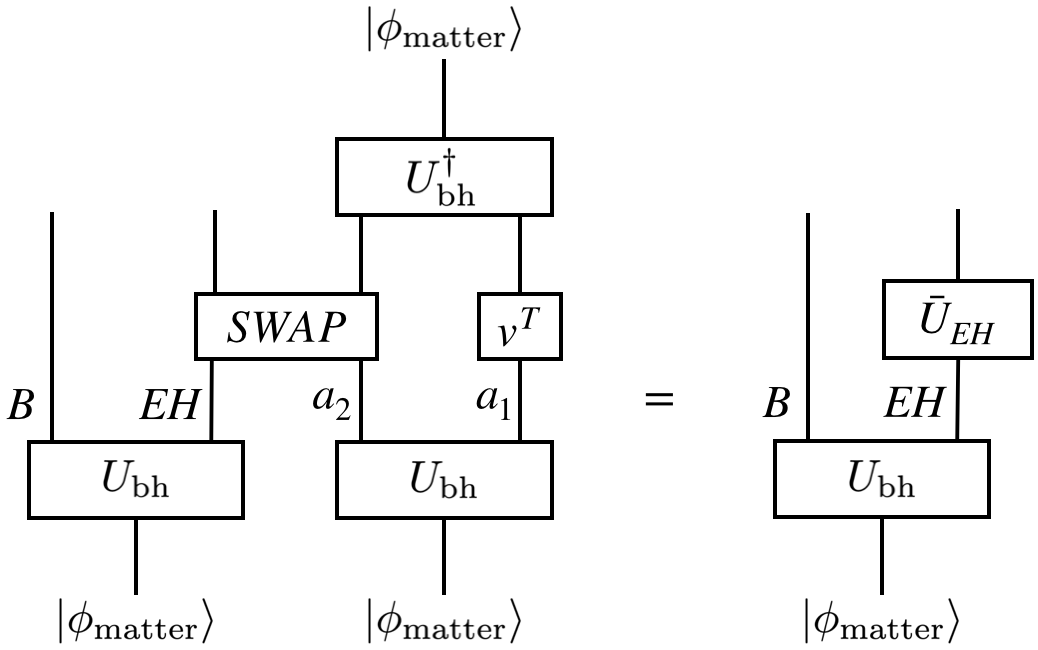}
    \caption{A probabilistic protocol which (with success probability $1/4$) applies an arbitrary unitary operator $v$ to $\tilde{B}$, the encoded interior partner of $B$. Here $\bar U_{EH}$ denotes $v$ acting on the code subspace of $EH$. \label{fig:bh_inside_probabilistic}}
\end{figure}

Let's analyze what happens when this protocol is executed. Suppose the state of $EH$ is an arbitrary pure quantum state $|\psi\rangle_{EH}$. After the second step, we have
\begin{equation}
    \frac{1}{\sqrt{2}}\left(|0\rangle_{a_1}|\psi_0\rangle_{a_2} + |1\rangle_{a_1}|\psi_1\rangle_{a_2}\right)|\psi\rangle_{EH}.
\end{equation}
Now expand $|\psi\rangle_{EH}$ in an orthonormal basis that includes both $|\psi_0\rangle$ and $|\psi_1\rangle$;  the remaining $2^{n}-2$ elements of the basis set are labeled $|\psi_i\rangle$ from $i=2$ to $i=2^{n}-1$, so that 
\begin{equation}
    |\psi\rangle = \sum_{i} \lambda_i |\psi_i\rangle.
\end{equation}
After the third step, we obtain
\begin{equation}
    \frac{1}{\sqrt{2}}(|0\rangle_{a_1}|\psi_0\rangle_{EH} + |1\rangle_{a_1}|\psi_1\rangle_{EH})|\psi\rangle_{a_2},
\end{equation}
which after the fourth step becomes
\begin{equation}
\begin{aligned}
   & \frac{1}{\sqrt{2}}((v_{00}|0\rangle_{a_1}+ v_{01}|1\rangle_{a_1})|\psi_0\rangle_{EH} + (v_{10}|0\rangle_{a_1}+ v_{11}|1\rangle_{a_1})|\psi_1\rangle_{EH})|\psi\rangle_{a_2}\\
    =  &\frac{1}{\sqrt{2}}(|0\rangle_{a_1}\left(v_{00}|\psi_0\rangle_{EH} + v_{10} |\psi_1\rangle_{EH}\right)+
    |1\rangle_{a_1}\left(v_{01}|\psi_0\rangle_{EH} + v_{11}|\psi_1\rangle_{EH}\right)) |\psi\rangle_{a_2}.
\end{aligned}
\end{equation}

Now we want to study what happens after we carry out the fifth and the sixth step. Instead of explicitly applying $U_{a_1a_2}^{\dagger}$, it is more convenient to think about an orthogonal measurement in a basis that includes $U_{a_1a_2}|0\ldots 0\rangle_a = \frac{1}{\sqrt{2}}(|0\rangle_{a_1}|\psi_0\rangle_{a_2} + |1\rangle_{a_1}|\psi_1\rangle_{a_2})$. After projecting onto this state, we obtain the (subnormalized) state
\begin{equation} 
\begin{aligned}
    &\frac{1}{2}(\lambda_0\left(v_{00}|\psi_0\rangle_{EH} + v_{10} |\psi_1\rangle_{EH}\right)+
    \lambda_1\left(v_{01}|\psi_0\rangle_{EH} + v_{11}|\psi_1\rangle_{EH}\right))\\
    =& \frac{1}{2}\left(\left(v_{00} \lambda_0 + v_{01} \lambda_1\right)|\psi_0\rangle_{EH} +\left(v_{10} \lambda_0 + v_{11} \lambda_1\right)|\psi_1\rangle_{EH} \right)),
    \end{aligned}
\end{equation}
which aside from the normalization factor of $1/2$ is equivalent to applying $v$ to the code vector $\lambda_0|\psi_0\rangle_{EH} + \lambda_1 |\psi_1\rangle_{EH}$. Hence, $\bar U_{EH}$ is applied with success probability $1/4$.

Now we explain how to upgrade this probabilistic operation to a unitary quantum circuit that applies $\bar{U}_{EH}$ deterministically. For this purpose, we use the %amazing 
\emph{oblivious amplitude amplification} technique introduced by Berry \emph{et al.}; see Lemma 3.6 of \cite{Berry2014}. For the reader's convenience, we restate this result.
\begin{lemma}
 (Oblivious amplitude amplification) Let $V'$ and $V$ be unitary matrices on $\mu + n$ qubits and $n$ qubits respectively, and let $\theta \in (0, \pi/2)$. Suppose that for any $n$-qubit state $|\psi\rangle$,
 \begin{equation}
     V'|0^{\mu}\rangle|\psi\rangle = \sin (\theta)|0^{\mu}\rangle V|\psi\rangle + \cos(\theta)|\Phi^{\perp}\rangle,
 \end{equation}
 where $(|0^{\mu}\rangle\langle 0^\mu|\otimes I)|\Phi^{\perp}\rangle=0.$ Let $R=2|0^{\mu}\rangle\langle 0^{\mu}|\otimes I-I$ and $S=-V'RV'^{\dagger}R^{\dagger}$. Then,
 \begin{equation}
     S^{\ell} V'|0^{\mu}\rangle|\psi\rangle = \sin((2\ell+1)\theta) |0^{\mu}\rangle V|\psi\rangle + \cos ((2\ell+1)\theta)|\Phi^{\perp}\rangle.
 \end{equation}
\end{lemma}

\noindent In our case, $V'$ is the unitary process described in the first five steps, $V$ is $\bar{U}_{EH}$, $|0^{\mu}\rangle$ is $|0\ldots 0\rangle_a$, and $\sin (\theta)= \frac{1}{2}$. Therefore, $\theta = \frac{\pi}{6}$, and we can choose $\ell=1$ to apply $V$ deterministically. For this choice of $\ell$, it suffices to apply $V'$ twice and its inverse once to achieve $V$. For each $V'$, we apply $U_{a_1a_2}$ and its inverse $U_{a_1a_2}^\dagger$ once each (as well as other simple unitary operations). In total, then, we can deterministically apply $\bar{U}_{EH}$ by using $U_{a_1a_2}$ three times and $U_{a_1a_2}^\dagger$ three times. In particular, the entire circuit is efficient if $U_{a_1a_2}$ is. Applying this protocol for $U_{BEH}=U_{\text{bh}}$, we obtain the circuit diagram in Figure~\ref{fig:bh_inside_deterministic}.

\begin{figure}[ht]
\centering
\includegraphics[width=0.875\columnwidth]{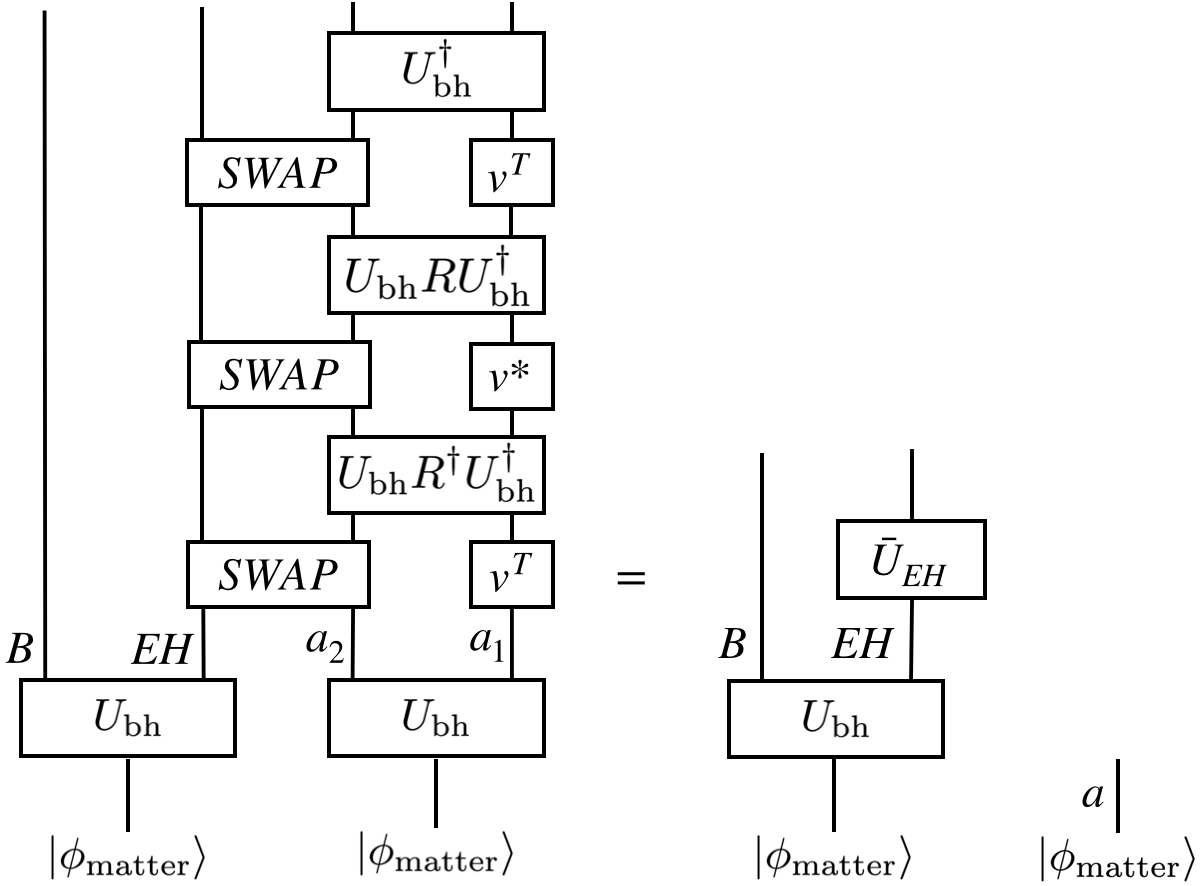}
\caption{A deterministic circuit which applies an arbitrary unitary operator $v$ to $\tilde{B}$, the encoded interior partner of $B$. Here $\bar U_{EH}$ denotes $v$ acting on the code subspace of $EH$. Note that the final $U_{\text{bh}}^\dagger$ and $v^T$ acting on the ancilla can be removed without changing how the circuit acts on the code space.   \label{fig:bh_inside_deterministic}}
\end{figure}

More generally, suppose that the register $B$ contains $|B|> 1$ qubits, so that the code subspace of $EH$ has dimension $2^{|B|}$. A probabilistic protocol for applying an arbitrary unitary transformation to the code space can be constructed that closely follows the construction for a single qubit, but now with success probability $2^{-2|B|}$. 
In particular, using the probabilistic protocol and oblivious amplitude amplification we can approximate any two-qubit gate ($|B|=2$) acting on the code space accurately and efficiently. 
%\IK{(if $|B|$ is small. We need to apply $U$ and its inverse roughly $2^{|B|}$) times}. \jp{For each two qubit gate, $|B|=2$, right?} \IK{Right, I missed the \textbf{two}-qubit part. I am happy with this sentence now.} 
From a universal set of such two-qubit gates, we can build a logical unitary circuit. Hence any low-complexity operation on the code space can be realized as a low-complexity quantum circuit acting on $EH$. %\ET{Just for my clarification, the protocol here would no longer be deterministic for $|B|>1$ right? It seems a rather happy accident that the coefficients for $|B|=1$ gave us a nice fraction of $\pi$ for the argument. We would be able to amplify to high success probabilities for $|B|>1$, but generally not $100\%$.} \IK{Correct. However, there is a trick to always make the success probability $100\%$. You can artificially reduce the success probability so that the $\theta$ is rational multiple of $\pi$.} \ET{Perhaps going a bit off-topic now, but I'm a bit curious. Could you outline how this artificial reduction is done? I'm not quite sure I see a good way of getting $100\%$ success.} \IK{You can add an additional ancilla, make it in a superposition state $\alpha|0\rangle + \sqrt{1-|\alpha|^2}|1\rangle$. You repeat the same circuit for all the unitary part, but declare success upon measuring $0$ on \emph{all} the qubits, including the newly introduced qubit. By tuning the value of $\alpha,$ you can decrease the success probability arbitrarily between the original success probability and $0$.}

%\jp{Comment on how this allows us to distill the code space into a small subsystem. Does it? I'm willing to consider $B$ to have constant size. If so, we can also comment that it is possible to efficiently decode $\tilde B$ in the case of the fully evaporated black hole. }

If we can perform logical gates on the code space, then we can also decode the logical state, distilling it to a small quantum memory in our possession. To be concrete, suppose the code space is two-dimensional. To decode, it suffices to prepare an ancilla qubit $b$ in an arbitrary state, and then perform a SWAP operation on $b$ and the encoded qubit. For this purpose we can use the quantum circuit identity shown in Figure~\ref{fig:swap}, where SWAP is constructed from controlled-$X$, controlled-$Z$, and Hadamard gates. The Hadamard gates act on $b$, and the C-$X$ and C-$Z$ gates act with $b$ as the control qubit and the code space as the target qubit.
\begin{figure}[ht]
    \centering
    \includegraphics[width=.95\columnwidth]{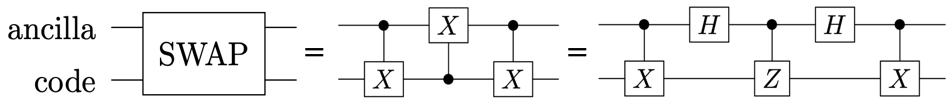}
    \caption{A two-qubit SWAP gate can be expressed in terms of Hadamard gates, controlled-$X$ gates, and a controlled-$Z$ gate. If there are efficient circuits for the $X$ and $Z$ gates acting on the code space, we may replace the gates in these circuits by gates controlled by an ancilla qubit, and use this identity to build a circuit that swaps the logical qubit in the code space with the ancilla qubit.}
   % \caption{The SWAP gate (left) is equivalent to three C-$X$ gates (middle). A C-$X$ gate can be decomposed into a C-$Z$ gate and two Hadamard gates. After applying this circuit identity, we obtain the figure on the right.}
    \label{fig:swap}
\end{figure}

%\begin{figure}[ht]
%\centering
%\includegraphics[width=1.0\columnwidth]{swap-circuit.png}
%\caption{A two-qubit SWAP gate can be expressed in terms of Hadamard gates, controlled-$X$ gates, and a controlled-$Z$ gate. If there are efficient circuits for the $X$ and $Z$ gates acting on the code space, we may replace the gates in these circuits by gates controlled by an ancilla qubit, and use this identity to build a circuit that swaps the logical qubit in the code space with the ancilla qubit. \jp{That's funny --- we both drew the same circuit!} \IK{I added a latex-ified version of your diagram(swap\_fancy.png). Feel free to use it if you want to!}}
%\label{fig:swap-circuit}
%\end{figure}

Suppose we have a circuit acting on $EH$ that applies $X$ to the code. We can replace each gate in that circuit by a controlled gate, with $b$ as the control qubit. The resulting circuit applies C-$X$ with $b$ as the control qubit, and if the circuit for $X$ is efficient, so is the circuit for C-$X$. Likewise, we can turn  an efficient circuit acting on $EH$ that applies $Z$ to the code into an efficient circuit for C-$Z$. Using the circuit identity, we obtain an efficient circuit acting on $b$ and $EH$ that swaps the encoded information into $b$. 
%%%
Using this realization of the SWAP gate, the entangled state of $B$ with the encoded interior mode $\tilde B$ becomes an entangled state of $B$ and $b$.

Note that this construction of logical gates, and of the decoding circuit, can also be applied to the fully evaporated black hole. After the evaporation is complete, $H$ is gone, but any Hawking radiation qubit $B$ is entangled with a highly scrambled subspace of $E$, a large system composed of all the other radiation quanta. Because the evolution of the initial infalling matter to the final outgoing Hawking radiation is described by an efficient unitary transformation $U$, we have seen how $U_{a_1a_2}$ and $U_{a_1a_2}^\dagger$ can be used three times each to construct either $X$ or $Z$ acting on the encoded qubit. By replacing the gates in $U$ by controlled gates, we can construct the SWAP operator, and hence distill the encoded qubit which is entangled with $B$ into a small quantum memory efficiently.

%\jp{Moved the rest to Appendix D, for reference, though I'm not inclined to include it in the final draft.}
%\IK{I changed the notation to make it a bit more clear that we are not touching $B$. I may decide to add a circuit diagram later.}
%\jp{Why the comma in $U_{B,EH}$? Since we don't comma separate $B$ from $EH$ elsewhere, it seems to mean something, but the reader is left wondering what it means.}
%\IK{It was mainly because $BEH$ consists of three letters whereas $a_1a_2$ consists of two letters. But I suppose the relation between the two has been clarified in the text so maybe this is unnecessary. I have removed the commas. I also reverted footnote about Bousso back to the previous version.}
%\IK{I added two figures to explain what happens if the protocol is applied to our setup. An interesting note: In Figure~\ref{fig:bh_inside_deterministic}, the last $U_{\text{bh}}^{\dagger}$ and $V^T$ is unnecessary for the purpose of applying $V$ on $B$. Applying $V_{a_1}^*U_{\text{bh}}$ on both sides, we obtain a slightly shortened circuit that applies $V$ on $B$ while creating another ``twin" black hole and applying $V^*$ on an outgoing Hawking quanta of the twin black hole.}

\section{Conclusion \label{sec:conclusion}}
From a purely quantum information perspective, the results in this paper apply to a tripartite pure state $\Psi_{EHB}$, where $|E| \gg |H| \gg |B|$. Our central assumption, from which all else follows, is that the marginal state $\rho_{EB}$ is \textit{pseudorandom} --- \emph{i.e.}, cannot be distinguished from a maximally mixed state with a bias better than $2^{-\alpha |H|}$ by any quantum computation with complexity polynomial in $|H|$. From this assumption, it follows that if a unitary transformation with complexity $\mathrm{poly}(|H|)$ acts on $E$ and an observer $O$, then $B$ and $O$ \emph{decouple} in the resulting state $\Psi'_{OEHB}$, \emph{i.e.}, $\rho'_{OB}\approx \rho'_O\otimes \rho'_B$ up to an error $O\left(2^{-\alpha |H| +|O|+|B|}\right)$. Here $\alpha = O(1)$ is a positive constant. 

The state $\Psi_{EHB}$ also defines an encoding map $V_\Psi: \mathcal{H}_{\tilde B}\rightarrow  \mathcal{H}_{EH}$, whose image is a subspace of $EH$ that is nearly maximally entangled with $B$. From the decoupling condition we can infer that the encoded system $\tilde B$ is hard to decode if $\alpha |H| {-} |O|{-}|B| \gg 1$; the observer can distill $\tilde B$ to a small subsystem only by performing an operation with complexity superpolynomial in $|H|$. Furthermore, if the observer $O$ performs any quantum computation on $E$ with complexity $\mathrm{poly}(|H|)$, there is a recovery operator $\mathcal{R}$ acting on $EH$ that corrects this ``error'' with fidelity $F=1-\epsilon$ where $\epsilon = O\left(2^{-\alpha|H|/2 +|O|/2 +|B|/2}\right)$. Here the size $|O|$ of the observer $O$ may be interpreted as the number of qubits in $O$'s quantum memory, or equivalently as the Kraus rank of the quantum channel applied to $E$ by $O$.

The existence of such a recovery operator $\mathcal{R}$ has a further implication. We can construct a complete set of \emph{ghost logical operators} for $\tilde B$ acting on $EH$; if $O$ applies a quantum channel to $E$ with complexity $\mathrm{poly}(|H|)$, then these ghost operators commute with all the Kraus operators of the channel, up to an error $O\left(2^{-\alpha|H|/4 +|O|/4+ 9|B|/4}\right)$. Thus the ghost operators fail to detect the action of any observer who performs an operation on $E$ with complexity $\mathrm{poly}(|H|)$.

For quantum informationists, these results may be viewed as a contribution to the theory of operator algebra quantum error-correcction in the approximate setting. 
What can be said about their potential physical consequences? 

The existence of pseudorandom quantum states that can be prepared by quantum circuits with depth $O(\mathrm{polylog} |H|)$ follows from standard assumptions used in post-quantum cryptography \cite{Ji2018}. Because black holes are efficient scramblers of quantum information, it is plausible that a pseudorandom state can be efficiently prepared by an evaporating black hole, where the black hole microstates of $H$ provide the concealed ``key'' of the state.  A similar remark may apply to other strongly chaotic systems as well. In the setting of black holes, our conclusion about the hardness of decoding the Hawking radiation of an old black hole builds on the work of Harlow and Hayden \cite{Harlow2013} by highlighting the role of pseudorandomness, and by clarifying that the the condition $|H|\gg 1$ already ensures that decoding is hard - even if $|H|$ is much smaller than $|E|$. 

We require in addition that $|H|$ is sufficiently large compared to the size $|O|$ of the observer's quantum memory, though we may allow the observer to wield a large probe system $P$ which interacts with $E$,  where $|P|\gg |H|, |O|$. In that case, the system $\tilde B$ becomes encoded in $PEH$ rather than $EH$. However, the conclusion that $\tilde B$ cannot be efficiently distilled to a subsystem of size $|O|$  still applies for $|O|\ll \alpha |H|$, if $B$ has constant size. Therefore, no agent with reasonable computational power can decode $\tilde B$ and carry it into the black hole without incurring a substantial backreaction on the black hole geometry. 

%We have studied a toy model of a black hole to understand the nature of spacetime locality between the interior and the exterior of a black hole. We found that there is a possible definition of the interior operators that approximately respects causality, unitarity, and effective field theory. By assuming that the early radiation is pseudorandom, we showed the existence of such operators that respects causality up to an exponentially small error in the size of the black hole. The main surprise is that these operators commute with a set of operators that the external observer can apply on the radiation, even though the interior operators act nontrivially on the radiation. The commutation relation does not hold as an operator equation. Instead, the commutators vanish on an appropriate code subspace up to an exponentially small error. 

To evade the black hole firewall problem, it has been proposed that (part of) the interior of an old black hole past its Page time is actually encoded in the radiation system $E$ emitted long ago. This encoding is profoundly nonlocal and therefore potentially problematic --- why can't an agent far outside the black hole who acts on $E$ send instantaneous messages to observers who are inside, or even create a firewall at the event horizon? Our view is that computational complexity should be invoked to reconcile the nonlocal encoding of the interior with the semiclassical causal structure of the black hole geometry.

The finding that ghost logical operators can be constructed when the Hawking radiation is pseudorandom fits neatly with this viewpoint. We propose that the observables accessible to observers inside the black hole are described by these ghost logical operators, though admittedly we have no compelling general basis for this claim other than to address the firewall problem. If we accept the claim, it follows that an agent outside the black hole can create detectable excitations behind the horizon only by performing operations of superpolynomial complexity. This conclusion, though based on different arguments, meshes with the proposal by Bouland \emph{et al.}~\cite{bouland2019,susskind2020horizons}, that the dictionary relating the black hole interior to its exterior in the context of AdS/CFT duality must be computationally complex. 

In our discussion, the encoding map relating the interior system $\tilde B$ to the early radiation $E$ and remaining black hole $H$ depends on the microstate of the initial collapsing body from which the black hole formed. It can also depend on how the observer interacts with the radiation \cite{Yoshida2019a}. Specifically, an observer who controls a large probe system $P$ that comes into contact with $E$ is empowered to alter the encoding substantially. But modifying the code does not help the observer to decode the radiation or to send a message to the interior --- achieving either task by acting on $E$ requires an operation with complexity superpolynomial in $|H|$.

Once an observer falls through the event horizon, the interior of the black hole should become accessible. From our point of view, this interior observer can interact not just with $E$ but also with $H$, which makes the task of manipulating the interior far easier. Indeed, %we showed in Section \ref{sec:efficient-manipulation} that 
for a code space of constant dimension, arbitrary unitary transformations on the code space can be realized by quantum circuits acting on $EH$ with complexity $\mathrm{poly}(|EH|)$. 

It is a familiar notion that, even in a theory of quantum gravity, local effective field theory on a curved background can provide an excellent approximation when the spacetime curvature is sufficiently small and the energy is sufficiently low. The story of ghost logical operators indicates that further constraints may need to be satisfied for physics to be approximately local: operations must have sufficiently low complexity and Kraus rank. Operations with high complexity and/or high rank can tear spacetime apart.

Our description of the robust encoded interior of an old black hole highlights the effectiveness of quantum error-correction against a nonstandard noise model. In the setting of fault-tolerant quantum computing, we normally seek an encoding that can protect against weakly correlated errors with a relatively low error rate. Here, though, the ``noise'' inflicted by our observer $O$ on the early radiation system $E$ is strong and chosen adversarially. As long as this noise process has computational complexity $\mathrm{poly}(|H|)$ and sufficiently small Kraus rank, the encoded system $\tilde B$ can be restored with high fidelity, and the ghost logical operators are barely affected at all. What makes this protection possible is that, although $E$ is treated very harshly, the ``key space'' $H$ is assumed to be noiseless. Perhaps related ideas can be exploited to protect quantum information in other physically relevant settings.

\begin{acknowledgments}
We thank Adam Bouland, Raphael Bousso, Anne Broadbent, Juan Maldacena, and Geoff Penington for valuable discussions. IK's work was supported by the Simons Foundation It from Qubit Collaboration and by the Australian Research Council via the Centre of Excellence in Engineered Quantum Systems (EQUS) project number CE170100009. Part of this work was done during IK's visit to the Galileo Galilei Institute during the ``Entanglement in Quantum Systems" workshop. ET and JP acknowledge funding provided by the Institute for Quantum Information and Matter, an NSF Physics Frontiers Center (NSF Grant {PHY}-{1733907}), the Simons Foundation It from Qubit Collaboration, the DOE QuantISED program ({DE}-{SC0018407}), and the Air Force Office of Scientific Research ({FA9550}-{19}-{1}-{0360}). ET acknowledges the support of the Natural Sciences and Engineering Research Council of Canada (NSERC).
\end{acknowledgments}

\bibliography{bib}
\bibliographystyle{apsrev}

\appendix
\section{Approximate Embedding\label{sec:embedding}}

\begin{lemma}\label{lem:embedding} 
Let $|\Psi\rangle_{EBH}$ be a pseudo-random state (see Definition~\ref{def:pseudorandom}), where $B$ is a single qubit. Then the operator $V_{\Psi}$ defined by equation~\eqref{eq:embedding} is an approximate embedding, i.e., there exists an embedding $V$ such that
\begin{align}
    \|V - V_{\Psi}\| \le 2\cdot 2^{-\alpha|H|}.
\end{align}
\end{lemma}
\begin{proof}
Let $\rho_{EBH} = |\Psi\rangle\langle\Psi|_{EBH}$. Applying the decoupling inequality~\eqref{eq:decoupling_main} without the presence of an observer (i.e., taking $|O|=0$), we see that $\rho_B$ is nearly maximally mixed, i.e.,
\begin{align}
    \|\rho_B - \frac{1}{2}I_B\|_1 \le 2^{-\alpha|H|}.
\end{align}
Equivalently, this implies that
\begin{align}
    \|V^\dagger_{\Psi}V_{\Psi} - I_{\tilde{B}}\|_1 \le 2\cdot 2^{-\alpha|H|}:=\epsilon.\label{eq:proj}
\end{align}
Now, let $U\Sigma_{\Psi} W^\dagger = V_{\Psi}$ be the singular value decomposition for $V_{\Psi}$. Let us denote the singular values of $V_{\Psi}$ as $\{\sigma_k\}$. Then~\eqref{eq:proj} implies that we have $|\sigma_k^2 - 1| \le \epsilon$. Since $|\sigma_k + 1| \ge 1$ (the singular values are nonnegative real numbers), we then also have
\begin{align}
|\sigma_k - 1| \le \epsilon\cdot |\sigma_k+1|^{-1} \le \epsilon.   
\end{align}
Now, let $\Sigma$ denote the matrix with the same shape as $\Sigma_{\Psi}$ whose diagonal values are all equal to $1$. Define $V = U\Sigma W^\dagger$, and note that $V$ is an isometric embedding since $V^\dagger V=I_{\tilde{B}}$. Finally, we have
\begin{align}
    \|V-V_{\Psi}\| &= \|U(\Sigma - \Sigma_{\Psi})W^\dagger\|\\
    &\le \|U\|\cdot\|W^\dagger\|\cdot \|\Sigma - \Sigma_{\Psi}\|\\
    &\le \epsilon,
\end{align}
where the last inequality follows since all singular values of $\Sigma - \Sigma_{\Psi}$ are bounded above by $\epsilon$ by construction.
\end{proof}

\section{Complete Set of Ghost Operators Implies Correctability}\label{sec:ghost_converse}
%\jp{I have made many changes in Appendix B. Please check!}

In this Appendix, we prove a converse to Theorem~\ref{thm:ghostapx}, showing that if a quantum error-correcting code $\mathcal{C}$ has a complete set of $\delta$-approximate ghost logical operators for a channel %\IK{In the original version, I believe $E_a$ and $E_a^{\dagger}$ were switched here.} \ET{Yes, I think that was a typo on my part. It looks good now.}
\begin{align}
\mathcal{E}(\rho) = \sum_{a=1}^r E_a \rho E_a^{\dagger}
\end{align}
with a set of $r$ Kraus operators $K=\{E_a\}$, then the channel $\mathcal{E}_\mathcal{I}$ with Kraus operators $K\cup \{I\}$ is $\epsilon$-correctable for $\mathcal{C}$, where $\epsilon =O(|K| \sqrt{(\dim \mathcal{C})\,\delta})$.

For this purpose, we will use the approximate version of the Knill-Laflamme error-correction conditions studied by B\'eny and Oreshkov \cite{BenyOreshkov10}; these may be expressed in the form
\begin{equation}
    PE_a^\dagger E_b P = \lambda_{ab} P + B_{ab},
\end{equation}
where $P$ is the projector on the code space $\mathcal{C}$, $\lambda_{ab}$ is a density matrix (a non-negative Hermitian operator with trace 1), and for each $a$ and $b$, $B_{ab}$ is an operator mapping $\mathcal{C}$ to $\mathcal{C}$. For $B_{ab}= 0$, these are the usual Knill-Laflamme conditions for exact correctability \cite{knill1997theory}. If $B_{ab}$ is small, the Knill-Laflamme conditions are approximately satisfied, and a recovery operator $\mathcal{R}$ exists that corrects the channel $\mathcal{E}$ acting on the code space, up to a small error $\epsilon$ as in equation \eqref{eq:bures-epsilon}.

A relation between $B_{ab}$ and $\epsilon$ was derived in \cite{BenyOreshkov10}. We define maps $\Lambda$ and $\mathcal{B}$ by
\begin{align}\label{eq:Lambda-B-define}
    \quad \Lambda(\rho) = \sum_{a,b=1}^r\lambda_{ab}\mathrm{Tr}(\rho)|a\rangle\langle b|,\quad\text{and}\quad \mathcal{B}(\rho)=\sum_{a,b=1}^r \mathrm{Tr}(\rho B_{ab})|a\rangle\langle b|,
\end{align}
    respectively. Consider the Bures distance $\mathfrak{B}(\Lambda + \mathcal{B},\Lambda)$ defined as in equation \eqref{eq:bures-define}, with the maximum taken over all code states $\rho$. 
Then the noise channel $\mathcal{E}$ is $\epsilon$-correctable for the code $\mathcal{C}$ if and only if  $\mathfrak{B}(\Lambda + \mathcal{B},\Lambda)\le \epsilon$ \cite{BenyOreshkov10}.

We may estimate this Bures distance as in equation \eqref{eq:bures_inequality}, finding 
\begin{align}\label{eq:bound-B-B}
  2 \mathfrak{B}^2(\Lambda +\mathcal{B}, \Lambda ) \le \max_\rho\left\|(\mathcal{B} \otimes \mathcal{I})(|\psi\rangle\langle\psi|)\right\|_1,
\end{align}
where $|\psi\rangle$ is a purification of the logical density operator $\rho$. Using equation \eqref{eq:Lambda-B-define}, we obtain
\begin{equation}\label{eq:B-bound-r}
\begin{aligned}
\left\|(\mathcal{B} \otimes \mathcal{I})(|\psi\rangle\langle\psi|)\right\|_1 &=
   \left\|\sum_{a,b=1}^r \langle \psi|B_{ab}|\psi\rangle |a\rangle\langle b|\right\|_1\\
    &\le r^2\max_{a,b}\left|\langle \psi|B_{ab}|\psi\rangle \right|\\
    &\le r^2\max_{a,b}\|B_{ab}\|\\
    &\le r^2(\dim\mathcal{C})\,\max_{a,b}\|B_{ab}\|_{\mathrm{max}}.
    \end{aligned}
\end{equation}
Here the entry-wise max norm of $\| A\|_{\mathrm{max}}$ of a matrix $A$ is defined as the largest (in absolute value) entry of the matrix in the computational basis; \emph{i.e.},
\begin{align}
    \|A\|_{\mathrm{max}} = \max_{i,j}\left|\langle i|A|j\rangle\right|,
\end{align}
and we used an inequality relating the operator and max norms,
\begin{align}
\|B_{ab}\| \le (\dim\mathcal{C})\|B_{ab}\|_{\mathrm{max}}.
\end{align}

We can now prove: 
\begin{lemma}\label{lem:knill-laflamme}
%\IK{Same problem here. I confirmed that the new $\mathcal{E}$ conforms the convention of Beny-Oreshkov paper.}
The channel 
\begin{align}
\mathcal{E}(\rho) = \sum_{a=1}^r E_a \rho E_a^{\dagger}
\end{align}
is $\epsilon$-correctable with respect to the code $\mathcal{C}$,
with
\begin{equation}
    \epsilon = r \sqrt{\frac{1}{2} (\dim\mathcal{C})\,\delta},
\end{equation}
if there is a density operator $\lambda_{ab}$ and an  orthonormal basis $\{|i\rangle\}$ for the code space such that for all $i$ and $j$
%\begin{align}
%\delta = \frac{2}{r^2(\dim\mathcal{C})}\,\epsilon^2.
%\end{align}
\begin{align}\label{eq:entrywise-bound}
    \left|\langle i | E_a^\dagger E_b |j\rangle - \delta_{ij}\lambda_{ab}\right| \le \delta.
\end{align}
\end{lemma}
\begin{proof}
According to the B\'eny-Oreshkov criterion \cite{BenyOreshkov10}, the channel is $\epsilon$-correctable if $\mathfrak{B}^2(\Lambda + \mathcal{B},\Lambda)\le \epsilon^2$, and from equations \eqref{eq:bound-B-B} and \eqref{eq:B-bound-r} we have 
\begin{equation}
\mathfrak{B}^2(\Lambda + \mathcal{B},\Lambda)\le \frac{1}{2} r^2 (\dim\mathcal{C})\max_{a,b}\|B_{ab}\|_{\mathrm{max}} \le 
\frac{1}{2} r^2 (\dim\mathcal{C})\delta,
\end{equation}
where we derived the last inequality from the definition of the $\|\cdot\|_{\mathrm{max}}$ norm and equation \eqref{eq:entrywise-bound}. This proves the Lemma.
\end{proof}

We will use the following Lemma in the proof of Theorem \ref{thm:controlled_unitary}, as well as in the proof of Theorem \ref{thm:ghost_converse} below.
\begin{lemma}\label{lem:ghostcommute}
Let $\mathcal{C}$ be a code subspace with code projector $P$. Let $T$ be an $\delta$-approximate ghost operator for the channel $\mathcal{E}$ and the set of Kraus operators $K$. Then
\begin{align}\label{eq:[TE]P-bound}
    \|[T,E]P\| \le 2\delta \|\tilde{T}\|
\end{align}
for all $E \in K$.
\end{lemma}
\begin{proof}
Let $V$ be the code embedding. By definition of the ghost operator, we have
\begin{align}
    \|TEV - EV\tilde{T}\| \le \delta \|\tilde{T}\|,
\end{align}
for all $E \in K \cup \{I\}$. Taking $E=I$ gives 
\begin{align}\label{eq:TEV-bound}
    \|TV - V\tilde{T}\| \le \delta \|\tilde{T}\|.
\end{align}
Then we have
\begin{equation}\label{[TE]V-bound}
\begin{aligned}
    \|[T,E]V\|=\|TEV-ETV\| &= \|TEV-EV\tilde{T}+EV\tilde{T}-ETV\|\\
    &\le \|TEV-EV\tilde{T}\|+\|EV\tilde{T}-ETV\|\\
    &\le 2\delta \|\tilde{T}\| + \|E\|\cdot \|V\tilde{T}-TV\|\\
    &\le 2\delta \|\tilde{T}\|,
\end{aligned}
\end{equation}
where in the last line we used equation \eqref{eq:TEV-bound} and the fact that $\|E\|\le 1$ since $E^\dagger E \le I$ implies $\|E^\dagger E\| = \|E\|^2 \le 1$.
We can now obtain equation \eqref{eq:[TE]P-bound} if we can replace $V$ in equation \eqref{[TE]V-bound} by $P$. 
This is justified because, for any operator $A$, we have 
\begin{align}
    \|AP\|=\|AVV^\dagger\|\le \|AV\|\cdot \|V^\dagger\|\le \|AV\|,
\end{align}
where we have used $\|V^\dagger \| \le 1$ in the last line since $V$ is an isometric embedding.
\end{proof}

With these Lemmas in hand, we can proceed to prove:
\begin{theorem}\label{thm:ghost_converse}
Suppose that there exists a complete set of $\delta$-approximate ghost logical operators for the channel $\mathcal{E}$ and its set of Kraus operators $K=\{E_a\}$. Then $\mathcal{E}_\mathcal{I}$ is $\epsilon$-correctable for the code $\mathcal{C}$, where
\begin{align}
    \epsilon = (|K|+1)\sqrt{2(\dim\mathcal{C})\,\delta}.
\end{align}
\end{theorem}
\begin{proof}
Suppose that there exists a complete set of $\delta$-approximate ghost logical operators for $\mathcal{E}$ with respect to some Kraus decomposition $K=\{E_a\}_{a=1}^r$. We will also define $E_0 = I$. 

Given any two orthogonal code states $|\psi\rangle,|\phi\rangle \in \mathcal{C}$, let us define the operators $\tilde{T}_1$ and $\tilde{T}_2$ as in the proof of Theorem~\ref{thm:ghost_exact2}. Note that $\|\tilde{T}_1\| = \|\tilde{T}_2\| = 1$. Let $T_1$ and $T_2$ be their respective $\delta$-approximate ghost operators. Then, for $0\le a,b\le r$, we get
\begin{equation}
\begin{aligned}
    &\left|2\langle\psi|E_{a}^{\dagger}E_{b}|\phi\rangle\right|=\left|\langle\psi|E_{a}^{\dagger}E_{b}T_1|\phi\rangle-\langle\psi|T_1E_{a}^{\dagger}E_{b}|\phi\rangle\right|\\
    &= \left|\langle\psi|E_{a}^{\dagger}E_{b}T_1|\phi\rangle-\langle\psi|E_{a}^{\dagger}T_1E_{b}|\phi\rangle+\langle\psi|E_{a}^{\dagger}T_1E_{b}|\phi\rangle-\langle\psi|T_1E_{a}^{\dagger}E_{b}|\phi\rangle\right|\\
    &\le\left|\langle\psi|E_{a}^{\dagger}E_{b}T_1|\phi\rangle-\langle\psi|E_{a}^{\dagger}T_1E_{b}|\phi\rangle\right|+\left|\langle\psi|E_{a}^{\dagger}T_1E_{b}|\phi\rangle-\langle\psi|T_1E_{a}^{\dagger}E_{b}|\phi\rangle\right|\\
    &\le\|\left(E_{b}T_1-T_1E_{b}\right)|\phi\rangle\|\|E_{a}|\psi\rangle\|+\|E_{b}|\phi\rangle\|\|\left(T_1E_{a}-E_{a}T_1\right)|\psi\rangle\|\\
    &\le 4\delta,
\end{aligned}
\end{equation}
where in the second-to-last line we used the Schwarz inequality, and in the the last line we used  Lemma~\ref{lem:ghostcommute} and the fact that $\|E_a\|\le 1$. Therefore we have
\begin{align}\label{eq:KL-nondiag}
    \left|\langle\psi|E_{a}^{\dagger}E_{b}|\phi\rangle\right| \le 2\delta.
\end{align}
Repeating the same argument for $\tilde{T}_2$, we likewise get
\begin{align}
    \left|\langle\phi-\psi|E_{a}^{\dagger}E_{b}|\phi+\psi\rangle\right| \le 2\delta.
\end{align}
Then we have
%\begin{equation}
%\begin{aligned}
\begin{align}\label{eq:KL-diag}
    &\ \left|\langle\phi|E_{a}^{\dagger}E_{b}|\phi\rangle-\langle\psi|E_{a}^{\dagger}E_{b}|\psi\rangle\right|\nonumber\\
    = &\ \left|\langle\phi|E_{a}^{\dagger}E_{b}|\phi\rangle-\langle\psi|E_{a}^{\dagger}E_{b}|\psi\rangle+\langle\phi|E_{a}^{\dagger}E_{b}|\psi\rangle-\langle\psi|E_{a}^{\dagger}E_{b}|\phi\rangle-\langle\phi|E_{a}^{\dagger}E_{b}|\psi\rangle+\langle\psi|E_{a}^{\dagger}E_{b}|\phi\rangle\right|\nonumber\\
	\le &\ \left|\langle\phi|E_{a}^{\dagger}E_{b}|\phi\rangle-\langle\psi|E_{a}^{\dagger}E_{b}|\psi\rangle +\langle\phi|E_{a}^{\dagger}E_{b}|\psi\rangle-\langle\psi|E_{a}^{\dagger}E_{b}|\phi\rangle\right|+2\left|\langle\phi|E_{a}^{\dagger}E_{b}|\psi\rangle\right|\nonumber\\
	\le &\ 2\left|\langle\phi-\psi|E_{a}^{\dagger}E_{b}|\phi+\psi\rangle\right|+4\delta\nonumber\\
	\le &\ 8\delta.
\end{align}
%\end{aligned}
%\end{equation}

Now consider an orthonormal basis $\{|i\rangle, i = 0, 1, 2, \dots , \dim\mathcal{C} -1\}$, for the code space and define $\lambda_{ab} = \langle 0 |E_a^\dagger E_b|0\rangle$. Noting that in the equations  \eqref{eq:KL-nondiag} and \eqref{eq:KL-diag}, $|\phi\rangle$ and $|\psi\rangle$ can be any two elements of the orthonormal basis, we see that 
\begin{equation}
    |\langle i|E_a^\dagger E_b|j\rangle| \le 2 \delta
\end{equation}
for $i\ne j$, while
\begin{equation}
    |\langle i|E_a^\dagger E_b |i\rangle - \lambda_{ab} | \le 8 \delta.
\end{equation}
Thus we find that the approximate Knill-Laflamme conditions for $\mathcal{E}_{\mathcal{I}}$ are satisfied:
\begin{align}
    \frac{1}{2}\left|\langle i|E_a^\dagger E_b|j\rangle - \lambda_{ab}\delta_{ij}\right| \le 4\delta.
\end{align}
Note that the factor of $1/2$ comes from the normalization of the Kraus operators for $\mathcal{E}_{\mathcal{I}}$. From Lemma~\ref{lem:knill-laflamme}, this implies that $\mathcal{E}_\mathcal{I}$ is $\epsilon$-correctable for $\mathcal{C}$, where
\begin{align}
    \epsilon = (|K|+1)\sqrt{2(\dim\mathcal{C})\,\delta}.
\end{align}
\end{proof}

\section{Complexity of Controlled Unitary}
\begin{lemma}\label{lem:controlU}
Let $U$ be a unitary of circuit complexity $k$ with respect to some universal $2$-qubit gate set $\mathcal{G}$. Given an ancillary system of $n$ qubits, let $\Lambda_m(U)$ be the operator controlled on the state $|m\rangle$, where $0\le m < 2^n$, i.e.,
\begin{align}
    \Lambda_m(U)(|\ell\rangle\otimes |x\rangle) = |\ell\rangle\otimes U^{\delta_{\ell m}}|x\rangle.
\end{align}
Then given any $\epsilon >0$, the operator $\Lambda_m(U)$ can be implemented with $\epsilon$-precision with circuit complexity $O\left(4^nk\log^4(k/\epsilon)\right)$. 
\end{lemma}
\begin{proof}
Let $U = U_k\cdots U_1$ be a decomposition of $U$ into elements of $\mathcal{G}$. To implement $\Lambda_m(U)$ to $\epsilon$-precision, it suffices to implement $\Lambda_m(U_i)$ to $\epsilon/k$-precision for each $1\le i \le k$. Since each $U_i$ is a $2$-qubit gate, it follows that $\Lambda_m(U_i)$ is supported on at most $n+2$ qubits. By the Solovay-Kitaev theorem \cite{kitaev2002classical}, each $\Lambda_m(U_i)$ can be implemented to $\epsilon/k$-precision with $O(4^n\log^4(k/m))$ gates from $\mathcal{G}$. It follows that $U$ itself can be implemented to $\epsilon$-precision with $O(4^nk\log^4(k/\epsilon))$ gates.

\end{proof}

\noindent  The scaling with $n$ can be considerably improved using circuit constructions from \cite{barenco1995elementary}, but Lemma \ref{lem:controlU} will suffice for our purposes. 

\section{What if the radiation is not pseudorandom?}
\label{sec:no_pseudorandomness_consequences}
The central assumption of this paper is that the state of the Hawking radiation $EB$ emitted by a partially evaporated black hole is pseudorandom. Here we ask what happens if this assumption is broken in a particular way.

%correlation between an operator on $B$ and any operator on $E$ that can be efficiently measured by an external observer is exponentially small in $|H|.$

Suppose $B$ is a single qubit and the pure state of $EBH$ is 
\begin{equation}
|\Psi\rangle_{EBH} = \frac{1}{\sqrt{2}}(|0\rangle_B|\psi_0\rangle_{EH} + |1\rangle_B|\psi_1\rangle_{EH}).
\end{equation}
Consider a Hermitian operator $M_E$ acting on $E$ %with operator norm $\| M_E \|\leq 1$ 
such that $ M_E \otimes Z_B$ can be efficiently measured, where $Z_B$ is the Pauli-$Z$ operator acting on $B$. Suppose that
\begin{equation}
    \langle M_E \otimes Z_B \rangle_{\Psi} - \langle M_E\rangle_{\Psi} \langle Z_B \rangle_{\Psi} = c, \label{eq:connected_correlation}
\end{equation}
where the subscript $\Psi$ indicates that the expectation value is evaluated in the global state $|\Psi\rangle_{EBH}$, or equivalently in the marginal state $\rho_{EB}$. Note that $c=0$ if $\rho_{EB}$ is maximally mixed. Therefore, by definition, if $\rho_{EB}$ is pseudorandom, then $c$ must be exponentially small in $|H|$. It follows that if $c$ is a nonzero constant, independent of $|H|$, then $\rho_{EB}$ is not pseudorandom (though the converse is not necessarily true).

We will now show that, if $c\ne 0$ there cannot be a complete set of logical operators that commute with $M_E$ acting on the code space spanned by $\{|\psi_0\rangle_{EH}, |\psi_1\rangle_{EH}\}$ . Note that because the marginal state $\rho_B$ is maximally mixed, we have $\langle Z_B\rangle_\Psi = 0$, and therefore
\begin{equation}
    2c =  2\langle M_E \otimes Z_B \rangle_{\Psi} = \langle \psi_0| M_E |\psi_0\rangle - \langle \psi_1| M_E |\psi_1\rangle.
\end{equation}
Consider a Hermitian operator $X_L$ on $EH$ that acts on the code basis states $\{|\psi_0\rangle_{EH}, |\psi_1\rangle_{EH}\}$ like the Pauli-$X$ operator:
%%%
%We show that, if there exists an efficient measurement that establishes a correlation between $B$ and $E$, then there cannot be a complete set of ghost logical operators for the interior mode that commute with this measurement. Let  Without loss of generality, consider a hermitian operator $M_E$ acting on $E$ with operator norm $\| M_E \|\leq 1$. Suppose that with $|c|>0$, where $Z_B$ is the Pauli-$Z$ operator acting on $B$. Here, the $\Psi$ in the subscript means that the expectation value is evaluated with respect to the global state $|\Psi\rangle_{EBH}$. Because the reduced density matrix of $B$ is maximally mixed, the second term in equation~\eqref{eq:connected_correlation} vanishes. Moreover, the first term reduces to $(\langle \psi_0| M_E |\psi_0\rangle - \langle \psi_1| M_E |\psi_1\rangle)/2$. Therefore, we conclude that
%%%
\begin{equation}
\label{eq:x_l_def}
%\begin{aligned}
   X_L |\psi_0\rangle = |\psi_1\rangle,\quad  X_L |\psi_1\rangle = |\psi_0\rangle,
  %  &X_L |\psi_+\rangle = |\psi_+\rangle,\quad  X_L |\psi_-\rangle = -|\psi_-\rangle,
%\end{aligned}
\end{equation}
%where
%\begin{equation}
%    |\psi_{\pm}\rangle = \frac{1}{\sqrt{2}}\left(|\psi_0\rangle \pm|\psi_1\rangle\right)
%\end{equation}
%the following bound holds:
%\begin{equation} \label{eq:no_ghost}
%    \|[X_L, M_E]V_{B\to EH} \| \geq 2c, 
%\end{equation}
and notice that
\begin{equation}
  \langle \psi_1|  [X_L,M_E] |\psi_0\rangle = \langle \psi_0| M_E |\psi_0\rangle - \langle \psi_1| M_E |\psi_1\rangle = 2c \ne 0.
\end{equation}
This shows that the commutator $[X_L,M_E]$ is $O(1)$ acting on the code space. Thus no logical Pauli-$X$ operator commutes with $M_E$ acting on the code space, and in particular there can be no complete set of ghost logical operators commuting with $M_E$. 

%where $V_{B\to EH}= |\psi_0\rangle_{EH} \langle 0|_B + |\psi_1\rangle_{EH} \langle 1|_B$. If $c$ is a constant of order unity, equation~\eqref{eq:no_ghost} implies that there is no logical $X$ operator that commutes with $M_E$. In particular, there cannot be a complete set of ghost logical operators that commutes with $M_E$.

%As a first step, consider the following object:
%\begin{equation}
%    a := \langle \psi_{-}| [X_L, M_E] |\psi_+\rangle,
%\end{equation}
%where $|\psi_{\pm}\rangle :=\frac{1}{\sqrt{2}}(|0\rangle \pm |1\rangle)$. \ET{I assume you meant $|\psi_{\pm}\rangle = \frac{1}{\sqrt{2}}(|\psi_0\rangle \pm |\psi_1\rangle)$?} By explicit computation, one can observe that $a + a^* = -4c$. Therefore, $|a| \geq |\Re(a)| \geq 2|c|$. Note that
%\begin{equation}
%\begin{aligned}
%    |\langle \psi_{-}| [X_L, M_E] |\psi_+\rangle| &\leq \max_{|\phi\rangle_B} |\langle \psi_{-}| [X_L, M_E] V_{B\to EH} |\phi_B\rangle| \\
%    &\leq \max_{|\phi\rangle_B, |\phi'\rangle} |\langle \phi'| [X_L, M_E] V_{B\to EH} |\phi_B\rangle| \\
%    &\leq \|[X_L, M_E]V_{B\to EH} \|,
%\end{aligned}
%\end{equation}
%where the maximization is taken over all normalized quantum states in their respective Hilbert spaces. Plugging in the lower bound for $|a|$, we arrive at equation~\eqref{eq:no_ghost}.

For this argument we chose the operator acting on $B$ to be $Z_B$, but a similar argument works for any Hermitian operator acting on $B$. Suppose $N_B$ is a Hermitian operator acting on $B$ such that
\begin{equation}
\label{eq:general_correlation}
    \langle M_E \otimes N_B \rangle_{\Psi} - \langle M_E\rangle_{\Psi} \langle N_B \rangle_{\Psi} = c\ne 0.
\end{equation}
Since $N_B$ is Hermitian, we can diagonalize it in a certain basis, and we may assume without loss of generality that $N_B$ is traceless. 
%so that $\langle M_B\rangle = 0$ in the maximally mixed state. 
(If $N_B$ is not traceless, we may replace $N_B$ by $N_B' =N_B -\mathrm{Tr}\left(N_B\right) (I/2)$ without modifying equation~\eqref{eq:general_correlation}.) In the basis in which it is diagonal, then, $N_B$ is equal to $Z_B$ up to a nonzero multiplicative constant.

\end{document}